\documentclass[12pt]{article}

\usepackage[left=1in,
right=1in,
top=1in,
bottom=1.2in,
footskip=.3in]{geometry}
\usepackage{url}
\usepackage{amsmath}
\usepackage{amsfonts}
\usepackage{amssymb}
\usepackage{xcolor} % coloring text
\usepackage{setspace}
\usepackage{yfonts}
\allowdisplaybreaks %  allow page breaks inside equations
\usepackage{placeins}
\usepackage{graphicx,subcaption}
\usepackage{tabularx} % fit table
\graphicspath{{./plots/}} 

\usepackage{tikz}
\usetikzlibrary{graphs,graphs.standard,calc}
\usepackage{tikz-qtree}

\newtheorem{theorem}{Theorem}[section]

\newtheorem{condition}{Condition}[section]

\newtheorem{corollary}{Corollary}[section]

\newtheorem{example}{Example}[section]

\newtheorem{lemma}{Lemma}[section]

\newtheorem{proposition}{Proposition}[section]
\newtheorem{remark}{Remark}[section]

\newenvironment{proof}[1][Proof]{\noindent\textbf{#1.} }{\ \rule{0.5em}{0.5em}}

\DeclareMathOperator{\bS}{\pmb{S}}

\DeclareMathOperator{\bu}{\pmb{u}}
\DeclareMathOperator{\bv}{\pmb{v}}
\DeclareMathOperator{\bV}{\pmb{V}}
\DeclareMathOperator{\bx}{\pmb{x}}
\DeclareMathOperator{\bX}{\pmb{X}}
\DeclareMathOperator{\bY}{\pmb{Y}}
\DeclareMathOperator{\bZ}{\pmb{Z}}

\DeclareMathOperator{\bM}{\pmb{M}}
\DeclareMathOperator{\bN}{\pmb{N}}
\DeclareMathOperator{\bB}{\pmb{B}}
\DeclareMathOperator{\bC}{\pmb{C}}

\DeclareMathOperator{\bmu}{\pmb{\mu}}
\DeclareMathOperator{\bnu}{\pmb{\nu}}
\DeclareMathOperator{\b1}{\pmb{1}}
\DeclareMathOperator{\bzero}{\pmb{0}}

\title{Least squares for cardinal paired comparisons data}

\author{Rahul Singh$^{1}$, George Iliopoulos$^2$ and Ori Davidov$^3$}

\date{$^1${  Department of Mathematics, Indian Institute of Technology Delhi, Delhi 110016, India}\\
$^2$Department of Statistics and Insurance Science, School of Finance and Statistics, University of Piraeus, 80 Karaoli and Dimitriou str., 18534 Piraeus, Greece\\
$^3$Department of Statistics, University of Haifa, Mount Carmel, Haifa 3498838 Israel\\
{\small
E-mail: \texttt{wrahulsingh@gmail.com} (R Singh), 
\texttt{geh@unipi.gr} (G Iliopoulos), 
\texttt{davidov@stat.haifa.ac.il} (O Davidov)}}

\begin{document}

\maketitle

\begin{abstract}
{  Least square estimators for} graphical models for cardinal paired comparison data with and without covariates are rigorously analyzed. Novel, graph--based, necessary and sufficient conditions that guarantee strong consistency, asymptotic normality and the exponential convergence of the estimated ranks are emphasized. A complete theory for models with covariates is laid out. In particular, conditions under which covariates can be safely omitted from the model are provided. The methodology is employed in the analysis of both finite and infinite sets of ranked items where the case of large sparse comparison graphs is addressed. The proposed methods are explored by simulation and applied to the ranking of teams in the National Basketball Association (NBA).

\medskip
   
\end{abstract}
 
\section{Introduction}

There are many situations in which ranking a set of items is desired. Examples include the evaluation of political candidates, sports, information retrieval, and a variety of modern internet and e-commerce applications (e.g., Cremonesi et al., 2010, Buhlmann and Huber 1963, Govan 2008, Barrow et al. 2013, Xu et al. 2014, Cururingu 2015). The theory and methodology of ranking methods have been extensively studied from a variety of perspectives by researchers in the fields of mathematics (Langville and Mayer 2012), economics and social choice theory (Sen 1986, Slutzki and Volij 2005), machine learning (Ailon et al. 2008, Furnkranz and Hullermeier, 2010), psychology (Davis-Stober 2009, Regenwetter et al 2011) as well as many other disciplines. The focus of most of the existing literature, both old and new, has been procedural rather than inferential. For a statistical perspective on ranking methods, of a somewhat different flavor ours, see the books by Marden (1995) and Alvo and Yu (2014). 

A ranking of a set of items can be inferred from different types of data including scores (Balinski and Lariki, 2010) and ranked lists (Marden, 1995). In particular, paired comparison data (PCD) is obtained if all comparisons involve only two items (David, 1988). Suppose that there are $K$ items labelled $1,2,\ldots,K$ which we would like to rank. Let $Y_{ijk}$ denote the outcome of the $k^{th}$ comparison among items $i$ and $j$. The random variable (RV) $Y_{ijk}$ may be binary, ordinal or cardinal.   In this paper we consider cardinal, i.e., continuous PCD. We assume that the observations $Y_{ijk}$ for $1\leq i\neq j\leq K$ and $k=1,\ldots,n_{ij}$ satisfy
\begin{equation} \label{model.Y_ijl}
Y_{ijk} = \mu_{ij} + \epsilon_{ijk}, 
\end{equation}
where 
$$\mu_{ij}=\mathbb{E}(Y_{ijk})=\mu_i-\mu_j$$ 
and $\epsilon_{ijk}$ are independent zero mean RVs. The parameters $\mu_1,\ldots,\mu_K$ are refered to as the merits or scores. If we further assume that $\epsilon_{ijk}$ are IID $\mathcal{N}(0,\sigma^{2})$ RVs then \eqref{model.Y_ijl} is a specially parameterized homoscedastic normal linear model. Clearly this parametrization results in a global ranking referred to as a total linear ranking, cf. Oliveira et al. (2018) and the references therein, in which the merits quantify the degree of preference. We refer to such models as linear models on graphs or more compactly as graph--LMs. It is useful to note that we may view $Y_{ijk}$ as the outcome of the $k^{th}$ ``game'' between team $i$ and team $j$ where team $i$ scored $a$ points while team $j$ scored $b$ points. Then naturally $Y_{ijk}=a-b$ and $Y_{jik} = b-a$ and thus $Y_{ijk}=-Y_{jik}$. Each item is viewed as a vertex of a graph and the set of vertices is denoted by $\mathcal{V}=\{1,\ldots,K\}$. Let $n_{ij}=n_{ji}$ denote the number of comparisons between item $i$ and item $j$. We do not assume that $n_{ij}>0$ for all or even most pairs $(i,j)$. If $n_{ij} > 0$, then items $i$ and $j$ are connected by an edge denoted by $(i,j)$. Let $\mathcal{E}$ denote the set of edges. The structure $\mathcal{G}=(\mathcal{V},\mathcal{E})$ is called a graph (Gould, 2012). We note that some authors consider each individual comparison between items $i$ and $j$ as an edge. Thus $n_{ij} > 1$ would imply multiple edges connecting items $i$ and $j$. Such a structure is sometimes referred to as a multigraph. The distinction between a graph and a multigraph is inconsequential from our perspective. With each edge $(i,j)$ we associate a random sample $\mathcal{Y}_{ij} = (Y_{ij1},\ldots,Y_{ijn_{ij}})$ of $n_{ij}$ comparisons; the set of all samples is denoted by $\mathcal{Y}$. We call the pair 
\begin{align}
\label{graph:model:def}
    (\mathcal{G},\mathcal{Y})
\end{align}
a pairwise comparison graph (PCG). 

Least square estimators for the model (\ref{model.Y_ijl})  have a long history and have been studied in diverse fields (cf. Mosteller 1951, Kwiesielewicz 1996, Csato 2015). Surprisingly, despite the simplicity of model \eqref{model.Y_ijl}, the extensive associated literature, including the many papers on least square estimators, the statistical properties of the least squares estimator have not been properly studied. This paper fills this gap by developing a graph--based statistical theory necessary for conducting inference on the vector of merits. As shown model \eqref{model.Y_ijl}, as well as its extensions, exhibit unique features which complicate and enliven their analysis. In particular our contributions are: 

\begin{enumerate}

\item  The large sample properties of the least square estimator for finite comparison graphs are established using novel necessary and sufficient graph--based conditions. The relationship between the topology of the comparison graph and the quality of the resulting estimators is emphasized. A strong law for the merit parameter is established and the derived rankings are shown to converge at an exponential rate. 

\item Although some papers on graph--LMs with covariates exist, e.g., Harville (2003) who considered a single binary covariate, to the best of our knowledge this is the first paper to rigorously examine such models. In particular, conditions under which models with covariates can be analyzed are provided and analyzed. Situations in which covariates can be omitted from the model are emphasized. 

\item Finally, an analysis of graph-LMs when the number of items compared grows to infinity is provided. It is shown that up to a logarithmic factor connectivity is sufficient for ensuring that asymptotic normality and uniform consistency.

\end{enumerate}

This paper focuses on cardinal PCD satisfying \eqref{model.Y_ijl} with $\mu_{ij}=\mu_i-\mu_j$. More generally, the outcomes $Y_{ijk}$ may be binary or ordinal. In binary PCD $Y_{ijk}\in\{0,1\}$ and the probability that item $i$ is preferred over item $j$ is $\mathbb{P}(Y_{ijk}=1)=F(\lambda_i-\lambda_j)$ where $F$ is a distribution function, symmetric about $0$, i.e., $F(x)+F(-x)=1$ for all $x\in\mathbb{R}$, and $\lambda_1,\ldots,\lambda_K$ are the merit parameters, cf. Oliveira et al. (2018). Specifically, when $F$ is the standard logistic distribution we obtain the relation
\begin{equation} \label{Eq.BT.model}
\mathbb{P}(Y_{ijk}=1) = \frac{\exp(\lambda_i-\lambda_j)}{1+\exp(\lambda_i-\lambda_j)} = \frac{\mu_i}{\mu_i+\mu_j},   
\end{equation}
where $\mu_i=\exp(\lambda_i)$ for all $i=1,\ldots,K$. The $\mu's$ are also referred to as the merits. This is the well--known Bradley--Terry model (BT) (Bradley and Terry 1952, Hunter 2004, Cattelan 2012) whereas when $F$ is the standard normal DF we obtain the so--called Thurstone model (Thurstone 1927, Böckenholt, 2006). Both of these are instances of discrete choice models (Luce 1959, Bierlaire 1998) with a binary choice. The relationship between models \eqref{model.Y_ijl} and \eqref{Eq.BT.model} will be discussed in a comprehensive manner in Section 8 of this paper.
\medskip

The paper is organized as follows. The least square estimator is introduced and its elementary properties are explored in Section 2. Section 3 investigates the large sample properties of the LSE. Section 4 and 5 extend model \eqref{model.Y_ijl} in two directions. In Section 4 models with covariates are analyzed and in Section 5 the number of items compared is allowed to grow to infinity. Simulation results are presented in Section 6 and an illustrative example is discussed in Section 7. We conclude in Section 8 with a brief summary and discussion. In particular, we indicate how the results from Sections 2--7 extend to binary models such as \eqref{Eq.BT.model} and beyond. All proofs are collected in the Supplement \eqref{section:supplement} and R-scripts for reproducing all experiments are available at \url{https://github.com/rahulstats/GLM-I}.

\section{Least squares on graphs: estimation and elementary properties}

Model (\ref{model.Y_ijl}) is the simplest possible graphical linear model (graph--LM) in which a global linear ranking is assumed. Although this model has been widely studied in diverse fields (cf. Mosteller 1951, Kwiesielewicz 1996, Csato 2015) it is surprising that a variety of very basic statistical questions have not been adequatly addressed. 

Consider the objective function
\begin{equation} \label{Q(mu)}
Q( \bmu ) =\sum_{1\leq i<j\leq
K}\sum_{k=1}^{n_{ij}}(Y_{ijk}-\left( \mu _{i}-\mu _{j}\right) )^{2},
\end{equation}
which is nothing but a sum of squares over the PCG \eqref{graph:model:def} where $\bmu^\top=\left(\mu _{1},\ldots ,\mu_{K}\right)$ is the vector of merits. Clearly, if for all $(i,j)\in\mathcal{E}$ and  $k=1,\ldots ,n_{ij}$ we have $Y_{ijk}\sim \mathcal{N}( \mu _{i}-\mu_{j},\sigma ^{2})$ then \eqref{Q(mu)} is the kernel of the likelihood for a normal linear model. Further note that $Q(\bmu) =Q(\pmb{\mu}+c\pmb{1})$ for any $c\in \mathbb{R}$ where $\b1 = (1,\ldots,1)^\top\in\mathbb{R}^K$ and consequently $\bmu$ is not estimable unless a constraint is imposed. Therefore, we define the least squares estimator (LSE) by
\begin{equation}  \label{def.mu.hat}
\widehat{\bmu}=\arg \min \{Q( \bmu) :
\bv^{\top}\bmu=0\}  
\end{equation}
where $\bv$ is a preselected vector satisfying $\b1^{\top}\bv\neq 0$, i.e., $\bv$ is not a contrast, cf. Remark \ref{rem.on.v}. We require the following notations. For any $i\neq j$ let $S_{ij}=\sum_{k=1}^{n_{ij}}Y_{ijk}$ and $S_{i}=\sum_{j\neq i}S_{ij}.$ Further define,
\begin{equation*}
\bS=(S_{1},\ldots ,S_{K})^{\top},
\end{equation*}
and { let $\bN=(N_{ij})$ be the $K\times K$ Laplacian of the graph $\mathcal{G}$. Its elements are
\begin{equation*}
N_{ij}=\left\{ 
\begin{array}{ccc}
\sum_{j}n_{ij} & \text{if} & i=j \\ 
-n_{ij} & \text{if} & i\neq j
\end{array}
\right.,
\end{equation*}}
i.e., the $i^{th}$ diagonal element of the Laplacian is $n_{i}=\sum_{j}n_{ij}$ which is the degree of vertex $i$. The total number of paired comparisons is denoted by {  $n = \sum_{1 \leq i < j \leq K} n_{ij}$.} Laplacians play an important role in graph theory (Bapat 2010). It is easy to verify that they are symmetric positive semidefinite with rows and columns that sum to $0$. The vertices $i$ and $j$ are said to be connected if there exists a sequence of edges  $(i,v_1),(v_1,v_2),\ldots,(v_{l-1},v_l),(v_l,j)$. Such a sequence is called a path. If all pairs of vertices are connected then $\mathcal{G}$ is a connected graph. 
 
\begin{theorem} \label{Thm-UniqueE}
A unique solution to \eqref{def.mu.hat} exists if and only if $\mathcal{G}$ is connected in which case the LSE is given by
\begin{equation} \label{mu.hat}
\widehat{\bmu}=\bN^{+}\bS-\frac{\bv^\top\bN^{+}\bS}{\bv^\top\b1}\b1. 
\end{equation}
where $\bN^{+}$ is the Moore--Penrose inverse of $\bN$.
\end{theorem}

The proof of Theorem \ref{Thm-UniqueE} shows that if $\mathcal{G}$ is not connected then there are at least two distinct collections of items $I$ and $J$, say, that can not be directly compared, i.e., the difference $\mu _{i}-\mu _{j}$ is not estimable whenever $i\in I$ and $j\in J$. Surprisingly, formula \eqref{mu.hat} has not been previously derived; however it was mentioned by Ghosh and Davidov (2020) who studied alternatives to the least square estimator. Typically, it has been implicitly assumed that $\bv=\b1$, e.g., Csato (2015), in which case (\ref{mu.hat}) reduces to $\bN^{+}\bS$. It is also worth noting that most papers on the LSE end here, i.e., they make no attempt at studying the statistical properties of \eqref{mu.hat}. Further note that
\begin{equation*}
\left\Vert \widehat{\bmu}\right\Vert _{2}^{2}=\bS^\top
\bN^{+}\bN^{+}\bS-2\frac{\bv^\top
\bN^{+}\bS}{\bv^\top\b1}
\bS^\top\bN^{+}\b1+K\left(\frac{\bv^\top
\bN^{+}\bS}{\bv^\top\b1}\right)^{2}.
\end{equation*}
The second term on the right--hand--side above equals zero for all $\bv$; the third term is zero if and only if $\bv$ is in the kernel of $\bN^{+}.$ However, if $\mathcal{G}$ is connected then $\mathrm{\ker }(\bN^{+})=\mathrm{\ker }(\bN)={\rm span}(\pmb{1})$ (Bapat, 2010). Thus, $\bN^{+}\bS$ is the minimum norm solution to the unconstrained minimization problem given in (\ref{Q(mu)}). { Next, we investigate the relationship among estimators obtained under different constraints. 
\begin{proposition} \label{Prop.LSE.uv}
Let $\widehat{\bmu}(\bv)=\arg \min \{Q\left(\bmu\right) :\bv^\top\bmu=0\}$ and $\widehat{\bmu}(\pmb{u})=\arg \min \{Q\left( \pmb{
\mu }\right) :\pmb{u}^\top\bmu=0\}$ where $\bv^{\top}\b1\neq 0$ and $\pmb{u}^{\top}\b1\neq 0$. Then
\begin{align*}
\widehat{\bmu}(\pmb{u}) = \pmb{C}_{u}\pmb{C}_{v}^{+}\widehat{\bmu}(\pmb{v})   
\end{align*} where $\pmb{C}_{\pmb{u}}=\pmb{I}-\b1\pmb{u}^{\top}/(\pmb{u}^{\top}\b1)$ and $\pmb{C}_{v}^{+}=\pmb{I}-\pmb{v}\pmb{v}^{\top}/(\pmb{v}^{\top}\pmb{v})$ are idempotent matrices and $\pmb{C}_{v}^{+}$ is the Moore-Penrose inverse of $\pmb{C}_{\pmb{v}}$.
\end{proposition}

In particular, Proposition \ref{Prop.LSE.uv} shows that for any $\pmb{u}$ such that $\pmb{u}^{\top}\b1\neq 0$ we have $\widehat{\bmu}(\pmb{u}) = \pmb{C}_{\pmb{u}}(\bN^{+}\bS)$ where $\bN^{+}\bS$ is simply $\widehat{\bmu}(\pmb{1})$. Therefore if $\pmb{u}=\pmb{e}_1$, where $\pmb{e}_{1}$ is first standard basis vector for $\mathbb{R}^K$, or equivalently if we impose the constraint that $\mu _{1}=0$, then $\pmb{C}_{\pmb{e}_1}=\pmb{I}-\pmb{1}\pmb{e}_{1}^{\top}$ and the resulting LSE is $(0,(\bN_{2}^{+}-\bN_{1}^{+})\bS,\ldots ,(\bN_{K}^{+}-\bN_{1}^{+}) \bS) ^\top$ where $\bN_{j}^{+}$ is the $j^{th}$ row of $\bN^{+}$; if it is further assumed that $n_{ij}=m$ for all $i$ and $j$ then the latter simplifies to ${(K-1)}{(mK^2)^{-1}}(0,S_{2}-S_{1},\ldots ,S_{K}-S_{1})^{\top}$.} We emphasize that the choice of the constraint in (\ref{def.mu.hat}) is not of great importance since the estimated values of $\mu_{ij}$ are invariant with respect to the chosen constraint: 

\begin{corollary}
\label{Cor-EstDiff}
For any $\pmb{v}$ and $\pmb{u}$ defined in Proposition \ref{Prop.LSE.uv} and any $1\leq i,j\leq K$ we have
\begin{equation*}
\widehat{\mu }_{i}(\bv)-\widehat{\mu }_{j}(\bv)=
\widehat{\mu }_{i}(\pmb{u})-\widehat{\mu }_{j}(\pmb{u}).
\end{equation*}
\end{corollary}

Therefore, for simplicity and unless explicitly stated otherwise, we henceforth assume that $\bv=\b1$. 

\begin{remark}
Some authors consider a weighted objective function, i.e., $Q_{\pmb{w}}\left( \bmu\right)=\sum_{i<j}\sum_{k=1}^{n_{ij}}\\ w_{ij}(Y_{ijk}-\left( \mu _{i}-\mu _{j}\right))^{2}$. The weighted model is used to handle heteroscedasticity, i.e., different comparisons carry different levels of information. The analysis of the weighted and unweighted sum of squares is similar. The only difference is that in equation \eqref{mu.hat} we replace $\bS$ by $\bS_w$ and $\bN$ by $\bN_{w}$  where the $j^{th}$ element of  $\bS_w$ is $\sum_{i\neq j}w_{ji}S_{ji}$ and $\bN_{w}$ is the $K\times K$ weighted Laplacian with elements $\sum_{j}w_{ij}n_{ij}$ when $i=j$ and $-w_{ij}n_{ij}$ otherwise. Since the analysis of the weighted and unweighted cases is similar we shall henceforth consider only the unweighted model.  
\end{remark}

We will assume:

\begin{condition}
\label{iid.errors}
The errors $\epsilon_{ijk}$ are IID with zero mean and a finite variance $\sigma^2$.   
\end{condition}

Next, we address the statistical properties of (\ref{mu.hat}). Since $\mathbb{E}(Y_{ijk})=\mu _{i}-\mu _{j}$ it follows that
\begin{eqnarray*}
\mathbb{E}(\bS) &=&(\mathbb{E}(S_{1}),\ldots ,\mathbb{E}
(S_{K}))^\top=(\mathbb{E}(\sum_{j\neq 1}\sum_{k=1}^{n_{1j}}Y_{1jk}),\ldots ,
\mathbb{E}(\sum_{j\neq K}\sum_{k=1}^{n_{Kj}}Y_{Kjk}))^\top \\
&=&(\sum_{j\neq 1}n_{1j}\left( \mu _{1}-\mu _{j}\right) ,\ldots ,\sum_{j\neq K}n_{Kj}\left( \mu _{K}-\mu _{j}\right) )^\top=\bN\bmu
\end{eqnarray*}
so $\mathbb{E}(\widehat{\bmu})=\bN^{+}\pmb{N\mu }$. The matrix $\bN^{+}\bN$ is symmetric and idempotent therefore it is a projection onto\textbf{ }$\mathrm{im}\left(\bN\right)$. Moreover, since $\mathrm{\ker }\left( \bN\right) =\pmb{1}$ and $\pmb{1}^\top\bmu=0$ by assumption, it follows that $\bmu\in \mathrm{im}\left( \bN\right)$ and therefore $\bN^{+}\pmb{N\mu }=\pmb{\mu }$. Hence $\widehat{\bmu}$ is unbiased. Further note that under Condition \ref{iid.errors}, $\mathbb{V} ar(S_{i})=\sigma ^{2}n_{i}$ and $\mathbb{C}ov(S_{i},S_{j})=-\sigma
^{2}n_{ij} $ so
\begin{equation}
\mathbb{V}ar(\widehat{\bmu})=\mathbb{V}ar(\bN^{+}
\bS)=\bN^{+}(\sigma ^{2}\bN)\bN
^{+}=\sigma ^{2}\bN^{+}.  \label{var.mu.hat}
\end{equation}
It is clear from \eqref{var.mu.hat} that the precision of the LSE is a function of $\bN^{+}$ and therefore depends on the properties of the comparison graph {  a fact further reinforced by  observing that the mean squared error of $\widehat{\bmu}_{n}$ is
\begin{eqnarray*}
\mathbb{E}(\left\Vert \widehat{\bmu}_{n}-\bmu%
\right\Vert _{2}^{2}) &=&\mathbb{E}(\left\Vert \bN^{+}\bS -\bN^{+}\pmb{N\mu }\right\Vert _{2}^{2})=\mathbb{E(}
\left\Vert \bN^{+}(\bS-\pmb{N\mu })\right\Vert
_{2}^{2}) \\
&=&\mathbb{E(}(\bS-\pmb{N\mu })^\top(\bN^{+}
\bN^{+})(\bS-\pmb{N\mu })) = \mathrm{trace}(\bN^{+}\bN^{+}(\sigma^2\bN))=\sigma^2\mathrm{trace}(\bN^{+}).
\end{eqnarray*}
}

The following example shows that even when the total number of paired comparisons is fixed, the structure of the graph $\mathcal{G}$ has a powerful effect on the precision of the estimators.

\begin{example} \label{example:graphs:Laplacian:mpinv}
We evaluate the precision of the LSE on several types of graphs. These include the complete, cycle, path, star, wheel and (knockout) tournament graphs as displayed in Figure \ref{graph:topology}. In the context of the tournament graph we assume, without any loss of generality, that the item with the smaller index moves up the tournament graph. 

\medskip
\centerline{Figure \ref{graph:topology} Comes Here}
\medskip

In a complete graph with $K$ items there are $K(K-1)/2$ paired comparisons whereas in a path graph there are only $K-1$ paired comparisons. Therefore a meaningful comparison among the graphs requires that all Laplacians be scaled so that the total number of paired comparisons, which is equal to $\rm{trace}(\bN)/2$, is common to all graphs. For convenience we choose to scale up to the complete graph. Since, the standard tournament graph has $2^m$ vertices, where $m\in\mathbb{N}$, we choose $K\in \{4,8,16,32,64,128,256,512\}$. We are interested in the overall precision of $\widehat{\bmu}$ which can be assessed, cf., Atkinson et al. (2007), by various functions of the (positive) eigenvalues of $\bN^{+}$ such as their sum, product, largest and smallest value. In Figure \ref{fig:sum_ev} and \ref{fig:max_ev}, we display the sum of the eigenvalues and the maximum eigenvalues of the complete, cycle, path, star, wheel and (knockout) tournament graphs as a function of the number of items we are comparing. Larger values indicate higher variability. 

\medskip
\centerline{Figure \ref{fig:precision_analysis1} Comes Here}
\medskip

Each type of graph is associated with a line in Figures \ref{fig:sum_ev} and \ref{fig:max_ev}. The top line corresponds to the path graph on which the estimators are most variable and the bottom line is associated with the complete graph which generates the most precise estimators; the performances of (the scaled) star, wheel and knockout graphs are close to those of the complete graph. Thus, graphs with high connectivity (Cvetkovi\'c et al. 2009) result in more precise estimators. Among the displayed graphs, the path graph shows the fastest increase in variability as a function of $K$. This is due to its low connectivity. As an example consider comparing the difference between the first and last items in a path graph. Obviously these items are not directly compared. In fact,
$$ \mathbb{V}ar(\widehat{\mu}_1-\widehat{\mu}_K)=\mathbb{V}ar((\widehat{\mu}_1-\widehat{\mu}_2)+(\widehat{\mu}_2-\widehat{\mu}_3)+\cdots +(\widehat{\mu}_{K-1}-\widehat{\mu}_K)),$$
i.e., the variance of their difference is compounded over a long path. In high connectivity graphs the paths between any two items is short and therefore comparisons among items are more precise. The performance of the knockout tournament graph is better than the path and cycle graph. It is also interesting to note that all lines in Figures \ref{fig:sum_ev} and \ref{fig:max_ev} are increasing except for the line associated with the complete graph. Finally, in symmetric graphs in which all vertices have the same degree, the components of $\bmu$ are estimated with the same variance, e.g., complete graph, cycle graph. If the vertices of the graph have different degrees then the components of $\bmu$ are estimated with different variances. The merits of central items (cf., Cvetkovi\'c et al. 2009) are also estimated with more precision. For example, the merit of the item in the center of the path graph is estimated more precisely than items at the periphery of the graph. Similarly for the central item of the star graph and the wheel graph. 
\end{example}

We introduce further notations. The graph $\mathcal{G}_{2}$ is a subgraph of the graph $\mathcal{G}_{1}$, denoted $\mathcal{G}_{1}\supseteq \mathcal{G}_{2}$, if $\mathcal{V}_{1}\supseteq \mathcal{V}_{2}$ and $\mathcal{E}_{1}\supseteq \mathcal{E}_{2}$. The latter notion extends to PCG where $(\mathcal{G}_{1},\mathcal{Y}_{1})\supseteq (\mathcal{G}_{2},\mathcal{Y}_{2})$ whenever $\mathcal{G}_{1}\supseteq \mathcal{G}_{2}$ and $\mathcal{Y}_{1}\supseteq \mathcal{Y}_{2}$, {  i.e., if $Y_{ijk}\in\mathcal{Y}_{2}$ then $Y_{ijk}\in\mathcal{Y}_{1}$}. A matrix $\bV_1$ is said to be smaller than a matrix $\bV_2$, in the Loewner order, if $\bV_2-\bV_1$ is non--negative definite. This relationship is denoted by $\bV_1\preceq\bV_2$. For more on the Loewner order see Pukelsheim (2006). In particular if $\bV_1$ and $\bV_2$ are the variances of two (asymptotically) unbiased estimators (of the same quantity) then $\bV_1\preceq\bV_2$ implies that the estimator associated with $\bV_1$ is more efficient than the estimator associated with $\bV_2$. This means, as an example,  that the volume of the confidence ellipsoid associated with $\bV_1$ is smaller than the volume of the confidence ellipsoid associated with $\bV_2$. 

\begin{proposition}\label{prop1}
Let $(\mathcal{G}_1,\mathcal{Y}_1)$ and $(\mathcal{G}_2,\mathcal{Y}_2)$ be comparison graphs {  on the same set of vertices} with unique LSEs $\widehat{\bmu}_1$ and $\widehat{\bmu}_2$ whose variances are $\bV_{1}$ and $\bV_{2}$ respectively. If $(\mathcal{G}_1,\mathcal{Y}_1) \supseteq (\mathcal{G}_2,\mathcal{Y}_2)$ then 
\begin{equation}
\bV_{1} \preceq \bV_{2}.
\end{equation}
Moreover, for any continuously differentiable function $\pmb{\Phi}(\bmu)$ we have $\bV_{1}^{\pmb{\Phi}} \preceq \bV_{2}^{\pmb{\Phi}}$ where $\bV_{1}^{\pmb{\Phi}}$ and $\bV_{2}^{\pmb{\Phi}}$ are the asymptotic variances of $\pmb{\Phi}(\widehat{\bmu}_1)$ and $\pmb{\Phi}(\widehat{\bmu}_2)$ respectively. 
\end{proposition}

Proposition \ref{prop1} shows that the variance of the LSE or any function thereof decreases when the number of paired comparisons on each edge increases.  

\section{Large sample theory} \label{section large sample theory}

The asymptotic theory for the LSE, denoted henceforth by $\widehat{\bmu}_{n}$ to emphasize its dependence on the total number of paired comparisons, requires some additional graph--based notions. A graph $\mathcal{T}$ is called a tree if any two distinct vertices in $\mathcal{T}$ are joined by a unique path. If the tree $\mathcal{T}$ is a subgraph of $\mathcal{G}$ and connects all vertices in $\mathcal{G}$ then it is called a spanning tree. Clearly, if $\mathcal{G}$ is connected then a spanning tree $\mathcal{T}$ exists (Gould, 2012). 

\begin{condition} \label{Con(AlgConnec)}
There exists a spanning tree $\mathcal{T}\subseteq\mathcal{G}$ such that $\min \{n_{ij}:\left(i,j\right) \in \mathcal{T} \}\rightarrow \infty$ as $n\rightarrow \infty$.
\end{condition}

Condition \ref{Con(AlgConnec)} ensures that the minimum number of comparisons along some spanning tree increases to infinity. More concretely, for any pair $i,\,j\in\mathcal{V}$, either $n_{ij}\to\infty$ or there exists a path  $v_1,v_2,\ldots,v_l\in\mathcal{V}$ such that $\min\{n_{iv_1}, n_{v_1v_2},n_{v_2v_3},\ldots,n_{v_{l-1}v_l},n_{v_lj}\}\to\infty$. That is, all pairs of items are compared infinitely many times either directly or indirectly. The rate at which $n_{ij}$ where $(i,j)\in \mathcal{T}$ grows to infinity is left unspecified; in fact the rate may be edge specific. {  There may be many spanning trees on which Condition \ref{Con(AlgConnec)} holds. We define
\begin{align} \label{m:maxmin}
m =\displaystyle{\max_{\mathcal{T}\subset \mathcal{G}}}\min \{n_{ij}:(i,j)\in \mathcal{T} \}.
\end{align}
and denote a tree on which \eqref{m:maxmin} is attained by $\mathcal{T}_m$. Thus, $\mathcal{T}_m$ is a tree on which the minimum number of paired comparisons grows the fastest. }

\begin{theorem}
\label{Thm-graph.WLLN} 
If Conditions \ref{iid.errors} and \ref{Con(AlgConnec)} hold then $\mathbb{V}ar(\widehat{\bmu}_{n})\rightarrow \pmb{O}$, where $\pmb{O}$ is the ${K\times K}$ matrix of zeros and consequently $\widehat{\bmu}_{n}\rightarrow \bmu$ in probability. { Moreover} for each $i$,
{ 
\begin{equation} \label{var.mu.i.hat}
\mathbb{V}ar(\widehat{\mu }_{i,n}) = O(1/m).
\end{equation}}
Finally, if Condition \ref{Con(AlgConnec)} does not hold then $\mathbb{V}ar(\widehat{\bmu}_{n})\not\to \pmb{O}$.
\end{theorem}

Theorem \ref{Thm-graph.WLLN} establishes that the LSE is consistent provided Condition \ref{Con(AlgConnec)} holds. 
In the proof of Theorem \ref{Thm-graph.WLLN} it is shown that Condition \ref{Con(AlgConnec)} is necessary and sufficient for the second smallest eigenvalue of $\bN$, known as the algebraic connectivity (Cvetkovic et al. 2009) and denoted by $\lambda_2=\lambda _{2}(\bN)$, to converge to $\infty$. Using graph decomposition arguments, and Weyl's inequality (Horn and Johnson, 2007) it can be shown that the latter implies that $\bN^{+}\to \pmb{O}$, the zero matrix, so $\widehat{\bmu}_{n}$ is consistent. Moreover, by the spectral decomposition $\bN^{+}=O(1/\lambda_2)$, {element--wise}. In addition, it is shown that
\begin{equation} \label{Eq.lam2=m}
\lambda_2 = O(m).    
\end{equation}
Thus, the variance of any merit estimator is of order $1/m$, where $m$ is defined in \eqref{m:maxmin}, and is not directly related to the sample size $n_i$ as in the usual $K$ sample case. Note that Condition \ref{Con(AlgConnec)} implies that $n_i\to\infty$ for all $i\in\mathcal{V}$; the reverse implication is not necessarily true. For example consider the graph $\mathcal{G}$ with $\mathcal{V}=\{1,2,3,4\}$ and $\mathcal{E}=\{(1,2),(2,3),(3,4)\}$ where $k=n_{12}=n_{34}$ and $n_{23}=1$. Clearly $n_{i}\to \infty$ for all $i$ as $k\to \infty$ but Condition \ref{Con(AlgConnec)} does not hold. Moreover, for this graph $\lambda_{2}=k+1-\sqrt{k^2+1}\leq 1$ for all $k\in\mathbb{N}$. Therefore $\bN^{+} \not\to \pmb{O}$ and the LSE is not consistent. It is also worthwhile noting that in general $O(m) < O(\min\{n_1,\ldots,n_K\})$, however there are situations in which either $m=O(\min\{n_{1},\ldots,n_K\})$ or $m=O(n)$.

\begin{remark}
As noted by a referee we use the $O(\cdot)$ notation correctly, albeit, somewhat casually. For example, in Equation \eqref{Eq.lam2=m} we write $\lambda_2 = O(m)$. However, in the proof of Theorem \ref{Thm-graph.WLLN} it is shown that $\lambda_2/m\to c$ for some constant $c$ which depends on the topology of the graph. Thus, it would have been more precise to write $\lambda_2 = \Theta(m)$ as common in the literature in computer science, e.g., Cormen et al. (2022). Nevertheless, we will adhere to the less precise $O(\cdot)$ notation as common in the statistical literature. 
\end{remark}

\medskip

Next, strong consistency without a variance is established. 

\begin{theorem}
\label{Thm-graph.SLLN}
If the errors in \eqref{model.Y_ijl} are IID with mean zero then $\widehat{\bmu}_{n}\rightarrow \bmu$ with probability one if and only if Condition \ref{Con(AlgConnec)} holds.
\end{theorem}
 
Thus, provided the errors have mean zero, the merits converge to their true values with probability one. The rate of convergence is explored by simulation in Section \ref{numerical.expriments}. The proof of Theorem \ref{Thm-graph.SLLN} is rather involved. In addition to the usual probabilistic arguments the proof relies on specialized results in linear algebra (Satorra and Neudecker, 2015) and graph theory. In particular we use the all minors matrix tree theorem (Chen 1976, Chaiken, 1982) and various properties of Laplacians (Chebotarev and Shamis 1998, and Bapat 2013). 

\medskip

The large sampling distribution of the LSE is explored under two different regimes.

\begin{condition}
\label{Con(Rate+CLT)}
There exists a spanning tree $\mathcal{T}\subset \mathcal{G}$ such that
\begin{equation*}
\min \{n_{ij}:(i,j)\in \mathcal{T}\}/n\to c\in(0,\infty)\text{ as }n\to\infty.
\end{equation*}
\end{condition}

It is clear that Condition \ref{Con(Rate+CLT)} is stronger than Condition \ref{Con(AlgConnec)}. Condition \ref{Con(Rate+CLT)} implies that $n_{ij}/n_{lk}\rightarrow c_{ijlk} \in (0,\infty)$ for all pairs $(i,j)$ and $(l,k)$ in $\mathcal{T}$. It can be further shown that Condition \ref{Con(Rate+CLT)} implies that for any two items $n_i/n_j\rightarrow c_{ij} \in (0,\infty)$, i.e., the total number of paired comparisons for all items grows at the same rate. Formally this means that $n_{ij}=O(n)$ for all $(i,j)\in \mathcal{T}$ and $n_i=O(n)$ for all $i \in \mathcal{V}$. However, $n_i=O(n)$ for all $i \in \mathcal{V}$ does not imply that Condition \ref{Con(Rate+CLT)} holds. {  For example, consider a path graph with four vertices for which $n_{12}=n_{34}=m^2$ and $n_{23}=m$. It is easy to verify that $n_i=O(m^2)$ whereas $\min{n_{ij}}/n \to 0$}.

\medskip

Condition  \ref{Con(Rate+CLT)} results in well--behaved variance matrices for the LSE. 

\begin{lemma} \label{cov:insure}
If Condition \ref{Con(Rate+CLT)} holds then $\pmb{\Theta }=\lim_{n}\bN/n$ exists and $\mathrm{rank}(\pmb{\Theta})=K-1$. Moreover $n\bN^{+}\to {\pmb{\Theta}^{+}}$ as $n\to\infty$ where all diagonal elements of $\pmb{\Theta }^{+}$ are positive and $\mathrm{rank}(\pmb{\Theta}^{+})=K-1$.
\end{lemma}

We are now ready for {  our first limit law.} 

\begin{theorem} \label{Thm-LST}
Suppose that Conditions \ref{iid.errors} and  \ref{Con(Rate+CLT)} hold. Then, $\sqrt{n}\left( \widehat{\bmu}_{n}-\bmu\right) \Rightarrow \mathcal{N}_{K}( \pmb{0},\sigma^{2}\pmb{\Theta }^{+})$.
\end{theorem}

We conclude that under Condition \ref{Con(Rate+CLT)} the centered LSE, scaled by $\sqrt{n}$,
approximately follows a normal distribution in large samples. The asymptotic variance in Theorem \ref{Thm-LST} is determined only by the limit of the graph Laplacian, i.e., by connectivity properties of $\mathcal{G}$. In particular the limiting variance depends only on only those $(i,j)\in \mathcal{E}$ for which $n_{ij}=O(n)$, all others are awash. 

\begin{remark}
Following a comment by a referee we note that using Proposition \ref{Prop.LSE.uv} and Theorem \ref{Thm-LST} we deduce that for any proper constraint vector $\pmb{u}$ we have $\sqrt{n}(\widehat{\bmu}_{n}(\pmb{u})-\bmu(\pmb{u})) \Rightarrow \mathcal{N}_{K}( \pmb{0},\sigma^{2}\pmb{C}_{\pmb{u}}\pmb{\Theta }^{+}\pmb{C}_{\pmb{u}}^{\top})$, where $\bmu(\pmb{u})=\bC_{\bu}\bmu$ and $\bmu \in \rm{im}(\bN)$.
\end{remark}

Further note that Condition \ref{Con(Rate+CLT)}, which implies Lemma \ref{cov:insure}, is necessary for Theorem \ref{Thm-LST} as exemplified below. 

\begin{example}\label{example.clt.with3.1}
Consider the tree--structured comparison graph with $K=3$,  $n_{12}=m$ and $n_{23}=m^2$ so $n=m+m^2$. Note that if $m\to\infty$ Condition \ref{Con(AlgConnec)} holds but condition \ref{Con(Rate+CLT)} does not. It is easy to see that
$$
\pmb{N}= \begin{pmatrix}
m & -m & 0\\
-m & m^2+m & -m^2\\
0 & -m^2 & m^2
\end{pmatrix} \text{ and }
\pmb{N}^+=\frac{1}{9m^2} \begin{pmatrix}
4m+1 & 1-2m & -2(m+1)\\
1-2m & m+1 & m-2\\
-2(m+1) & m-2 & m+4
\end{pmatrix}.
$$
Observe that $\pmb{N}/n$ converges to a matrix with rank one, whereas the matrix $n\pmb{N}^+$ diverges. Thus, the conclusions of Lemma \ref{cov:insure} do not hold. These linear algebraic facts hide the probabilistic relations
$$
\frac{S_{12}-m(\mu_{1}-\mu_{2})}{\sqrt{m}}\Rightarrow \mathcal{N}(0,\sigma^2)
\ \quad \ \mbox{and} \quad \quad \frac{S_{23}-m^2(\mu_{2}-\mu_{3})}{m}\Rightarrow \mathcal{N}(0,\sigma^2),
$$ 
i.e., convergence occurs at different rates and these two limits can not be combined as in the proof of Theorem \ref{Thm-LST}. However,   
\begin{align*}
\widehat{\bmu}_{n}=&~\bN^{+}\pmb{S} = \bN^{+} \begin{pmatrix}
    S_{12}\\-S_{12}+S_{23}\\-S_{23} 
\end{pmatrix}
= \bN^{+} \begin{pmatrix}
    m(\mu_1-\mu_2)+ \sqrt{m}\sigma Z_{12}+o_p(\sqrt{m})\\
    m(\mu_2-\mu_1)+m^2(\mu_2-\mu_3)-\sqrt{m}\sigma Z_{12}+m\sigma Z_{23}+o_p(m)\\
    m^2(\mu_3-\mu_2)-m\sigma Z_{23}+o_p(m)
\end{pmatrix}   \\
=&~ \bN^{+} \{\begin{pmatrix}
    m(\mu_1-\mu_2)\\ m(\mu_2-\mu_1)+m^2(\mu_2-\mu_3)\\ m^2(\mu_3-\mu_2)
\end{pmatrix}
+ \sigma  \begin{pmatrix}
    \sqrt{m} & 0 \\
    -\sqrt{m} & m\\
    0 & -m
\end{pmatrix} \begin{pmatrix}
    Z_{12}\\ Z_{23}
\end{pmatrix} + o_p(m)\}\\
=&~ \bmu +  \frac{\sigma}{9m^2} \begin{pmatrix}
    6m\sqrt{m} & 3m\\
    -3m\sqrt{m} & 3m\\
     -3m\sqrt{m} & -6m
\end{pmatrix}\begin{pmatrix}
    Z_{12}\\ Z_{23}
\end{pmatrix}+ o_p(1),
\end{align*}
where $Z_{12},Z_{23}$ are independent $\mathcal{N}(0,1)$ RVs. It follows that
\begin{align*}
\sqrt{m}(\widehat{\bmu}_{n}- \bmu)= \frac{\sigma}{9} 
\begin{pmatrix}
    6 & 3m^{-1/2}\\
   -3 & 3m^{-1/2}\\
   -3 & -6m^{-1/2}
\end{pmatrix}\begin{pmatrix}
Z_{12}\\ Z_{23}
\end{pmatrix}+ o_p(1) \Rightarrow   
\begin{pmatrix}
    2/3 \\
   -1/3 \\
   -1/3 
\end{pmatrix}\sigma Z_{12},
\end{align*}
as $m\to\infty$. We conclude that a limiting distribution for the LSE is possible when Condition \ref{Con(AlgConnec)} holds but Condition \ref{Con(Rate+CLT)} does not. However, in the absence of Condition \ref{Con(Rate+CLT)} the limiting distribution can not be scaled by $\sqrt{n}$ but by the much smaller $\sqrt{m}$, moreover, the resulting distribution is not supported on the two dimensional subspace spanned by the columns of $\bN$ but a one dimensional subspace thereof.   
\end{example}

We are now ready for a second limit law which generalizes Example \ref{example.clt.with3.1}.   

\begin{theorem} \label{theorem clt with 3.1}
If Conditions \ref{iid.errors} and \ref{Con(AlgConnec)} hold then 
\begin{align} \label{Eq.LSD.II}
\sqrt{m}(\widehat{\bmu}_n-\bmu) \Rightarrow  \mathcal{N}_K(\pmb{0},\sigma^2\pmb{\Psi})    
\end{align}
where $m$ is defined in \eqref{m:maxmin} and $\pmb{\Psi}=\lim_m m\bN^{+}$. 
\end{theorem}

Theorem \ref{theorem clt with 3.1} shows that Condition \ref{Con(AlgConnec)} is sufficient for a limiting distribution of the merits. However, the scaling factor $\sqrt{m}$ in \eqref{Eq.LSD.II} may be much smaller than the scaling factor $\sqrt{n}$ appearing in Theorem \ref{Thm-LST}. In fact, under Condition \ref{Con(AlgConnec)} it is possible for $m/n \to 0$. Moreover, while the rank of $\sigma^2\pmb{\Theta}^{+}$ is $K-1$, the rank of $\pmb{\Psi}$ is typically smaller. Finally, it is not difficult to verify that approximating the distributions of $\sqrt{n}(\widehat{\bmu}_n-\bmu)$ and $\sqrt{m}(\widehat{\bmu}_n-\bmu)$ by their plug--in--values results in the same expression. Thus, differences arise only asymptotically. Elucidating possible limits under Theorem \ref{theorem clt with 3.1} requires a detailed knowledge of how $n_{ij}\to \infty$ for all $(i,j)\in\mathcal{E}$ from which the limiting eigenvalues and eigenvectors can be deduced.

\medskip

Next, we focus on the associated ranks. Recall that corresponding to the vector of merits $\bmu$ we have a vector of ranks $\pmb{r}=\pmb{r}\left( \bmu\right)=(r_1,\ldots,r_K)^{\top}$ where $r_{i}=\sum_{j}\mathbb{I}_{\left( \mu _{i}\leq \mu _{j}\right) }$. Clearly $\pmb{r}$ is a permutation of the integers $1,\ldots ,K$. Note that if $r_{i}=1$ then the $i^{th}$ score is the largest score, if $r_{j}=2$ then the $j^{th}$ score is the second largest and so forth. 

\begin{theorem} \label{Thm-AS}
Let $\bmu\in \mathbb{R}^{K}$ with $\mu _{i}\neq \mu _{j}$ for all $i \neq j$. If Conditions \ref{iid.errors} and \ref{Con(AlgConnec)} hold then \linebreak 
$\mathbb{P}\left( \pmb{r}( \widehat{\bmu}_{n}) \neq \pmb{r}\left( \bmu\right) \right) \rightarrow 0$ as $n\rightarrow \infty$. Moreover, if the errors in \eqref{model.Y_ijl} are subgaussian and Condition \ref{Con(Rate+CLT)} holds then there exist positive constants $C_{1}$ and $C_{2}$ for which $\mathbb{P}( \pmb{r}( \widehat{\bmu}_{n}) \neq \pmb{r}( \bmu) ) \leq C_{1}\exp (-nC_{2})$.
\end{theorem}

Theorem \ref{Thm-AS} shows that the ranks can be identified with probability one at an exponentially fast rate if the tails of the error distribution decay fast enough. Let $\rho(\cdot,\cdot)$ be any metric on ranks such as the Kendall, Caley and the Footnote distances (cf. Marden 1995). Theorem \ref{Thm-AS} insures that $\mathbb{P}( \rho(\pmb{r}( \widehat{\bmu}_{n}), \pmb{r}( \bmu) )>0 )$ is exponentially small as $n$ grows to $\infty$. 

\begin{remark}
In fact, since: $(i)$ $\widehat{\bmu}_{n} \rightarrow\bmu$ with probability one; and $(ii)$ the function  $\pmb{r}(\bmu)$ is continuous on the set $(\mu_1-\epsilon, \mu_1+\epsilon)\times \cdots \times (\mu_K-\epsilon, \mu_K+\epsilon)$ where $\epsilon < \min\{|\mu_i-\mu_j|/2:1\leq i,j\leq K\}$; it follows by the continuous mapping theorem that $\lim \pmb{r}(\widehat{\bmu}_{n}) = \pmb{r}(\bmu)$ with probability one. 
\end{remark}

\medskip

We conclude this section with a selection of hypothesis testing problems, formulated assuming Condition \ref{Con(Rate+CLT)} hold, which can be addressed using the theory developed thus far. 

\begin{enumerate}
 
\item Consider the hypothses $H_0^{(1)}:\, \bmu = \pmb{0}$ versus $H_1^{(1)}:\, \bmu \neq \pmb{0}$. Thus, under the null all items are equivalent whereas under the alternative they are not. Hotelling's statistic for this problem is  
\begin{align} \label{test:all:same}
T_{n,K}^{(1)} =n \widehat{\bmu}^{\top} (\widehat{\sigma}^2\pmb{\Theta }_n^{+})^{+} \widehat{\bmu}= \frac{1}{\widehat{\sigma}^2}\widehat{\bmu}^{\top} \bN \widehat{\bmu} =\frac{1}{\widehat{\sigma}^2} \sum_{i\neq j}n_{ij}(\widehat \mu_i- \widehat \mu_{j})^2, 
\end{align}
where $\pmb{\Theta }_n=\bN/n$, $\widehat{\sigma}^2= Q(\widehat{\bmu})/n$ is a natural estimator of $\sigma^2$ and the last equality above follows from Lemma 4.3 in  Bapat (2010). Furthermore, by combining Theorem \ref{Thm-LST} with Theorem 2 of Moore (1977) we find that under local alternatives, i.e., when the vector of merits equals $\bmu/\sqrt{n}$, we have 
$$T_{n,K}^{(1)} \Rightarrow \chi^{2}(K-1, \bmu^{\top}\pmb{\Theta }\bmu),$$ 
Therefore the null is rejected if the right--hand--side of \eqref{test:all:same} exceeds the $1-\alpha$ quantile of $\chi^{2}(K-1)$. This procedure can be modified for testing the equality among any subset of merits. For example, when testing $H_0^{(1)}:\, \bmu_R= \bzero$ against $H_1^{(1)}:\, \bmu_R\neq \bzero$ where $\bmu_R= (\mu_{i_1}, \mu_{i_2}, \ldots, \mu_{i_R})^{\top}$ and $i_1,\ldots,i_R$ is a strict subset of $1,\ldots,K$ then \eqref{test:all:same} reduces to $T_{n,R}^{(1)} = \frac{n}{\widehat{\sigma}^2} \widehat{\bmu}_R^{\top} \pmb{\Theta}_{n,R} \widehat{\bmu}_R$ and $T_{n,R}^{(1)}\Rightarrow \chi^{2}(R,\bmu_R^{\top} \pmb{\Theta}_R \bmu_R)$ under local alternatives, where $\pmb{\Theta}_R^{-1}$ is the submatrix of $\pmb{\Theta}^{+}$ corresponding to items $\{i_1,\ldots, i_R\}$. In particular if $r=2$ then 
$$T_{n,2}^{(1)}=\frac{1}{\widehat{\sigma}^2} n(\widehat\mu_i, \widehat\mu_j)
\begin{pmatrix}
    \theta_{n,ii}^{+} & \theta_{n,ij}^{+}\\
    \theta_{n,ij}^{+} & \theta_{n,jj}^{+}
\end{pmatrix}^{-1} 
\begin{pmatrix}
\widehat\mu_i\\ \widehat\mu_j
\end{pmatrix}
\Rightarrow \chi^{2} (2, (\mu_i,\mu_j) \begin{pmatrix}
    \theta_{ii}^{+} & \theta_{ij}^{+}\\
    \theta_{ij}^{+} & \theta_{jj}^{+}
\end{pmatrix}^{-1} \begin{pmatrix}
\mu_i\\ \mu_j
\end{pmatrix}),$$
where $\theta_{n,ij}^{+}$ and $\theta_{ij}^{+}$ are $(i,j)^{th}$ element of $\pmb{\Theta}_n^{+}$ and $\pmb{\Theta}^{+}$, respectively.

\item Next, consider an ANOVA type hypothesis of the type $H_0^{(2)}:\, \mu_{i_1}= \mu_{i_2}= \ldots= \mu_{i_R}$ against $H_0^{(2)}:\, \mu_{i_l}\neq \mu_{i_s}$ for some $i_1\leq l\neq s \leq i_R$. This system can be reformulated as $H_0^{(2)}:\,\pmb{R}\bmu_{R} = \bzero$ versus $H_0^{(2)}:\, \pmb{R}\bmu_{R} \neq  \bzero$ where $\pmb{R}$ is a $(R-1) \times R$ matrix whose $i^{th}$ row is a contrast comparing the $i^{th}$ and $(i+1)^{th}$ elements of $\bmu_R$. The natural test statistic is 
\begin{align} \label{anova:type_stat}
T_{n,R}^{(2)} = \frac{n}{\widehat{\sigma}^2} \widehat\bmu_{R}^{\top} \pmb{R}^{\top} \pmb{R} \pmb{\Theta }_{n,R} \pmb{R}^{\top} \pmb{R}\widehat\bmu_{R},
\end{align}
which converges in distribution (under local alternatives) to a $\chi^2(R-1, \bmu_{R}^{\top} \pmb{R}^{\top} \pmb{R} \pmb{\Theta }_R \pmb{R}^{\top} \pmb{R}\bmu_{R})$ RV. In particular if $R=2$, i.e., items $i$ and $j$ are being compared then \eqref{anova:type_stat} reduces to 
$$T_{n,2}^{(2)} = \frac{n}{\widehat{\sigma}^2} (\widehat\mu_i - \widehat\mu_j)^2/(\theta_{n,ii}^{+}+ \theta_{n,jj}^{+}-2 \theta_{n,ij}^{+}) \Rightarrow \chi^2 (1, (\mu_i - \mu_j)^2).$$

\item In some circumstances it is of interest to test whether all items have different merits. Formally, this amounts to testing $H_0^{(3)}:\,\min_{1\leq i\neq j\leq K} (\mu_i- \mu_j)^2=0$ against $H_1^{(3)}:\,\min_{1\leq i\neq j\leq K} (\mu_i- \mu_j)^2>0$. Clearly this is an intersection--union type test (Casella and Berger, 2021), in which the null is the union of $\cup_{i\neq j}\{\mu_i=\mu_j\}$ and the alternative is of the form $\cap_{i\neq j}\{\mu_i \neq \mu_j\}$. An appropriate test statistic is 
$$
T_{n}^{(3)}=  \min_{i \neq j} n(\widehat\mu_i- \widehat\mu_j)^2.
$$
It is not hard to see that the least favorable configuration under the null hypothesis is attained at $\bmu = \pmb{0}$. Let $\pmb{D}$ be an $\binom{K}{2}\times K$ matrix of pairwise contrasts, e.g., its first row is of the form $(1,-1,0,\ldots,0)$ comparing the first and second items. The limiting distribution of $T_{n}^{(3)}$ is equal to the distribution of $T^{(3)} = \min_{1\leq i\leq \binom{K}{2}}\{{W}_i^{2}\}$ where $\pmb{W}$ is a $\mathcal{N}_{\binom{K}{2}}(\pmb{0},\sigma^2\pmb{D} \pmb{\Theta}^{+}\pmb{D}^{\top})$ RV. Hence $H_0^{(3)}$ can be tested by comparing $T_{n}^{(3)}$ with the $1-\alpha$ quantile of the distribution of $T^{(3)}$ which is easily computed by simulation. 

\item Many questions regarding the ranking of items are of interest. For example one may be interested in testing whether a specific item, the first say, is better than all others. Nettleton (2009) studied such a problem for multinomial cell probabilities. The complement of the latter testing problem is $H_0^{(4)}:\, \mu_1\leq \mu_j$ for all $j\neq 1$ against $H_1^{(4)}:\, \mu_1 > \mu_j$ for some $j\neq 1$; clearly under the null item $1$ is associated with the smallest merit whereas under the alternative it is not. Note that the null can be expressed as $\cap_{2\leq i\leq K} \{\mu_1\leq \mu_i \}$ and the alternative as $\cup_{2\leq i\leq K} \{\mu_1 > \mu_i \}$. A natural test statistic is for this problem is
$$ T^{(4)}_n = \sqrt{n} ( \widehat\mu_1- \min_{2\leq i\leq K}\widehat\mu_i)
$$
Clearly the least favourable configuration under the null is attained when $\bmu = \bzero$. Let $\pmb{E}$ be $(K-1)\times K$ matrix whose rows are contrasts with respect to the first item, i.e., in $i^{th}$ row first element is $1$, $(i+1)^{th}$ element is $-1$ and other elements are zero. The limiting distribution of $T_{n}^{(4)}$ is equal to $T^{(4)}=\min_{1\leq i\leq K-1} \{Y_1\,\ldots,Y_{K-1}\}$ where $\bY\sim \mathcal{N}_{K-1}( \pmb{0}, \sigma^{2}\pmb{E}\pmb{\Theta }^{+}\pmb{E}^{\top})$. The critical value can be easily obtained by Monte--Carlo methods. 

\end{enumerate}

There are many other related problems involving inference on item ranks (Xie et al. 2009). Unlike the classical theory PCGs induce a dependence among the estimated merits which complicates the analysis. See also Hung and Fithian (2019) for a related problem.

\section{Incorporating covariates} \label{section with covariates}

There are situations in which it is advisable and effective to incorporate covariates into models for paired comparisons. For example, Harville (2003) modeled the expected outcome of NCAA basketball league games between teams $i$ and $j$ by $\mu_i-\mu_j +\beta (x_{ik}-x_{jk})$ where $x_{ik} = 1$ for a home game and $x_{ik} = 0$ otherwise. Clearly, if $x_{ik} = 1$ then $x_{jk} = 0$ and vice versa. Therefore $x_{ik}-x_{jk}$ equals $1$ if the game is played on the home--court of team $i$ and $-1$ if the game is played on the home--court of team $j$. The parameter $\beta$ measures the effect of playing on ones own turf. More generally each paired comparison is associated with a pair $(\bx_{ik},\bx_{jk})$ of covariates. The resulting data structure which we denote by 
\begin{equation} 
\label{PCG+Covariates}
(\mathcal{G},\mathcal{Y},\mathcal{X})    
\end{equation}
is a PCG with covariates. Note that $\mathcal{G}$ and $\mathcal{Y}$ are as before while 
$\mathcal{X} = \{(\bx_{ik},\bx_{jk}): k=1,2,\ldots,n_{ij}; i,j=1,2,\ldots,K \}$. 

There are numerous ways of modeling the effect of covariates on the outcomes. We adopt the following model: 
\begin{equation}
\label{model:with:covariate}
Y_{ijk} = \mu_i - \mu_j + \bx_{ijk}^\top \pmb{\beta} + \epsilon_{ijk} 
\end{equation}
which is a natural, ANCOVA type, extension of \eqref{model.Y_ijl}. Covariates are incorporated by setting $\bx_{ijk} = \psi(\bx_{ik},\bx_{jk})$ for some function $\psi$. For example a simple and reasonable choice is $\psi(\pmb{u},\pmb{v})=\pmb{u}-\pmb{v}$ in which case $\bx_{ijk}= \bx_{ik}-\bx_{jk}$ mimicking the modelling choice of Harville (2003) and mirroring the relationship $Y_{ijk}=-Y_{jik}$. Implicitly this choice of $\psi$ corresponds to the very natural situation in which the merit of the $i^{th}$ item in the presence of a covariate $\bx$ is $\mu_i+\bx^{\top}\pmb{\beta}$. Many other options are possible including models in which each item is associated with an individual regression parameter. However, for the simplicity of the exposition we will not venture beyond the relatively general model \eqref{model:with:covariate}. Furthermore, we will not address model selection with respect to $\psi(\cdot,\cdot)$, which is of course an important problem, and assume that $\psi(\cdot,\cdot)$ is given a priori. We shall refer to $\bx_{ijk}$ as the covariate and emphasize that its value may vary across comparisons even among the same two items. It is assumed here that $\bx_{ijk}\in \mathbb{R}^{p}$ where $p$ is a given positive integer. The covariates may be fixed constants or stochastic; if so they are assumed, as usual, to be independent of the errors.  

The parameters in \eqref{model:with:covariate} may be estimated by
\begin{equation}  
\label{def.mu.hat.with.cov}
\widehat{\pmb{\theta}}^{\top}=(\widehat{\bmu}^{\top},\widehat{\pmb{\beta}}^{\top})=\arg \min \{Q(\bmu, \pmb{\beta}) : \bv^{\top}\bmu=0\},
\end{equation}
where ${\pmb{\theta}}^{\top}=({\bmu}^{\top},\pmb{\beta}^{\top})$ is the parameter of interest, $\bv$ was defined earlier, and  
\begin{align}
Q(\bmu,\pmb{\beta}) =\sum_{1\leq i<j\leq
K}\sum_{k=1}^{n_{ij}}(Y_{ijk}-\left( \mu _{i}-\mu _{j}\right) -
\pmb{\beta }^\top\bx_{ijk})^{2},
\label{obj:covariates}
\end{align}
is a sum of squares. We will refer to the resulting estimator as the LSE. It will be useful to view \eqref{model:with:covariate} as the following multiple linear regression model: 
\begin{align}  \label{joint-eq-LM-with-cov}
\bY = \pmb{H} \pmb{\theta} + \pmb{\epsilon},
\end{align}
where $\bY$ is the $n \times 1$ vector of outcomes arranged by their lexicographical order, i.e.,  
\begin{align*}
\bY = (Y_{121},\ldots,Y_{12n_{12}}, Y_{131},\ldots,Y_{13n_{13}}, \ldots, Y_{K-1,K1},\ldots,Y_{K-1,Kn_{K-1,K}})^{\top},
\end{align*}
$\pmb{\epsilon}$ is the corresponding vector of errors and $\pmb{H}=(\bM, \bX)$ is the design matrix. The matrix $\bM$ is a $n \times K$ matrix whose elements $\bM_{ij} \in \{-1,0,1\}$. Specifically if the $l^{th}$ element of  $\bY$ is $Y_{ijk}$ where $i<j$,  then the elements of the $l^{th}$ row of $\bM$ are $\bM_{li}=1,\bM_{lj}=-1$ and $\bM_{lk}=0$ for all $k \notin \{i,j\}$. The matrix $\bX$ is a $n\times p$ matrix whose $a^{th}$ column corresponds to the $a^{th}$ covariates arranged by their lexicographic order. With these notations it is easy to verify that
\begin{equation}
\label{example:cov:line:graph}
    \bS=\bM^{\top}\bY \text{ and } \bM^{\top}\bM=\bN.
\end{equation}

To fix ideas an example is provided. 
\begin{example}
\label{example1:with:covariates}
Consider a PCG $(\mathcal{G},\mathcal{Y},\mathcal{X})$ where $\mathcal{G}$ is a line graph with three vertices, $n_{12}= 1$ and $n_{23}=3$, and $p=2$. Here,  
\begin{align*}
\bY =\begin{pmatrix}
    Y_{121}\\ Y_{231}\\ Y_{232}\\ Y_{233}
\end{pmatrix},\, 
    \bM = \begin{pmatrix}
        1 & -1 & 0 \\
        0 & 1 & -1 \\
        0 & 1 & -1 \\
        0 & 1 & -1 
    \end{pmatrix},\,
    \bX = \begin{pmatrix}
        x_{1211} & x_{1212}\\ x_{2311} & x_{2312}\\ x_{2321} & x_{2322}\\ x_{2331} & x_{2332}
    \end{pmatrix} \,
    \pmb{\theta}= \begin{pmatrix}
        \mu_1\\ \mu_2\\ \mu_3\\ \beta_1\\ \beta_2
    \end{pmatrix}
     \text{and } 
    \pmb{\epsilon} = \begin{pmatrix} \epsilon_{121}\\ \epsilon_{231}\\ \epsilon_{232}\\ \epsilon_{233}\end{pmatrix}.
\end{align*}
Further, $\bS = (Y_{121},-Y_{121} + Y_{231}+ Y_{232}+ Y_{233}, -Y_{231}- Y_{232}- Y_{233})^{\top}$ and 
\begin{equation*}
\bN = \begin{pmatrix}
        1 & -1 & 0 \\
       -1 & 4 & -3 \\
        0 & -3 & 3 \\       
    \end{pmatrix}.
\end{equation*}
\end{example}

\begin{condition}     \label{condition:covariate:existence}
The matrices $\bM$ and $\bX$ satisfy 
$$\mathrm{rank}(\bM)=K-1 \text{ and }
\mathrm{rank}((\pmb{I}-\bM\bM^{+})\bX)=p.$$
\end{condition} 

By Theorem 5 in Puntanen et al. (2011) Condition \ref{condition:covariate:existence} is equivalent to $\mathrm{rank}(\pmb{H})=K+p-1$. Now $\mathrm{rank}(\bM)=\mathrm{rank}(\bM^{\top}\bM)=K-1$ if and only if $\mathcal{G}$ is connected, which is the sufficient condition for the existence of the LSE without covariates. It is not difficult to show that Condition \ref{condition:covariate:existence} holds provided: (i) $\mathrm{rank}(\bM)=K-1$; (ii) $\mathrm{rank}(\bX)=p$; and (iii) $\mathrm{im}(\bM)\cap \mathrm{im}(\bX) = \{\pmb{0}\}$ , i.e., the columns of $\bM$ and $\bX$ span different non--null subspaces. In particular, if $\bX$ is comprised of a single covariate, i.e.,  $\bX=\bx$, then verifying Condition \ref{condition:covariate:existence} is exceptionally simple.   

\begin{proposition} \label{span:M:condition}
Let $\bx\in\mathbb{R}^n$ then $\bx\in \mathrm{im}(\bM)$ if and only if $x_{ijk} = z_i-z_j$ for some constants $z_1,\ldots,z_K$ in which case Condition \ref{condition:covariate:existence} does not hold.  
\end{proposition}

The following theorem extends Theorem \ref{Thm-UniqueE} to the situation where covariates are present.

\begin{theorem} \label{Thm-Est.Covariates}
The LSE \eqref{def.mu.hat.with.cov} is unique if and only if Condition \ref{condition:covariate:existence} holds. 	Moreover, assuming $\bv=\b1$, we have: 
\begin{align}
    \label{LSE:cov:mu}
    \widehat{\bmu} =\bN^{+} (\bS-\bM^{\top}\bX\widehat{\pmb{\beta}}),
\end{align}
  where 
  \begin{align}
      \label{LSE:cov:beta}
      \widehat{\pmb{\beta}} = (\bX^\top (\pmb{I}-\bM\bN^{+}\bM^{\top}) \bX)^{-1}\bX^\top(\pmb{I}-\bM\bN^{+}\bM^{\top})\bY.
  \end{align}
\end{theorem}
Note that the expression for $\widehat{\pmb{\beta}}$ appearing in \eqref{LSE:cov:beta} is reminiscent of the weighted LSE in a linear regression model, whereas  $\widehat{\bmu}$ is the LSE of a graph--LM without covariates with (adjusted) outcomes $\widehat{Y}_{ijk}=\mu_i -\mu_j -\widehat{\pmb{\beta}}^{\top} \bx_{ijk}$. In the proof of Theorem \ref{Thm-Est.Covariates} a general expression for the LSE for any constraint vector $\bv$ is provided see \eqref{cov:thm1:eq4}. Observe that the normal equations associated with \eqref{joint-eq-LM-with-cov} are 
$\pmb{H}^{\top}\pmb{H}\pmb{\theta}= \pmb{H}^{\top} \bY$ which can also be written as
\begin{equation} \label{eq:18}
\begin{pmatrix}
    \bN & \bM^{\top}\bX\\
    \bX^{\top} \bM & \bX^{\top}\bX
\end{pmatrix}
\pmb{\theta} =
\begin{pmatrix}
    \pmb{S} \\ \bX^{\top}\bY
\end{pmatrix}.
\end{equation}
The Laplacian appearing above highlights the dependence of the estimators on the structure of the graph. Proposition \ref{span:M:condition} indicates that the LSE is well defined and unique for all reasonable design matrices. 

Let 
$$\phi_n=\measuredangle(\bM,\bX)=\inf\{\measuredangle(\pmb{u},\pmb{v}): \pmb{u}\in\mathrm{im}(\bM),\, \pmb{v}\in\mathrm{im}(\bX)\}$$ 
denote the smallest principal angle between the subspaces spanned by columns of the matrices $\bM$ and $\bX$ where the index $n$ on $\phi$ emphasizes its implicit dependence on the total sample size. 

\begin{condition} \label{cond:wlln:with:covariates}
The smallest eigenvalue of $\bX^{\top}\bX$ diverges and for all large $n$ we have $\phi_n \ge \phi$ for some $\phi\in(0,\pi/2)$. 
\end{condition}

\begin{lemma} \label{lemma:cod:with:cov:wlln}
If Conditions \ref{Con(AlgConnec)} and \ref{cond:wlln:with:covariates} hold then the second smallest eigenvalue of $\pmb{H}^{\top}\pmb{H}$ diverges. 
\end{lemma}

In fact, Conditions \ref{Con(AlgConnec)} and \ref{cond:wlln:with:covariates} are necessary and sufficient for the conclusion of Lemma \ref{lemma:cod:with:cov:wlln} and, as we show below, imply consistent estimation. It is well known, e.g., Fahrmeir and Kaufmann (1985), that consistent estimation in standard regression problems requires that the smallest eigenvalue of $\bX^{\top} \bX$ diverges. The latter together with Condition \ref{cond:wlln:with:covariates} guarantee that the norm of the residual of $\bX$ projected onto the span of $\bM$ does not vanish as $n\to\infty$. Otherwise the second smallest eigenvalue of $\pmb{H}^{\top} \pmb{H}$ may not diverge. Two relevant examples are provided.

\begin{example}
Consider a PCG $(\mathcal{G},\mathcal{Y},\mathcal{X})$ where $\mathcal{G}$ satisfies Condition \ref{Con(AlgConnec)}, $p=1$, so the matrix $\bX$ is just a vector $\bx$, and $n$ is large. Suppose that for the first $n_1$ observations we have $\mathrm{rank}(\pmb{H}_{n_1})=K$. Next assume that for all $i>n_1$ we set $x_i = M_{i1}$, i.e., we equate the $i^{th}$ elements of $\bx$ and $\bM_{1}$ the first column of $\bM$. Clearly, $\bM_1\in\mathrm{span}(\bM)$ and $\|\bx-\bM_1\|<\infty$. Since the orthogonal projection $\bM \bN^{+} \bM^{\top}\bx$ is the nearest point to $\bX$ in $\mathrm{span}(\bM)$, we have 
\begin{align} \label{example:3.2:eq1}
\|\bx-\bM \bN^{+} \bM^{\top}\bx\| \leq \|\bx-\bM_1\|< \infty.
\end{align}
Further, the angle between $\bx$ and $\bM_1$ is 
\begin{align*}
\cos^{-1}(\frac{\bx^{\top}\bM_1}{\|\bx\|\|\bM_1\|})\to 0 \text{ as } n\to\infty.
\end{align*}
Consequently, Condition \ref{cond:wlln:with:covariates} does not hold. Note that \eqref{example:3.2:eq1} implies that ${\bx}^{\top}(\pmb{I}- \bM \bN^{+} \bM^{\top})\bx$ does not diverge; however $\bx^{\top}\bx$ does diverge. Consequently, the second smallest eigenvalue of  $\pmb{H}^{\top} \pmb{H}$ is finite as $n\to\infty$.
\end{example}

\begin{example} \label{example_4.3}
Consider a PCG $(\mathcal{G},\mathcal{Y},\mathcal{X})$ with $K=2$ and $p=1$. Clearly, Condition \ref{Con(AlgConnec)} is satisfied. The first and second columns of the matrix $\bM$ are $\b1$ and $-\b1$ respectively. We choose $\bx$ such that $x_i=-1$ if $i=j^2$ for some $j\in\mathbb{N}$ and $x_i=1$ otherwise. Clearly, $\mathrm{rank}(\pmb{H})=2$, $\mathrm{span}(\bM)=\mathrm{span}(\b1)$ and $\bx^{\top}\bx\to\infty$ as $n\to\infty$. Next, the orthogonal projection of $\bx$ onto $\mathrm{span}(\bM)$ is
	$$
	\frac{\b1^{\top}\bx}{\|\b1\|^2}\b1= \frac{n-2m}{n}\b1, \text{ where } m = \lfloor\sqrt{n}\rfloor.
	$$
	Therefore,
	$$
	\cos \phi_n = \frac{n-2m}{n}\to 1 \text{ as }n\to\infty.
	$$
	Consequently, $\phi_n\to 0$ as $n\to\infty$. However, the second smallest eigenvalue of  $\pmb{H}^{\top}\pmb{H}$ is
	$$
	\lambda_2= \frac{3n}{2}-\frac{1}{2} \sqrt{9n^2-32nm+32m^2}.
	$$
	With some algebra, it can be shown that $\lambda_2\to\infty$ as $n\to\infty$ so Condition \ref{cond:wlln:with:covariates} is not necessary for the divergence of the second smallest eigenvalue of $\pmb{H}$.
\end{example}

Condition \ref{cond:wlln:with:covariates} is trivially satisfied when the spaces spanned by the columns of $\bM$ and $\bX$ are orthogonal. In such situations the matrix on the LHS of \eqref{eq:18} reduced to $\mathrm{BlockDiag}(\bN,\bX^{\top}\bX)$ and the conclusion of Lemma \ref{lemma:cod:with:cov:wlln} is immediate. Interestingly, it can also be shown that if $\bX^{\top}\bM \pmb{\xi}=\pmb{0}$ where $\pmb{\xi}$ is the Fiedler vector of the graph $\mathcal{G}$ (i.e., $\pmb{\xi}$ is the eigenvector associated with the second smallest eigenvalue of $\bN$) then Condition \ref{cond:wlln:with:covariates} holds as well. Finally, we note that Condition \ref{cond:wlln:with:covariates} is unlikely not to hold in practice. For if it were so at least one column of $\bX$, or a linear combination thereof, would need to at least approximately satisfy the conditions of Proposition \ref{span:M:condition}. This seems very unlikely. Further support for the assertion that Condition \ref{cond:wlln:with:covariates} will hold in applications is given in Lemma \ref{proposition:random:cov:properties} appearing below.   

\begin{lemma} \label{proposition:random:cov:properties}
If the columns of $\bX$ are independent and the components in each column of $\bX$ are IID with zero mean and finite variance, then as $n\to\infty$ the smallest eigenvalue of $\bX^{\top}\bX$ diverges and $\phi_n\to \pi/2$ with probability one.
\end{lemma}

Thus, if the covariates are centered and random samples, then condition \ref{cond:wlln:with:covariates} holds with probability one.  

\begin{condition} \label{cond:clt:with:covariates}
The smallest eigenvalue of $\lim_n\bX^{\top}\bX/n$ is positive and for all large $n$ we have $\phi_n \ge \phi$ for some $\phi\in(0,\pi/2)$.
\end{condition}

If Condition \ref{cond:clt:with:covariates} holds then all eigenvalues of $\bX^{\top}\bX$ grow at the same rate. Clearly Condition \ref{cond:clt:with:covariates}
implies Condition \ref{cond:wlln:with:covariates} and is an analogue of Condition \ref{Con(Rate+CLT)} for PCGs without covariates. 

\begin{lemma} \label{lemma:cod:with:cov:clt}
Let $\pmb{\Sigma}=\lim_n\pmb{H}^{\top}\pmb{H}/n$. If Conditions \ref{Con(Rate+CLT)} and \ref{cond:clt:with:covariates} hold then $\mathrm{rank}(\pmb{\Sigma})=K+p-1$.
\end{lemma}

Consequently the conditions in Lemma \ref{lemma:cod:with:cov:clt} guarantee the existence of an asymptotic variance of the corresponding LSEs. Note that $\mathrm{rank}(\pmb{\Sigma})=K+p-1$ implies that the second smallest eigenvalue of $\pmb{H}^{\top}\pmb{H}$ diverges, however the converse is not true, cf., Example \ref{example_4.3}. 

\begin{remark} \label{two:set:covariates}
Condition \ref{cond:wlln:with:covariates} essentially generalizes the usual condition for consistency in standard regression models of the form $\bY=\bX\pmb{\beta} +\pmb{\epsilon}$, i.e., that $\lambda_{1}(\bX^{\top}\bX) \to \infty$ to the case of partitioned matrices. That is if $\bX=(\bX_1,\bX_2)$ then consistent estimation is guaranteed if $\lambda_{1}(\bX_i^{\top}\bX_i) \to \infty$ for $i=1,2$ and that $\measuredangle(\bX_1,\bX_2) \ge \phi>0$ for all large $n$. Similarly, Condition \ref{cond:clt:with:covariates} extends the standard condition for asymptotic normality of $\widehat{\pmb{\beta}}$ in regression models, i.e., that $\lambda_{1}(\bX^{\top}\bX/n) > 0$ to the case of partitioned matrices. That is if $\bX=(\bX_1,\bX_2)$ then the asymptotic normality of the LSE is ensured if $\lambda_{1}(\bX_i^{\top}\bX_i/n) > 0$ for $i=1,2$ and $\measuredangle(\bX_1,\bX_2) \ge \phi>0$ for all large $n$.      
\end{remark}

\begin{condition} \label{cond:with:cov:clt+hajek+sidak}
Let
$$ \frac{\max_{1 \le i \le n} \pmb{h}_i \pmb{h}_i^{\top}}{\lambda_2(\pmb{H}^{\top}\pmb{H})} \to 0 \text{ as } n\to\infty $$
where $\pmb{h}_i$ is the $i^{th}$ row of  $\pmb{H}$.  
\end{condition}

Condition \ref{cond:with:cov:clt+hajek+sidak} is essentially a  Hajek--Sidak type condition. It ensures that the contribution of any finite set of observations is negligible as $n\to\infty$. 

\medskip

The following theorem extends Theorems \ref{Thm-graph.WLLN}, \ref{Thm-LST} and \ref{Thm-AS}. 

\begin{theorem}
\label{graph-GLM:wlln+clt}
We have: 
\begin{enumerate}
\item Under the assumptions of Theorem \ref{Thm-Est.Covariates} the LSE $(\widehat{\bmu}, \widehat{\pmb{\beta}})$ is unbiased.
    
\item If the conditions of Lemma \ref{lemma:cod:with:cov:wlln} hold then $\widehat{\pmb{\beta}}\to\pmb{\beta}$ and $\widehat{\bmu}\to\bmu$ in probability as $n\to\infty$.
        
\item If $\widehat{\bmu}\to\bmu$ in probability, then the derived rank vector $\pmb{r}(\widehat{\bmu})$ converges in probability to $\pmb{r}({\bmu})$ as $n\to\infty$.
         
\item Assume that the conditions of Lemma \ref{lemma:cod:with:cov:clt} hold together with Condition \ref{cond:with:cov:clt+hajek+sidak}. Then, 
\begin{align*}
\sqrt{n}(\widehat{\pmb{\theta}} - \pmb{\theta}) \Rightarrow \mathcal{N}( \pmb{0},\sigma^2\, \pmb{\Sigma}^{+} ),\text{ where } \widehat{\pmb{\theta}} = (\widehat{\bmu}^{\top},\, \widehat{\pmb{\beta}}^{\top} )^{\top}.
\end{align*}
\end{enumerate}
\end{theorem}

Theorem \ref{graph-GLM:wlln+clt} shows that under mild conditions the LSE and the derived rank vector are consistent and that the LSE is asymptotically normally distributed for graph--LMs with covariates. The asymptotic variance is easily estimated so inference on $(\widehat{\bmu}, \widehat{\pmb{\beta}})$ is readily carried out. 

\medskip
\medskip

Next we examine the situation where the true model is \eqref{model:with:covariate} but a model without covariates, i.e., model \eqref{model.Y_ijl}, is fit to the data. Note first that the LSE given in \eqref{def.mu.hat} is generally biased and not consistent if the columns of $\bM$ and $\bX$ are not orthogonal (cf., Zimmerman 2020). However, orthogonality among $\bM$ and $\bX$ may arise. For example, consider a round--robin tournament where the covariate indicates whether the game was played on the home court or not, cf. Harville (2003). If the tournament is completely balanced, i.e., teams $i$ and $j$ compete on both home courts the same number of times, then it is easy to verify that columns of $\bM$ and $\bX$ are orthogonal. Therefore the merits can be unbiasedly and consistently estimated. Furthermore, it turns out that the LSE for the model without covariates is often consistent when the true model does include covariates. This, and related themes, are explored in the next theorem for which we introduce the following notation. Let $f_t$, indexed by the parameters $(\bmu_t,\pmb{\beta})$, denote the density associated with the true model, i.e., the model with covariates. Similarly let $f_m$, indexed only by the parameter $\bmu_m$, denote the density of the misspecified model, i.e., the model without covariates. Let 
\begin{align*}
\mathrm{AKL}(t,m) =&~\frac{1}{n} \mathbb{E}_{f_t} (\log(\frac{f_t(\bY)}{f_m(\bY)})|\bX).
\end{align*}
denote the conditional average Kullback--Leibler (AKL) divergence between the two models. 

\begin{theorem} \label{misspecified:model:normal:errors:thm}
We have the following: 
\begin{enumerate}
\item If the errors in \eqref{model.Y_ijl} are IID $\mathcal{N}(0,\sigma^2)$ RVs and $\bX^{\top}\bM/n =o(1)$ as $n\to\infty$ then $\mathrm{AKL}$ is minimized with respect to $\bmu_m$ when $\bmu_m=\bmu_t+o(1)$.  
\item Suppose that the conditions in Theorem \ref{graph-GLM:wlln+clt} hold. Suppose further that the vector of merits is estimated by $\widehat{\bmu}_m = \bN^{+}\bS$ and that $ \pmb{b}=\lim \bN^{+} \bM^{\top} \bX\pmb{\beta}$ exists. If so,  
\begin{equation*}
\sqrt{n}( \widehat{\bmu}_{m}-\bmu_t-\pmb{b})
\Rightarrow \mathcal{N}_{K}( \pmb{0},\sigma ^{2}\pmb{\Theta}^{+}) .
\end{equation*}
Moreover if $\bM^{\top}\bX/\sqrt{n}=o(1)$ then $\sqrt{n}\,\pmb{b}=o(1)$ and $\sqrt{n}(\widehat{\bmu}_m-{\bmu}_t)\Rightarrow \mathcal{N}_{K}( \pmb{0},\sigma ^{2}\pmb{\Theta}^{+})$.
\item Suppose that the columns of $\bX$ are comprised of IID zero mean finite variance RVs. Then the LSE \eqref{def.mu.hat.with.cov} satisfies
\begin{align*}
\widehat{\bmu}_t &~= \bN^{+}\bS +~ \pmb{\xi}_{1n},\\
\widehat{\pmb\beta}&~= (\bX^{\top}\bX)^{-1}\bX^{\top}\bY +~ \pmb{\xi}_{2n},
\end{align*}
where $\sqrt{n}\,\pmb{\xi}_{1n}=O_p(1)$ and $\sqrt{n}\,\pmb{\xi}_{2n}=O_p(1)$. Furthermore $(\widehat{\bmu}_t,\widehat{\pmb\beta})^{\top}$ are asymptotically normally distributed as specified in Theorem \ref{graph-GLM:wlln+clt} with $\pmb{\Sigma}= \mathrm{BlockDiag}(\pmb{\Theta},\,\pmb\Sigma_{\bX})$ where $\pmb\Sigma_{\bX}$ is variance of covariate.
\end{enumerate}
\end{theorem}

In the proof of Part $1$ of Theorem \ref{misspecified:model:normal:errors:thm} an expression for $\mathrm{AKL}(t,m)$, under normality, is found and shown to be minimized at $\bmu_m=\bmu_t+o(1)$ when $n^{-1}\bX^{\top}\bM =o(1)$, i.e., when $\bX$ and $\bM$ are ``nearly'' orthogonal. This is a distributional result which suggests that under near orthogonality the merit parameters in the true and misspecified model have asymptotically similar values. Implications for estimation are explored in Parts $2$ and $3$ of Theorem \ref{misspecified:model:normal:errors:thm}. Part $2$ of Theorem \ref{misspecified:model:normal:errors:thm} explores the limiting distribution of the LSE $\widehat{\bmu}_{m}$ when the true model holds. In particular the form of the bias term $\pmb{b}$ is found. The bias can be shown to vanish when $\bM^{\top}\bX/n=o(1)$. However removing the bias from the asymptotic distribution requires that $\bM^{\top}\bX/\sqrt{n}=o(1)$ which is a stronger form of near orthogonality. See section 5.2.2 in Claeskens and Hjort (2008) for related analyses. Finally, in Part $3$ of Theorem \ref{misspecified:model:normal:errors:thm} a stochastic representation of the LSE when the covariates are centered RVs is found. In essence it is shown that random covariates enable decoupling the estimation of the merits and regression coefficients. This phenomenon is a consequence of the geometry of high--dimensional spaces (Vershynin 2020).

\section{Infinite number of items}

So far we have treated $K$ as fixed. In this section graph-LMs are extended to the case where $K\rightarrow \infty$. For simplicity we will first assume that no covariates are present. In this setting one can imagine a sequence of nested comparison graphs $(\mathcal{G}_i,\mathcal{Y}_i)$, $i=1,2,\ldots$. A common practice in such settings is to choose a reference item, denoted by $v$, and set $\mu_{v}=0$. In this way the merits of all other items are fixed. However, we will continue with the assumption that $\sum_{i=1}^{K}\mu _{i}=0$ which is fully justified by Corollary \ref{Cor-EstDiff} but implies that the value of any merit may vary with $K$.  

There are numerous comparison regimes to consider when $K\rightarrow \infty$. Although it is possible that some $n_{ij} \rightarrow \infty$ when $K\rightarrow \infty$ such scenarios are less likely; moreover, as will become evident, such situations are much easier to analyze. Therefore in this section we shall assume that $n_{ij} \le 1$ for all pairs $(i,j)\in \{1,\ldots,K\}^2$. The following theorem describes the behavior of the vector of estimated merits in a round--robin tournament in which the number of items grows to infinity.

\begin{theorem}
\label{Thm-LST.LargeK}
Consider the model \eqref{model.Y_ijl}. Let $n_{ij}=1$ for all $1\leq i\neq j\leq K$. If $K \rightarrow \infty$ then for every finite subset of indices $\{i_1,\ldots,i_r\}\subset \{1,2,\ldots,K\}$ the elements of the LSE satisfy
\begin{equation}
\sqrt{K}((\widehat{\mu }_{i_{1}},\ldots ,\widehat{\mu }_{i_{r}})-(\mu_{i_{1}},\ldots ,\mu _{i_{r}}))^\top\Rightarrow \mathcal{N}_{r}\left(\pmb{0},\sigma^2 \pmb{I}\right).  \label{RR.CLT}
\end{equation}
Additionally, if the errors are subgaussian then the maximum estimation error satisfies  
\begin{equation}
\mathbb{E}\{\max_{1\leq i\leq K}| \widehat{\mu }_{i}-\mu _{i}|\}=O(\sqrt{\frac{{\log (K)}}{K}}).  \label{max.MuDiff}
\end{equation}
\end{theorem}

Equation \eqref{RR.CLT} shows that the CLT holds for any finite collection of items. Note that the independence among the estimated scores is a consequence of the fact that $n_{ij}/n_{i}\rightarrow0$ for all $i$ and $j$. Moreover, for convenience, the scaling factor in \eqref{RR.CLT} is $\sqrt{K}$, the square root of the number of items and not the total number of comparisons as in Theorems \ref{Thm-LST} and \ref{graph-GLM:wlln+clt}. Since $n(K)=K(K-1)/2$, \eqref{RR.CLT} holds also when scaled by $(2n)^{1/4}$ instead of $K^{1/2}$. Further note that \eqref{RR.CLT} implies that the RVs $\sqrt{K}(\widehat{\mu}_i-\mu_i)$'s behave as a sequence of independent zero mean normal RVs with variance $\sigma^2$, the largest of which is of order $O(\sqrt{\log K})$. Thus \eqref{max.MuDiff} holds and can be understood as the rate at which the largest $| \widehat{\mu }_{i}-\mu _{i}|$ converges to $0$. Theorem \ref{Thm-AS} can also be used in conjunction with Theorem \ref{Thm-LST.LargeK} to show that the probability of correct ranking converges exponentially fast for any finite subset of indices. 

In a complete graph there are $K(K-1)/2$ distinct comparisons. Next, we consider sparse graphs in which the number of comparisons is much smaller than $O(K^2)$. One way of doing so is by letting the topology of the comparison graph be determined randomly. To this affect we employ Erd\H{o}s--R\'enyi random graphs in which the $(i,j)^{th}$ comparison (edge) is present with probability $p_K$, independently from every other edge (Janson et al., 2000). Concretely, consider the PCG $(\mathcal{G},\mathcal{Y})$ where the graph $\mathcal{G}$ is a Erd\H{o}s--R\'enyi graph with parameters $(K,p_K)$. Its Laplacian is a random matrix with elements
{ 
\begin{equation*}
N_{ij}=\left\{ 
\begin{array}{ccc}
\sum_{j}B_{ij} & \text{if} & i=j \\ 
-B_{ij} & \text{if} & i\neq j,
\end{array}
\right.
\end{equation*}}
where $B_{ij}$ are IID $\mathrm{Ber}(p_K)$ RVs. The number of paired comparisons on an Erd\H{o}s--R\'enyi graph is of the order $O(p_{K}K^2)$ and it is well known that such graphs are connected with probability one, as $K\to\infty$, if and only if $p_{K} \ge (1+c)\log(K)/K$ for some $c>0$ (Erdos and Renyi 1961). Clearly, the merits can not be estimated unless the comparison graph is connected.  

\begin{condition} \label{cond:inf:s1}
For all large $K$ assume that $p_K\geq (1+c)\log(K)/K$ for some $c>0$.
\end{condition}

Then,

{ 
\begin{theorem} \label{inf:lse:properties}
Consider model \eqref{model.Y_ijl}. If Condition \ref{cond:inf:s1} holds then for any finite set of indices $\{i_1,\ldots,i_r\}$ we have

\begin{equation} \label{random:clt}
{\sqrt{Kp_K}}((\widehat{\mu}_{i_1},\ldots ,\widehat{\mu}_{i_l})-(\mu_{i_{1}},\ldots ,\mu_{i_{r}}))^\top\Rightarrow \mathcal{N}_{r}\left( \pmb{0},\sigma^2 \pmb{I}\right).  
\end{equation}
Moreover, if the errors are subgaussian then
\begin{align} \label{inf:ER:wlln}
\mathbb{E}(\max_{i=1,\ldots ,K}| \widehat{\mu }_{i}-\mu _{i}|) =O(\sqrt{\frac{{\log (K)}}{Kp_K}}).
\end{align}
\end{theorem}

Theorem \ref{inf:lse:properties} shows that the LSE is consistent and asymptotically normal provided the Erd\H{o}s--R\'enyi graph is connected as $K\to\infty$. The proof is based on the asymptotic equivalence of the LSE $\widehat{\bmu}$ and a randomized version of the so--called row--sum method (Huber, 1963) which we denote by $\tilde{\bmu}$ and  whose $i^{th}$ element is given by
$$\tilde{\mu}_{i}=\frac{S_i}{N_{ii}}.$$ 
Here $S_i$ is the random sum $\sum_{j\neq i}Y_{ij}B_{ij}$ satisfying 
\begin{equation*}
\mathbb{E}(S_{i}) = \mathbb{E}(\sum_{j\neq i}Y_{ij}B_{ij}) =
\sum_{j\neq i}\mathbb{E}(Y_{ij}B_{ij}) = 
\sum_{j\neq i}p_K(\mu_i-\mu_j) = Kp_K\mu_i, 
\end{equation*}
and $N_{ii}$ is a binomial RV with parameters $(K-1,p_K)$. Therefore, it is clear that both $S_i/(Kp_K)$ and $\tilde{\mu}_{i}$ are consistent for $\mu_i$. An advantage of $\tilde{\bmu}$ is that its computation does not involve the inversion of the Laplacian.   

We conclude this section by demonstrating that the LSE is uniformly consistent even on sparse graphs. Markov's inequality, Jensen's inequality and Equation \eqref{inf:ER:wlln} together imply that
\begin{align*}
\mathbb{P}(\|\widehat{\bmu}-\bmu\|_{\infty} \leq \epsilon) \geq 
1- \frac{\mathbb{E}\|\widehat{\bmu}-\bmu\|_{\infty}^r}{\epsilon^r}
\geq 
1- \frac{(\mathbb{E}\|\widehat{\bmu}-\bmu\|_{\infty})^r}{\epsilon^r}
= 1- {O(({\frac{\log K}{\epsilon^2Kp_K}})^{r/2})}.
\end{align*}
Choosing $\epsilon = {{\log K}/{Kp_K}}^{1/r}$ and $p_K= {(\log K)^{1+c}}/{K}$ the display above may be rewritten as
\begin{align} \label{prob:bound}
\mathbb{P}(\|\widehat{\bmu}-\bmu\|_{\infty} \leq ({\frac{\log K}{Kp_K}})^{1/r})\geq 
1- O((\log K)^{cr(1/r-1/2)}). 
\end{align}
Thus, for any $c>0$ and $r>2$, the RHS of \eqref{prob:bound} converges to $1$, which implies in turn that the LSE in sparse random comparison graphs is uniformly consistent. We believe that the rate of convergence in \eqref{prob:bound} is not the best possible. 
}

\begin{remark}
Theorem \ref{inf:lse:properties} can be extended to situations where $\mathcal{G}$ is a random graph with varying edge probabilities (cf. Janson et al., 2000) and to the general case where $n_{ij}$ are arbitrary independent RVs. Such a graph may arise when the degree of an item depends on its merit, i.e., more popular items are compared more than less popular ones. Examples include IMDB, a movie rating website and Amazon, an online retailer. In such a case, we believe that consistency of the LSE will hold if the graph is connected and the degrees of all items increase to infinity as $K\to\infty$. However, we do not attempt to establish such a result in the current article.
\end{remark}

\begin{remark}
Finally we note that Theorems \ref{Thm-LST.LargeK} and \ref{inf:lse:properties} can be readily adapted to situations where covariates are present. The sampling procedure generating Erd\H{o}s--R\'enyi comparison graphs guarantees that the covariates are random variables. The high dimensional geometry implies that the columns of $\bX$ are orthogonal to columns of $\bM$ with high probability. Therefore, the conclusions of Part $3$ of Theorem \ref{misspecified:model:normal:errors:thm} apply and consequently the vector of merits can be consistently estimated without considering the covariates. 
\end{remark}

\section{Numerical experiments} \label{numerical.expriments}
\FloatBarrier

We evaluate the performance of the LSEs on graph-LMs by simulation in three different settings: 

\begin{enumerate}

\item \underline{\bf{Simple graph-LMs:}} Motivated by Theorems \ref{Thm-graph.WLLN} and \ref{Thm-AS} we start our numerical investigation by studying the rate at which the LSE, as well as the derived ranks converge to their true values. 

\item \underline{\bf{Graph--LMs with covariates:}} Next, we investigate the properties of LSEs in PCGs with covariates. Our focus is on assessing the influence of covariates on the accuracy of the estimators described in Theorem \ref{misspecified:model:normal:errors:thm}.  

\item \underline{\bf{Large graph--LMs:}} Guided by Theorems \ref{Thm-LST.LargeK} and \ref{inf:lse:properties}, we investigate the performance of the LSE in large and sparse PCGs. 

\end{enumerate}

\subsection{Simple graph-LMs}
\label{example_without_covariates_K=8}
\FloatBarrier

In our simulation study of model \eqref{model.Y_ijl} we assume a complete graph with $K=8$ and $\bmu=\bmu_0\times 10^{-\gamma}$ where $\bmu_0 = (-7,-5,-3,-1,1,3,5,7)$ and $\gamma\in\{0,1,2\}$. The number of paired comparisons $n_{ij}=m\in\{10,20,\ldots,1000 \}$. We experimented with three error distributions, namely $\mathcal{T}_2$, a t--distributed RV with two degrees of freedom, $\mathcal{T}_3/\sqrt{3}$ and $\mathcal{N}(0,1)$. Note that a t--distributed RV with two degrees of freedom has a mean $0$ but does not have a variance. The variance of the RV $\mathcal{T}_3/\sqrt{3}$ is one. In Figure \ref{fig:mse_normal_rvs_without_cov} we plot the empirical mean squared error (MSE), i.e., $\sum_{i=1}^{R}(\|\widehat\bmu_{i,n}-\bmu_0\|_2)/R$, against $m$, where $\widehat\bmu_{i,n}$ is the LSE computed at the $i^{th}$ simulation run and $R=1000$ is the number of runs. In Figure \ref{fig:corrrect_ranking_without_cov} we plot the probability of correct ranking, i.e., $\mathbb{P}\left( \pmb{r}( \widehat{\bmu}_{n}) = \pmb{r}\left( \bmu\right) \right)$, against $m$ for three values of $\gamma$ assuming the errors are $\mathcal{N}(0,1)$. 

\medskip
\centerline{Figure \ref{fig:MSE_without_cov} Comes Here}
\medskip

It is clear from Figure \ref{fig:mse_normal_rvs_without_cov} that in all three cases the MSE decreases as $m$ increases. The MSEs when the errors are either $\mathcal{T}_3/\sqrt{3}$ or $\mathcal{N}(0,1)$ seem similar. Nevertheless the MSEs associated with $\mathcal{N}(0,1)$ are always smaller than those associated with $\mathcal{T}_3/\sqrt{3}$ errors; the difference among the aforementioned MSEs is most pronounced for small values of $m$. In both cases the MSE decays at the expected rate. The MSEs associated with $\mathcal{T}_2$ errors behaves quite differently. They seem to decrease at a much slower rate and the empirical MSE curve is not smooth. This is caused by occasional very large errors. In fact since $\mathcal{T}_2$ RVs do not have a finite variance the theoretical MSE, i.e., $\mathbb{E}(\|\widehat\bmu_{n}-\bmu_0\|_2)$ does not exist. This means that once in a while the value of the empirical MSE will shoot up. Nevertheless by Theorem \ref{Thm-graph.SLLN} the LSE is consistent even in this case. Figure \ref{fig:corrrect_ranking_without_cov} provides a complementary view. Note that the probability of correct ranking increase with $m$. When the spacing's among the means is large this probability is very close to unity. Clearly, it is difficult to rank correctly when the spacing's are small. 

The simulation experiments described above were also carried out for path and cycle graphs with qualitatively similar results.

\subsection{Graph--LMs with covariates}
%\FloatBarrier

We assume the same graph topology, $K$, $\bmu$ and $n_{ij}$'s used in Section \ref{example_without_covariates_K=8}. We set $\pmb{\beta}=(1,1)^{\top}$ and let the bivariate covariate be independent Rademacher RVs, i.e., $\mathbb{P}(X_i=\pm 1)=1/2$ for $i=1,2$. Data were simulated from model \eqref{model:with:covariate} assuming the error distribution is $\mathcal{N}(0,1)$.

We fit the model with and without covariates. When fitting the model with covariates we observe, as expected from Theorem \ref{graph-GLM:wlln+clt} that the LSE is consistent and converges at the expected rate. Next we fit the model ignoring the covariates. Recall that by Theorem \ref{misspecified:model:normal:errors:thm} the LSE for the vector of merits is also consistent when the covariates are ignored. 

Probabilities of correct ranking are presented in Figure \ref{fig:rank_with_cov_est_cov1}. Observe that the derived ranks are estimated correctly with high probability even after ignoring covariates, although less efficiently. For example, when $m=1000$ and $\bmu=10^{-2}\times \bmu_0$ the probability of correct ranking when covariates are included in the model is close to $0.4$ and only $0.1$ when the covariates are omitted.  

\medskip
\centerline{Figure \ref{fig:rank_with_cov_est_cov1} Comes Here}
\medskip

The simulation experiments described above were also carried out for other distributions of covariates, e.g., the $\mathcal{N}_2(\bzero, \pmb{I})$ distribution, as well as correlated covariates with similar qualitative results.

\subsection{Large graph-LMs}

We consider model \eqref{model.Y_ijl} with $K\in\{100,150,\ldots,500 \}$. We set $\bmu= (-(K-1), -(K-3)\ldots, -1, 1,3,\ldots, K-1)$ and the errors are IID $\mathcal{N}(0,1)$ RVs. The $n_{ij}$'s are IID Bernoulli RVs with parameter $p\in\{1,\, 0.5,\, K^{-1}(\log K)^3,\,\sqrt{K^{-1}(\log K)^3} \}$ for $1\leq i<j\leq K$. Clearly when $p=1$ we have $n_{ij}=1$ for all $1\leq i,j\leq K$, i.e., the graph is complete graph. The other values of $p$ are the same as those in Han et al. (2022)

The value of $\max_{1\leq i\leq K}|\widehat\mu_i-\mu_i|$ against $K$ is plotted, for the aforementioned values of $p$, in Figure \ref{fig:large_dim_example}. Clearly for any fixed $p$ the estimation error is decreasing in $K$. Similarly for every fixed $K$ the estimation error decreases as $p$ increases. We also experimented with $p=O({K^{-1}(\log K)^2})$, a value which is slightly above the theoretical threshold required for convergence, cf. Theorem \ref{inf:lse:properties}. We have found the convergence of the LSE for this value of $p$ to be very slow.   

\medskip
\centerline{Figure \ref{fig:large_dim_example} Comes Here}
\medskip

\section{Illustrative example}

Ranking sports teams is of general public interest and an active application area for ranking methodologies. For illuminating details see the book by Langville and Meyer (2012).  

Here, we illustrate the proposed methodology by consider the ranking of NBA teams in 2022–2023 season. The data were obtained from the Basketball Reference Website (\url{https://www.basketball-reference.com/leagues/}) and is comprised of $1320$ regular season games played by the $30$ teams comprising the league. In addition to the score and the point--spread, our outcome variable, we record whether the game was played at home or away, and the ranking of the teams in the past two seasons, i.e., seasons 2020--21 and 2021--22. We fit the model
\begin{equation} \label{model.example}
Y_{ijk}=\mu_i-\mu_j+\beta_1(x_{ik,1}-x_{jk,1})+\beta_2(x_{i2}-x_{j2}) +\beta_3(x_{i3}-x_{j3})+\epsilon_{ijk}.   
\end{equation}
Here $x_{ik,1}=1$ for a home game, $x_{ik,1}=0$ otherwise, $x_{i2}$ and $x_{i3}$ denote the expanded standings of the team in seasons 2020--21 and 2021--22 respectively. Note that the first regressor $x_{ik,1}-x_{jk,1}\in \{-1,1\}$ the second regressor $x_{i2}-x_{j2} \in \{-29,\dots,29\}$, as is the third regressor. The parameter $\beta_1$ quantifies the effect of playing at home, whereas $\beta_2$ and $\beta_3$ quantify historical effects. 

Table \ref{table:ranks} list the rankings derived by fitting four nested submodels of \eqref{model.example} referred to as Models I--IV. Model I, corresponds to a model without any covariates, subsequent models incorporate the covariates, as they appear in \eqref{model.example}, culminating in Model IV which includes all three covariates. 
Teams are listed from left to right according to their expanded standing for the 2022--23 season. Observe that Models I and II yield identical rankings, whereas Models III and IV also yield similar rankings which differ quite substantially from those given by the models without covariates. 

\medskip
\centerline{Table \ref{table:ranks} Comes Here}
\medskip

Table \ref{table.rank.distance} provides the Cayley distance among the different models. We note that the full model is highly significant so there is clear evidence that Model IV should be adopted. The fact that rankings derived from Model I and Model IV are different indicates that $\bM$ and $\bX$ are not orthogonal and the Conditions in Theorem \ref{misspecified:model:normal:errors:thm} do not hold in this example and omitting the covariates may result in an improper ranking.   

\medskip
\centerline{Table \ref{table.rank.distance} Comes Here}
\medskip

The accuracy and variability of the estimated ranks were assessed by nonparametric bootstrap. Concretely, rankings were estimated in $200$ bootstrap samples and the resulting boxplots of ranks are given in Figure \ref{fig:real_eg}.  

\medskip
\centerline{Figure \ref{fig:real_eg} here.}
\medskip

Figure \ref{fig:real_eg} reveals several interesting features. First, variability is not even among all teams. The interquartile interval, that is, $(Q_1,Q_3)$, for the rank of the Boston Celtics is $(1,2)$ in Model I and $(2,7)$ in Model IV. This is a team whose ranking exhibits little variability. In contrast, the corresponding intervals for the Oklahoma City Thunders are $(11,19)$ in Model I and $(3,6)$ in Model IV. It is also clear that overall variability is high for middle ranked teams and less so for the best and worst teams. We also observe that the variability in Model IV is somewhat smaller than the variability in Model I although the difference is not very large. 

\section{Discussion} \label{discussion}

Despite the simplicity of models for cardinal PCDs and the extensive literature on the LSE, the statistical properties of the LSE have not been been well studied. Consequently, and surprising, methods for conducting inference were not sufficiently grounded. This paper fills this gap by developing a graph--based, rigorous statistical theory, that enables the conduct of inference for such models.  

Our investigation starts with model \eqref{model.Y_ijl}. We find explicit necessary and sufficient graph--theoretic conditions that guarantee consistency, {  asymptotic normality and exponential rates of convergence for rank estimators}. Next, we analyzed PCGs with covariates. To the best of our knowledge this paper is the first to do so in any generality. Conditions for consistency and asymptotic normality are provided. It is also shown that in many situations arising in practice the merits can be estimated consistently even if covariates are omitted from the model. This surprising result is a consequence of the geometry of high--dimensional spaces. Note also Remark \ref{two:set:covariates} where the connection between PCGs with covariates and standard regression problems with two disjoint sets of explanatory variables is made. Another extension provided is to situations in which the number of items compared is large, i.e., $K\to\infty$ and the number of paired comparisons between any two items is small, i.e., $0$ or $1$. 

\bigskip

{  This paper focuses on graph--LMs, i.e., graphical linear models for cardinal outcomes. In the following we contrast our results with those available in the literature on binary choice models with a focus on the Bradley--Terry (BT) model given in \eqref{Eq.BT.model}. 

\begin{enumerate}

\item Existence and uniqueness of estimators (Theorem \ref{Thm-UniqueE}): The merit parameters in BT models are typically estimated by maximizing the likelihood subject to the constraint that $\mu_1=1$. Other constraints such as $\sum_{i=1}^{K}\mu_i=1$ are also commonly employed (Firth and Turner 2012). As noted the LSE can also be viewed as an MLE if the errors in \eqref{model.Y_ijl} are normally distributed. The LSE for cardinal PCD requires that the comparison graph is connected, cf. Theorem \ref{Thm-UniqueE}, BT models require, in addition, that for every partition of the items into two nonempty sets, an item in the second set has been preferred over an item in the first at least once (Yan, 2016). This condition is necessary for the existence of the MLE and is, essentially, an "overlapping" condition for logistic regression (Speckman et al. 2009). For finite sample sizes the aforementioned (additional) condition may not hold if $\lambda_i \gg \lambda_j$ for all $i\in\mathcal{V}'$ and $j\in\mathcal{V}\setminus\mathcal{V}'$ where $\mathcal{V}'$ is some subset of $\mathcal{V}$.    

\medskip
\item Large sample theory for finite graphs (Theorems \ref{Thm-graph.WLLN}, \ref{Thm-graph.SLLN}, \ref{Thm-LST}, \ref{theorem clt with 3.1} and \ref{Thm-AS}): The large sample properties of the merit estimators for BT models have been analyzed in the literature, e.g., Bradley and Gart (1962). However, in that paper, and to the best of our knowledge all others, it assumed, often only implicitly, that all $n_{ij}$ increase to $\infty$ at the same rate whenever $(i,j)\in\mathcal{E}$. In some papers only the assumption that $n_i=\theta_i n$ is mentioned. Arguing as we have it is not difficult to establish that that Conditions \ref{Con(AlgConnec)} and \ref{Con(Rate+CLT)} are necessary and sufficient for consistency and asymptotic normality in BT models. Clearly the formula for the asymptotic variance under BT models is more complicated than \eqref{var.mu.hat}, see Simons and Yao (1999). 

Although, several authors (e.g., Shah et al. 2015) have noted that bounds on the mean square error of the MLE in BT models depend on the algebraic connectivity. They did not, connect their bounds with the necessary graph based conditions. We are also unaware of the equivalents of Theorems \ref{Thm-graph.SLLN}, \ref{Thm-LST} and \ref{theorem clt with 3.1} for BT models. We believe such extensions are possible.  
\medskip

\item Theory for models with covariates (Theorems \ref{Thm-Est.Covariates}, \ref{graph-GLM:wlln+clt} and \ref{misspecified:model:normal:errors:thm}): BT models with covariates have also appeared in the literature, see Stern (2011)  and more recently Fan et al. (2022). However, we are not aware of any authors who have formalized Conditions such as  \ref{condition:covariate:existence}, \ref{cond:wlln:with:covariates}, \ref{cond:clt:with:covariates} and \ref{cond:with:cov:clt+hajek+sidak} which are necessary for  guarantying consistency and asymptotic normality. We strongly believe that the conditions formalized in Section 4 are also necessary for limit theorems for BT models with covariates. We note that Fan et al. (2022) studied large sparse BT models with covariates. The principal difference between their approach and ours is that Fan et al. (2022) assume orthogonality in the parameter space. Our assumptions do not involve the parameters, only the design matrix $(\bM,\bX)$ and are therefore more natural to pose, and as importantly, to verify. We believe that our formulation leads to a crisp analysis and a clear view of the role of covariates in all models for PCD. 

\item Models with infinite number of items (Theorems, \ref{Thm-LST.LargeK} and \ref{inf:lse:properties}): BT models in which $K\to\infty$ have appeared in the literature. Simons and Yao (1999) were the first to analyze such models assuming $n_{ij}=1$ for all $1 \le i \neq j \le K$ whereas Han et al. (2022), among others, analyzed the sparse case. We have found that in cardinal PCD consistency and asymptotic normality hold provided $p_K=(1+c)\log K/{K}$ for some $c>0$ as $K\to\infty$ whereas uniform consistency requires that $p_K= O(\log(K)^{(1+c)}/K)$ for some $c>0$. In BT models, both complete and sparse, more stringent conditions are necessary. For example, Simons and Yao (1999) show that the MLEs exist only if $M_K=o(\sqrt{\log K/K})$ 
where $M_{K}=\sup_{1\leq i\neq j\leq K}({\mu_i}/{\mu_j})$ whereas asymptotic normality requires that $M_K=o(K^{1/10}/(\log K)^5)$. These restrictions on the parameter space do not exist in our cardinal models. Han et al. (2020), who analyzed large sparse graphs, found that if $M_{K}$, defined above, is fixed, then $p_K\ge {(\log K)^{{(1+c)}/5}}/{K}^{1/10}$ for some $c>0$ is required for asymptotic normality, and uniform consistency holds as long as $p_K\ge {(\log K)^{3}}/{K}$. Finally, if  $M_{K}$ is allowed to vary then higher levels of connectivity are required. Next we note that Han et al. (2023) considered $p_K={(\log K)^{3+c}}/{K}$ for some $c>0$ and found that the merit vector in  the Bradley--Terry model satisfies the uniform bound
$\|\widehat{\bmu}-\bmu\|_{\infty} \leq O(\sqrt{{(\log K)^3}/{Kp_K}}$ with probability at least $1-K^{-2}$. They attained the same bound for their cardinal model assuming some restrictions on the parameter space. From \eqref{prob:bound}, we see that the uniform consistency holds for graph--LMs on much sparser graphs compared with Han et al. (2023). 
\end{enumerate}

One future avenue of research is rigorously establishing points (2) and (3), listed above, for BT models. More ambitiously, the theory developed in this paper, addressing  the simplest type of PCD, can be used as a foundation for the analysis of much more general models. 

\medskip

Finally, we conclude by listing some directions where substantive extensions are possible.
}

\begin{enumerate}

\item \textbf{Robust methods for cardinal PCD:} Our simulation results indicated that even when the errors are $\mathcal{T}_2$ the merits can be estimated consistently but at a slower rate compared to the situations when the variance of the error is finite. For error distributions without a mean, e.g., the Cauchy distribution, the estimator is not consistent. Thus there is a need for suitably modified robust estimators which can handle heavy--tailed error distributions. Huber--type estimators, or least absolute deviation type estimators are possibilities which require a detailed investigation.

\item \textbf{Optimal design for PCD:} The planning of experiments for PCD is another area which has not received adequate attention. In the context of cardinal PCDs one can use the form of $\pmb{\Theta }^{+}$ to construct optimal experimental designs for eliciting information on $\bmu$. There is some work in this area (Grasshoff, 2004) but not for cardinal PCD. It would be interesting to examine A--optimal, D--optimal and other optimality criteria (cf. Pukelsheim, 2006) for sequential and non--sequential designs for estimating $\bmu$. It is also clear that different experimental  objectives, e.g., finding the best item, the $k$--best and so forth, will yield different designs.

\item \textbf{Non--transitive and cyclical PCD:} Model \eqref{model.Y_ijl} imposes, what is called a linear transitive ordering, cf. Latta (1979), Oliveira et al. (2018), on the items $\{1,\ldots,K\}$. However the model may not fit the data, i.e., the expected value of $Y_{ijk}$ may not equal $\mu_i-\mu_j$ for all pairs $(i,j)\in \mathcal{E}$. In such cases the model $Y_{ijk}=\nu_{ij}+\epsilon_{ijk}$ where $\nu_{ij}=\mathbb{E}(Y_{ij})$ must be adopted. The most general model imposes no restrictions on the vector of means $\bnu = (\nu_{12},\ldots,\nu_{1K},\nu_{23},\ldots,\nu_{2K}, \ldots,\nu_{K-1,K})$. It is important to note that the model indexed by $\bnu$ captures a much broader range of preference relations compared to the model indexed by $\bmu$. In particular it may capture a variety of transitivity relations as well as preference relations which are not transitive. Such relations are called cyclical preference relations. The problems of non--transitivity and cyclicality are, unresolved, long standing problems in the field. We have made progress towards a principled resolution of the aforementioned problem, which will be published separately.

\item \textbf{Inference under weak asymptotic connectivity:} In Section \ref{section large sample theory} we studied the behavior of the LSE under two asymptotic regimes specified by Conditions \ref{Con(AlgConnec)} and \ref{Con(Rate+CLT)} respectively. The behavior of the LSE under Condition \ref{Con(Rate+CLT)} is well understood. This is not so under Condition \ref{Con(AlgConnec)}. For example, what more can be said about $\pmb{\Psi}$, the limiting variance in Theorem \ref{theorem clt with 3.1}? Are there classes of inferential problems that can not be addressed using Theorem \ref{theorem clt with 3.1}? How does the theorem extend to models with covariates? Answering such  questions requires an understanding of the sampling process governing the evolution of $\bN$ as  well as some experience with their respective empirical performance. These are all very interesting problems to be addressed elsewhere.

\item \textbf{Order restricted inference for PCD:} Finally we note that many scientific questions arising in the context of paired comparison data are questions about order. Such problems are best formulated and addressed within the context of order restricted statistical inference (cf. Silvapulle and Sen 2005). Very little work in this direction has been carried out. For example, by specifying an orderings on subsets of items, i.e., $\mu _{i}\geq \mu _{j}$ for all $(i,j)\in \mathcal{R}$ where $\mathcal{R}$ is a set of prespecified order relations there is a potential of improving the overall accuracy. This is especially true for large sparse graphs where unstable rankings (e.g., Hsieh et al. 2011) are common. The power of tests on the merits and ranks can be also improved by imposing order restrictions. The relations with the classical problems of selection of the best population(s) is obvious.

\end{enumerate}

\section*{Acknowledgments}
The work of Rahul Singh was conducted while a post doctoral fellow at the University of Haifa. The work of Ori Davidov was partially supported by the Israeli Science Foundation Grants No. 456/17 and 2200/22 and gratefully acknowledged.

\newpage

%===============================
%===============================

%----------------------------

\clearpage
%\FloatBarrier
\section*{\underline{Figures and Tables:}}

\bigskip
\bigskip
\bigskip

%----------------------------------
\begin{figure}[!htb]
	\centering
	\begin{subfigure}[!htb]{0.33\textwidth}            
		\tikz \graph[circular placement, radius=3cm, nodes={circle,draw}] { subgraph K_n [--, n=8, clockwise, radius=1.6cm] };
		\caption{Complete}
		\label{fig1:complete}
	\end{subfigure}%
	\begin{subfigure}[!htb]{0.33\textwidth}
		\centering
		\tikz \graph[circular placement, radius=3cm, nodes={circle,draw}] { subgraph C_n [--, n=8, clockwise, radius=1.6cm] };
		\caption{Cycle}
		\label{fig1:cycle}
	\end{subfigure}
	\begin{subfigure}[!htb]{0.33\textwidth}
		\centering
		\tikz \graph[circular placement, radius=3cm, nodes={circle,draw}] { subgraph P_n [--, n=8, clockwise, radius=1.6cm] };
		\caption{Path}
		\label{fig1:path}
	\end{subfigure}
	\begin{subfigure}[!htb]{0.33\textwidth}
		\centering
		\begin{tikzpicture}
			\def \radius {1.6cm}
			
			\node[draw, circle] at (360:0mm) (center) {1};
			\foreach \i  in {2,...,8}{
				\node[draw, circle] at ({\i*50}:\radius) (u\i) {\i};
				\draw (center)--(u\i);
			}
		\end{tikzpicture}
		\caption{Star}
		\label{fig1:star}
	\end{subfigure}
	\begin{subfigure}[!htb]{0.33\textwidth}
		\centering
		\begin{tikzpicture}
			\node[circle,draw] (v) at (0,0) {$1$};
			\foreach \i in {2,...,8}{
				\node[circle,draw] (u_\i) at (50 * \i:1.6cm) {$\i$};
				\draw (u_\i) -- (v);
			}
			\foreach \i [evaluate=\i as \j using {int(mod(\i,8)+1)}] in {2,...,7}
			{
				\draw   (u_\i) -- (u_\j);
			}
			\draw   (u_2) -- (u_8);
		\end{tikzpicture}
		\caption{Wheel}
		\label{fig1:wheel}
	\end{subfigure}
	\begin{subfigure}[!htb]{0.32\textwidth}
		\centering
		\begin{tikzpicture}%
			[ every text node part/.style={draw, align = left, inner sep = 0pt} ]
			
			% Setup for horizontal tree
			% Grow tree to right with nodes placed clockwise(')
			\tikzset{grow'=left}
			% Use edges with 90° bends instead of default straight
			\tikzset{edge from parent/.style = { draw,
					edge from parent path = { (\tikzparentnode.west) 
						-- +(-8pt, 0)
						|- (\tikzchildnode.east) }}}
			% Increase horizontal spacing (adjust if length of name is long)
			\tikzset{level distance = 4em}
			% Adjust the alignment of the nodes
			\tikzset{every tree node/.style = {draw, circle,inner sep=3pt, anchor = base west}}
			
			\Tree[ .{1} 
			[.{5}
			[.{7} {8}
			[.{7}  ] ]
			[.{5} {6} {5} ]]
			[.{1}
			[.{3} {4}
			[.{3}  ] ]
			[.{1} {2} {1} ]] ]
		\end{tikzpicture}
		\caption{Tournament}
		\label{fig1:tournament}
	\end{subfigure}
	\caption{Some standard graphs with $8$ vertices. }\label{graph:topology}
\end{figure}
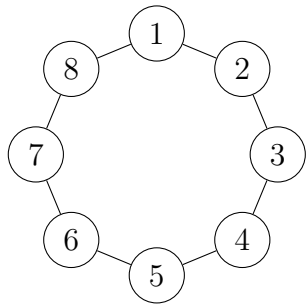
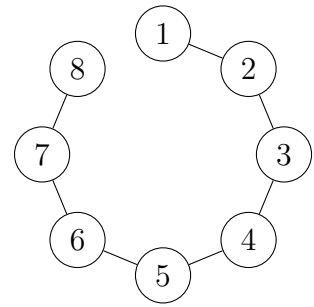
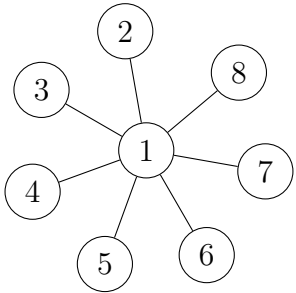
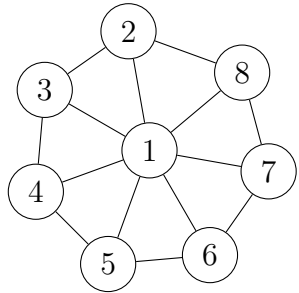
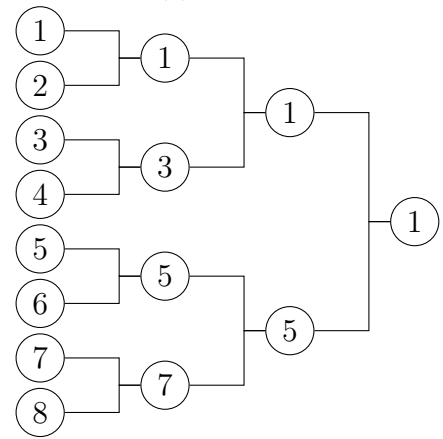
%------------------------------------
\begin{figure}[!htb]
	\centering
	\begin{subfigure}[!htb]{0.5\textwidth}            
		\includegraphics[width=\textwidth]{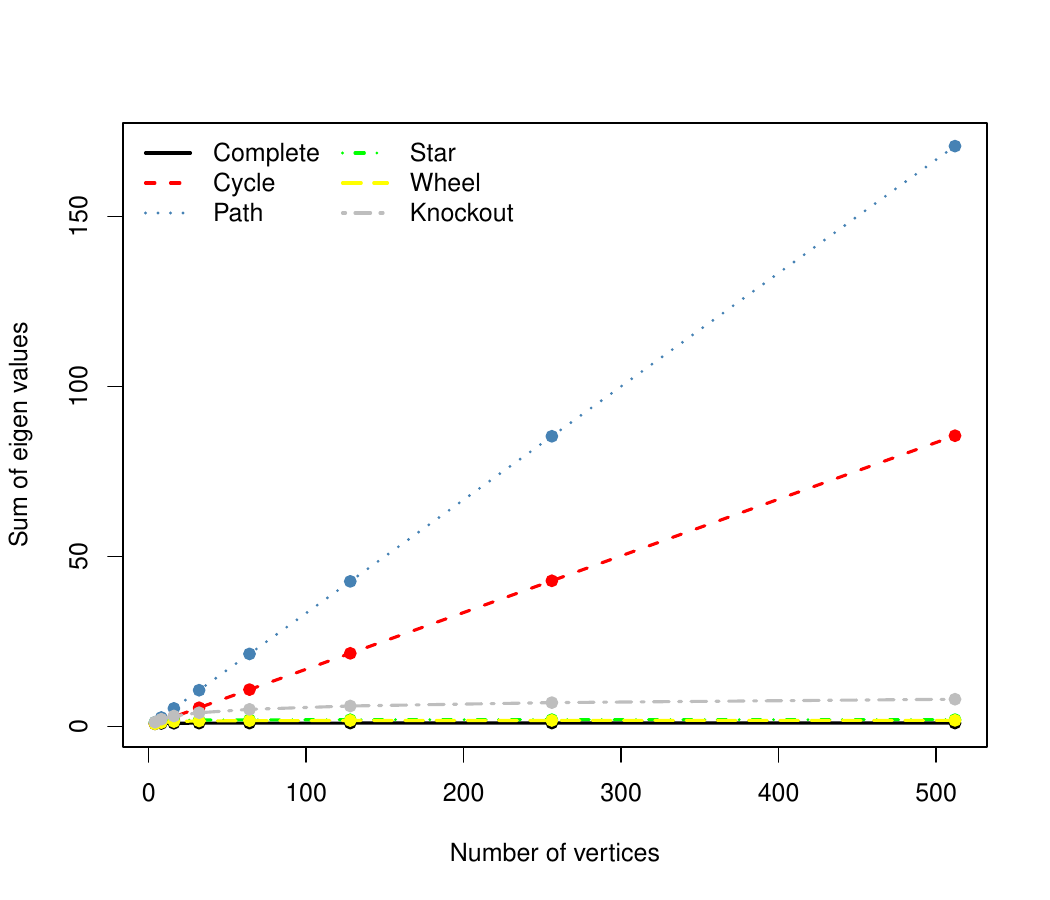}
		\caption{Sum of eigenvalues of $N^+$}
		\label{fig:sum_ev}
	\end{subfigure}%
	\begin{subfigure}[!htb]{0.5\textwidth}
		\centering
		\includegraphics[width=\textwidth]{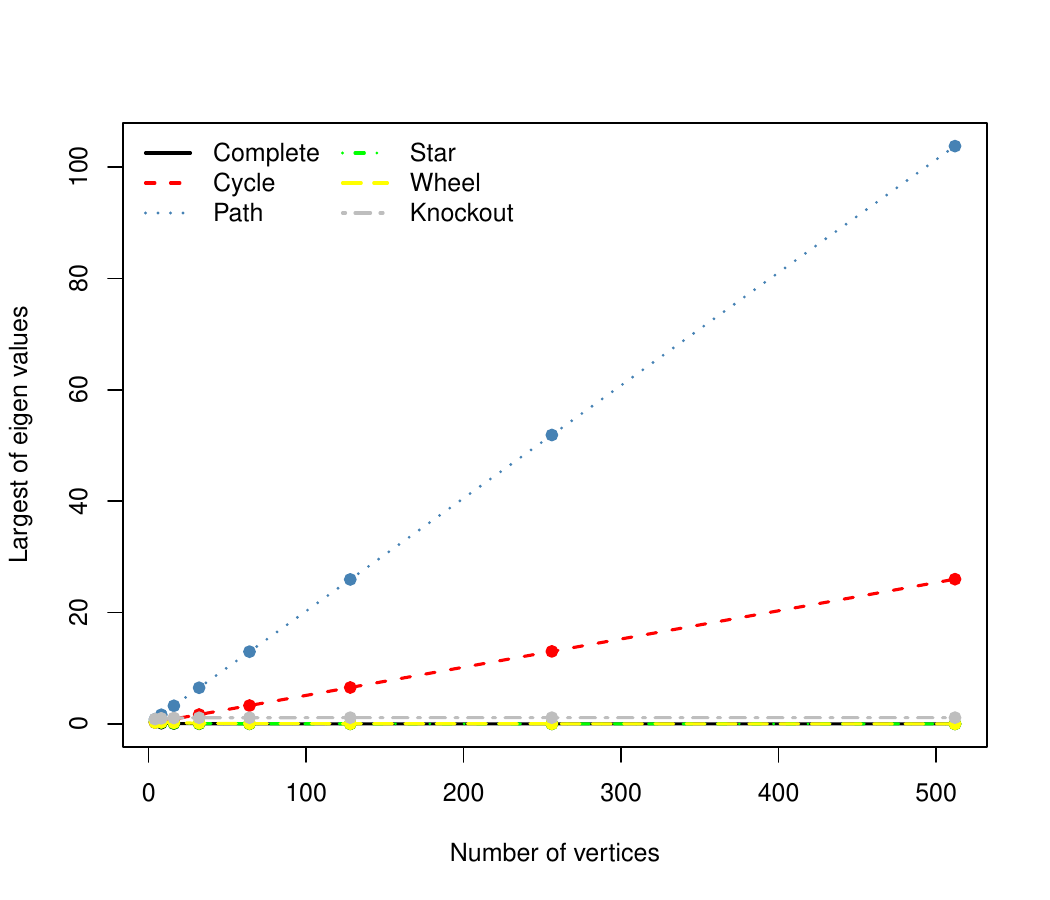}
		\caption{Largest eigenvalue of $N^+$}
		\label{fig:max_ev}
	\end{subfigure}
	\caption{Sum and largest of eigenvalues of $N^+$ against number of vertices of scaled graphs, when the number of edges is the same in all graphs. }\label{fig:precision_analysis1}
\end{figure}

\begin{figure}
	\centering
	\begin{subfigure}[!htb]{0.5\textwidth}            
		\includegraphics[width=\textwidth]{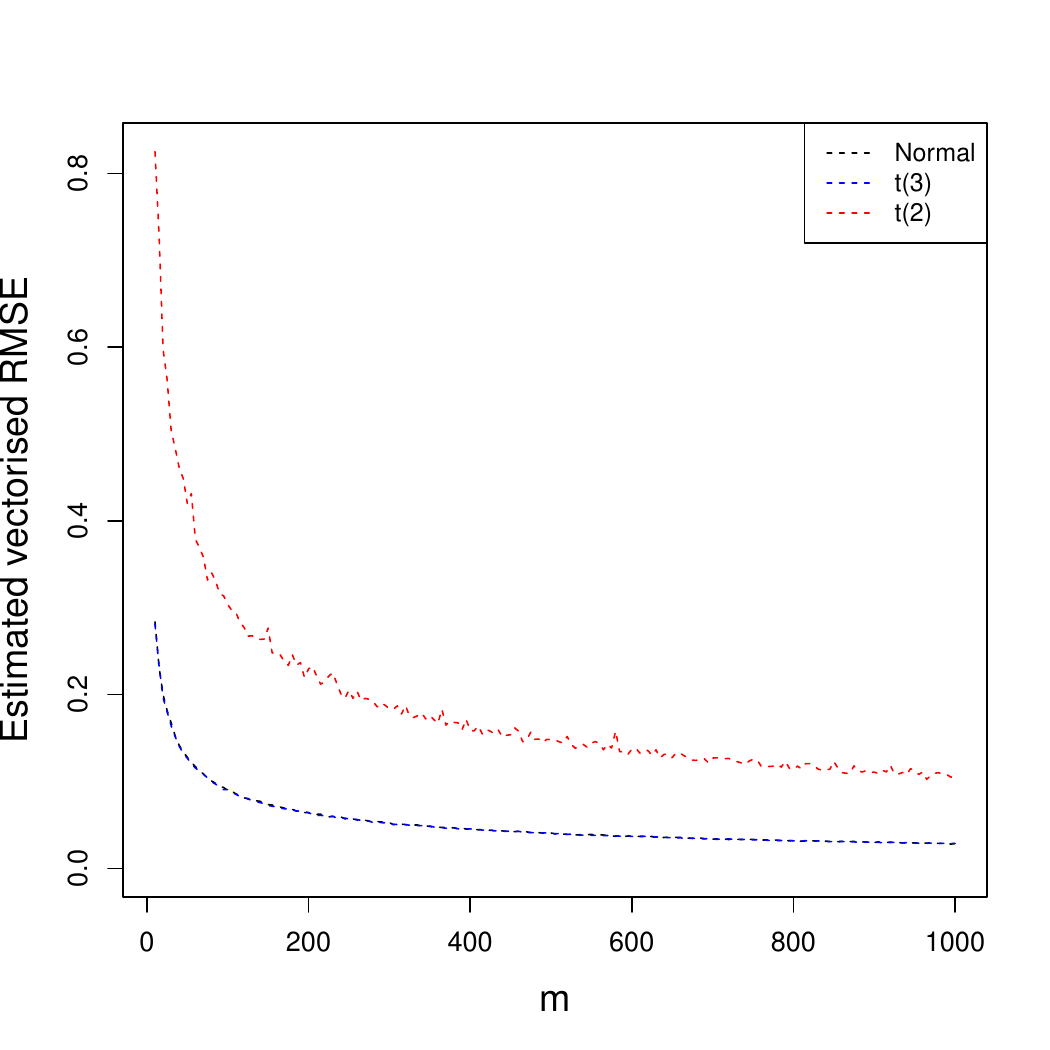}
		\caption{}
		\label{fig:mse_normal_rvs_without_cov}
	\end{subfigure}%
	\begin{subfigure}[!htb]{0.5\textwidth}
		\centering
		\includegraphics[width=\textwidth]{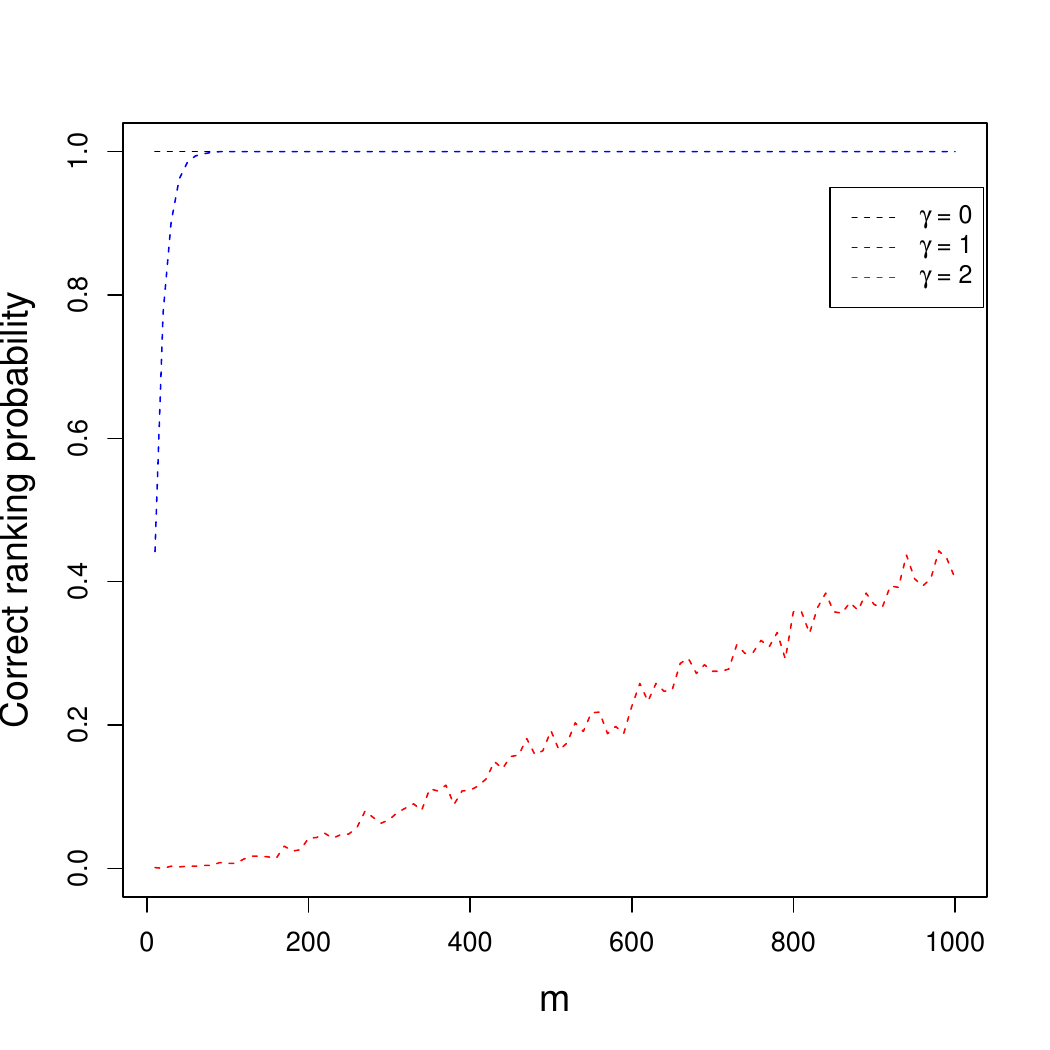}
		\caption{}
		\label{fig:corrrect_ranking_without_cov}
	\end{subfigure}
\caption{MSEs and probabilities of correct ranking} \label{fig:MSE_without_cov}
\end{figure}

\begin{figure}
\centering
\begin{subfigure}[!htb]{0.5\textwidth}    
\includegraphics[width=\textwidth]{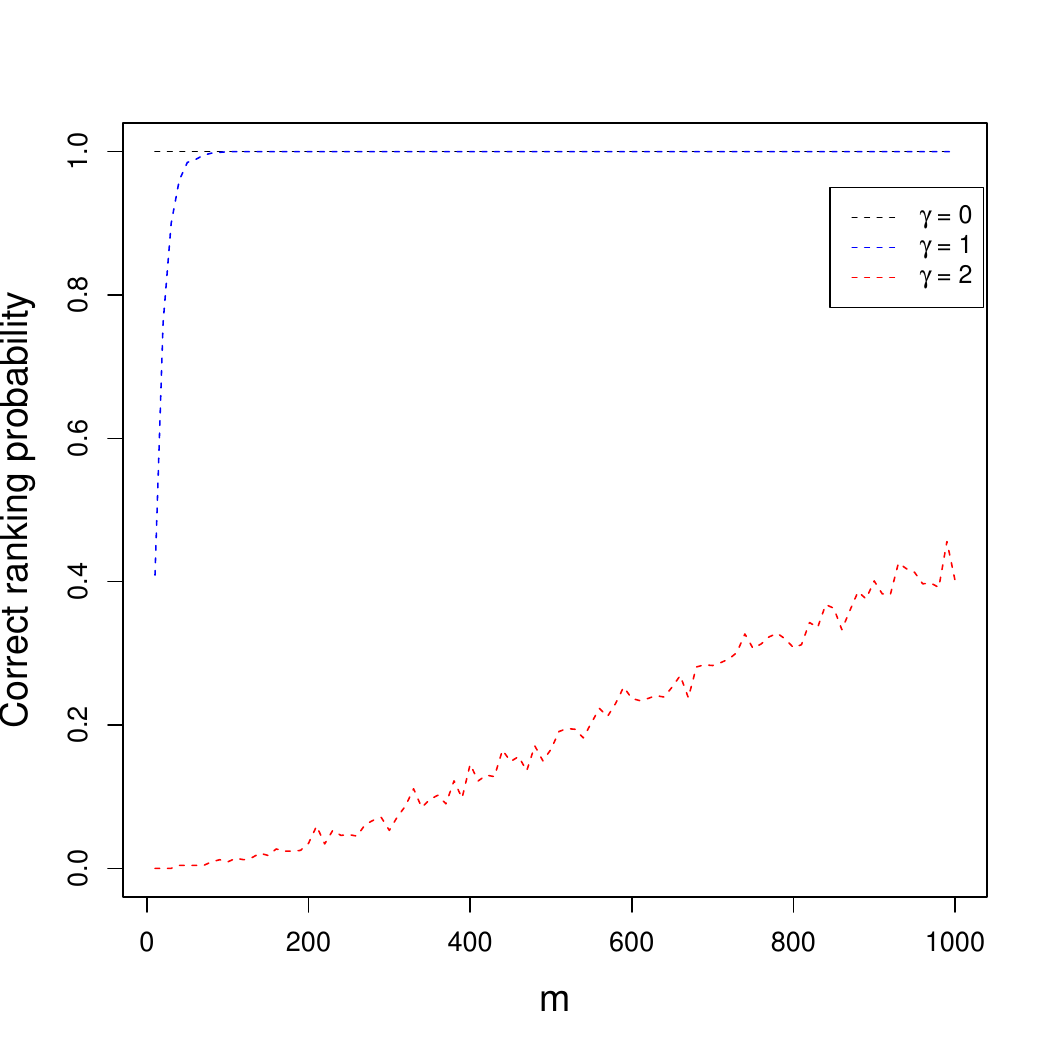}
\caption{Model with covariates}
%\label{fig:example2_binary_cov}
\end{subfigure}%
\begin{subfigure}[!htb]{0.5\textwidth}    
\includegraphics[width=\textwidth]{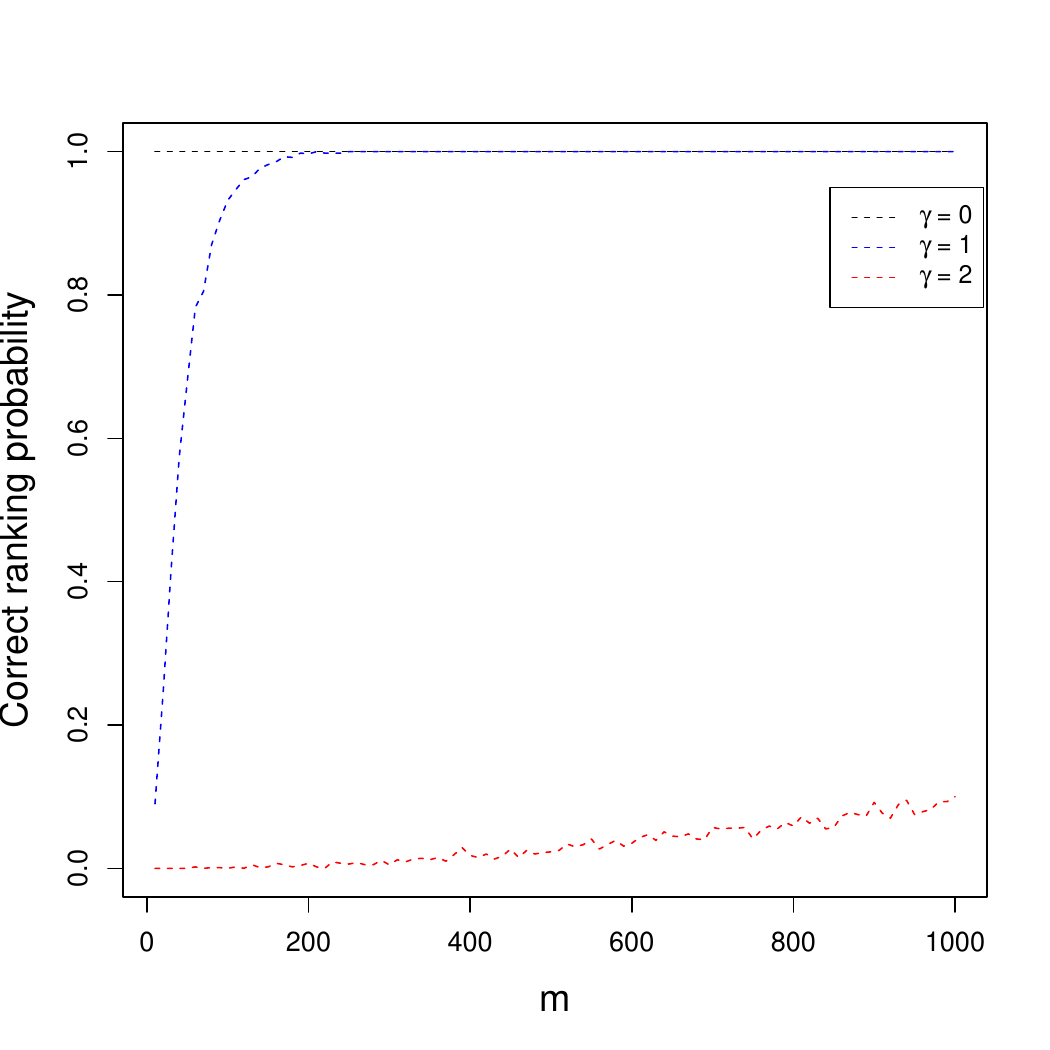}
\caption{Model without covariates}
%\label{fig:example2_binary_cov}
\end{subfigure}%
\caption{Probabilities of correct ranking for Rademacher distributed covariates} \label{fig:rank_with_cov_est_cov1}
\end{figure}

\begin{figure}[!htb]
\centering   
\includegraphics[scale=0.7]{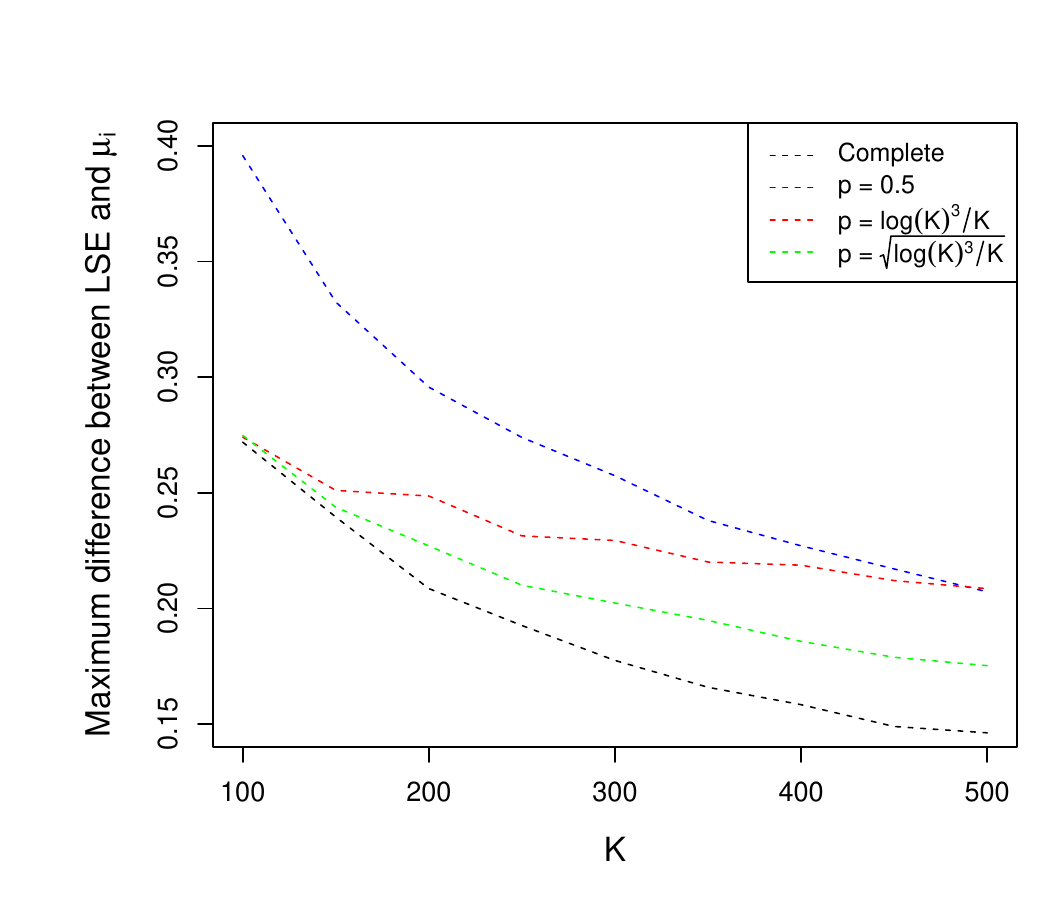}
\caption{Values of $\max_{1\leq i\leq K}|\widehat\mu_i-\mu_i|$ for complete and Erd\H{o}s--R\'enyi graphs with $p\in\{0.5,\, K^{-1}(\log K)^3,\,\sqrt{K^{-1}(\log K)^3}\}$ against $K$} \label{fig:large_dim_example}
\end{figure}

\FloatBarrier

\begin{table}[ht]
\centering
\caption{Ranking of NBA teams for different submodels of \eqref{model:with:covariate} }\label{table:ranks}
\setlength\tabcolsep{2pt}
\begin{tabular}{|l|llllllllllllllllllllllllllllll|} 
\hline
{{Teams}} &{\rotatebox{90}{Milwaukee Bucks}} & {\rotatebox{90}{Boston Celtics}} & {\rotatebox{90}{Philadelphia 76ers}} & {\rotatebox{90}{Denver Nuggets}} & {\rotatebox{90}{Cleveland Cavaliers}} & {\rotatebox{90}{Memphis Grizzlies}} & {\rotatebox{90}{Sacramento Kings}} & {\rotatebox{90}{New York Knicks}} & {\rotatebox{90}{Brooklyn Nets}} & {\rotatebox{90}{Phoenix Suns}} & {\rotatebox{90}{Golden State Warriors}} & {\rotatebox{90}{Los Angeles Clippers}} & {\rotatebox{90}{Miami Heat}} & {\rotatebox{90}{Los Angeles Lakers}} & {\rotatebox{90}{Minnesota Timberwolves}} & {\rotatebox{90}{New Orleans Pelicans}} & {\rotatebox{90}{Atlanta Hawks}} & {\rotatebox{90}{Toronto Raptors}} & {\rotatebox{90}{Chicago Bulls}} & {\rotatebox{90}{Oklahoma City Thunder}} & {\rotatebox{90}{Dallas Mavericks}} & {\rotatebox{90}{Utah Jazz}} & {\rotatebox{90}{Indiana Pacers}} & {\rotatebox{90}{Washington Wizards}} & {\rotatebox{90}{Orlando Magic}} & {\rotatebox{90}{Portland Trail Blazers}} & {\rotatebox{90}{Charlotte Hornets}} & {\rotatebox{90}{Houston Rockets}} & {\rotatebox{90}{San Antonio Spurs}} & {\rotatebox{90}{Detroit Pistons}} \\ \hline
Model I & 5 & 1 & 4 & 3 & 2 & 7 & 8 & 6 & 17 & 9 & 10 & 19 & 15 & 14 & 21 & 11 & 18 & 12 & 13 & 16 & 20 & 22 & 25 & 23 & 24 & 26 & 27 & 28 & 30 & 29 \\ 
Model II & 5 & 1 & 4 & 3 & 2 & 7 & 8 & 6 & 17 & 9 & 10 & 19 & 15 & 14 & 21 & 11 & 18 & 12 & 13 & 16 & 20 & 22 & 25 & 23 & 24 & 26 & 27 & 28 & 30 & 29 \\ 
Model III & 11 & 5 & 9 & 7 & 1 & 19 & 2 & 3 & 17 & 24 & 22 & 12 & 23 & 6 & 20 & 8 & 13 & 16 & 14 & 4 & 26 & 25 & 18 & 15 & 10 & 21 & 29 & 27 & 30 & 28 \\ 
Model IV & 11 & 5 & 9 & 8 & 1 & 19 & 2 & 4 & 17 & 24 & 21 & 15 & 23 & 6 & 20 & 7 & 14 & 13 & 12 & 3 & 25 & 26 & 18 & 16 & 10 & 22 & 29 & 27 & 30 & 28 \\ 
\hline
\end{tabular}
\end{table}

\begin{table}[!htb]
\centering
\caption{Cayley distance among derived ranking of the models} \label{table.rank.distance}
\begin{tabular}{lrrrr}
\hline
 & \multicolumn{1}{l}{Model I} & \multicolumn{1}{l}{Model II} & \multicolumn{1}{l}{Model III} & \multicolumn{1}{l}{Model IV} \\ \hline
\multicolumn{1}{l|}{Model I} & 0 & 0 & 23 & 23 \\
\multicolumn{1}{l|}{Model II} & 0 & 0 & 23 & 23 \\
\multicolumn{1}{l|}{Model III} & 23 & 23 & 0 & 8 \\
\multicolumn{1}{l|}{Model IV} & 23 & 23 & 8 & 0\\ \hline
\end{tabular}
\end{table}

%\FloatBarrier

\begin{figure}[!htb]
\centering
\begin{subfigure}[htb]{0.5\textwidth}    
\includegraphics[width=\textwidth]{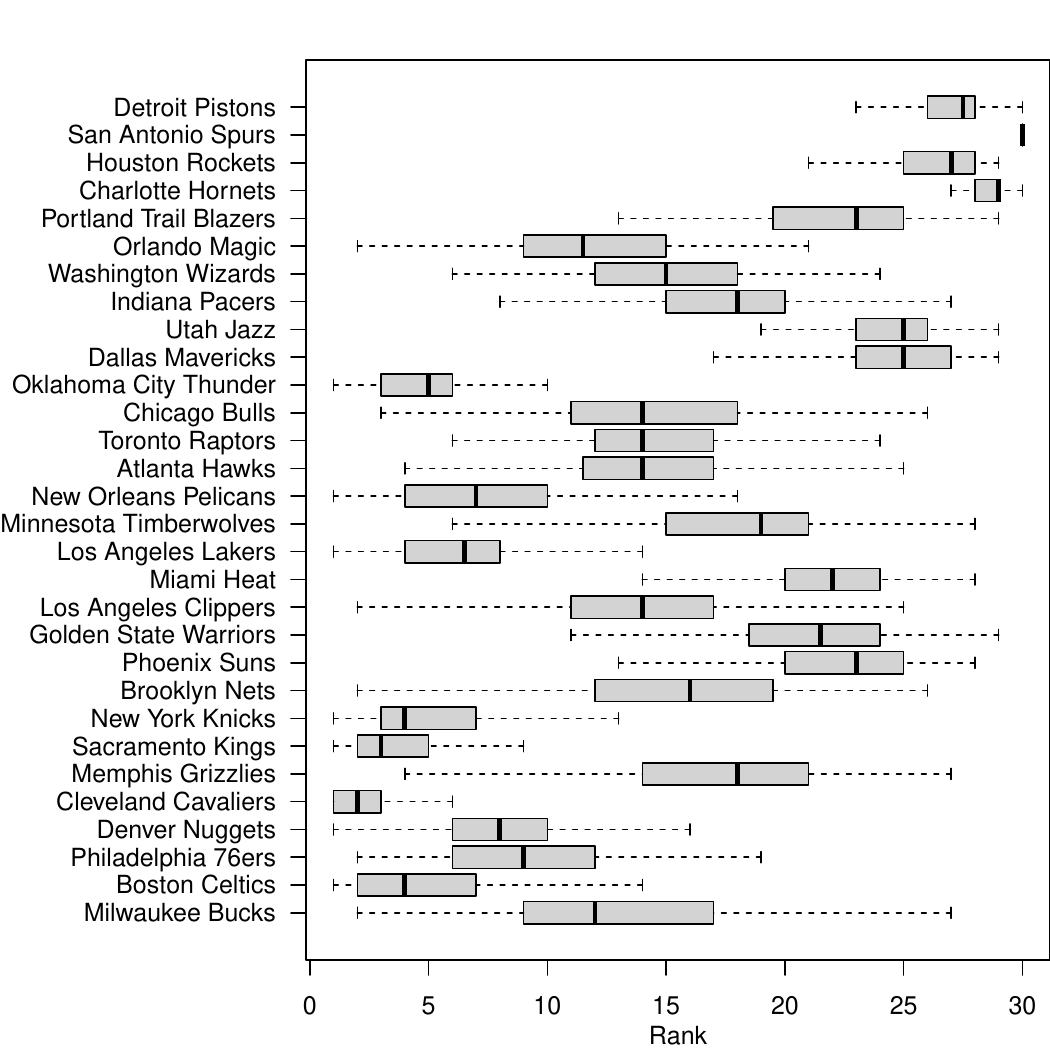}
\caption{Model with covariates}
\end{subfigure}%
\begin{subfigure}[htb]{0.5\textwidth}    
\includegraphics[width=\textwidth]{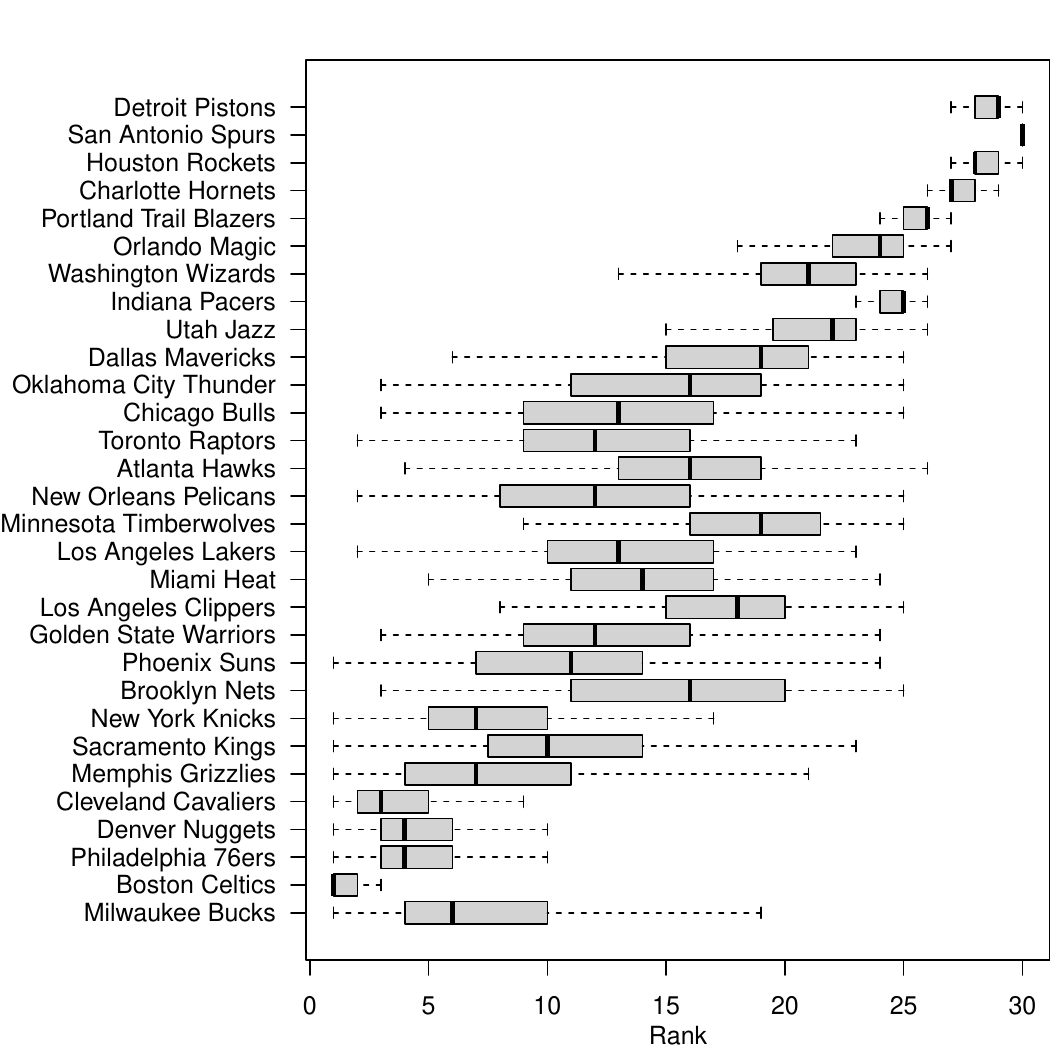}
\caption{Model without covariates}
\end{subfigure}%
\caption{Boxplots of ranking of NBA teams based on $200$ bootstrapped samples} \label{fig:real_eg}
\end{figure}
%\FloatBarrier
\clearpage
%--------------------------------------------------
\appendix
\section{Supplement} \label{section:supplement}

\subsection*{Proof of Theorem \ref{Thm-UniqueE}:}

\begin{proof}
Let 
\begin{equation}
\widehat{\bmu}=\arg \min \{Q\left( \bmu\right) :
\bv^\top\bmu=0\}.  \label{Pf-1-1}
\end{equation}
where $Q\left( \bmu\right) $ is given by (\ref{Q(mu)}). Since $Q(\bmu)$ is a positive unbounded quadratic a minimizer must exist. Now, if $\mathcal{G}$ is not connected then we may decompose it as $\mathcal{G}=\cup _{c=1}^{C}\mathcal{G}_{c}$ where for each $c=1,\ldots ,C\leq K$ the graph $\mathcal{G}_{c}$ is a connected subgraph of $\mathcal{G}$. It follows that
\begin{equation*}
Q\left( \bmu\right)
=\sum_{i<j}\sum_{k=1}^{n_{ij}}(Y_{ijk}-\left( \mu _{i}-\mu _{j}\right)
)^{2}=\sum_{c=1}^{C}\sum_{\substack{ i<j  \\ i,j\in \mathcal{G}_{c}}}%
\sum_{k=1}^{n_{ij}}(Y_{ijk}-\left( \mu _{i}-\mu _{j}\right)
)^{2}=\sum_{c=1}^{C}Q_{c}\left( \bmu_{c}\right)
\end{equation*}
where $\bmu_{c}$ is the score vector and $Q_{c}$ is the sum of squares associated with the vertices in $\mathcal{G}_{c}$ and $\bmu^{\top}=( \bmu_{1}^\top,\ldots ,\bmu_{C}^\top)$. Let $\b1_{c}$ be a vector of ones of the same dimensions as $\bmu_{c}$. Let $\widehat{\bmu}^\top=( \widehat{\bmu}_{1}^\top,\ldots ,\widehat{\bmu}_{C}^\top)$ be a minimizer of (\ref{Pf-1-1}). Define
$$ 
\widehat{\bmu}_{\pmb{\alpha}} =  (\widehat{\bmu}_{1}+\alpha_1\b1_1,\ldots ,\widehat{\bmu}_{C}+\alpha_C\b1_C).
$$ 
Clearly $\widehat{\bmu}_{\pmb{\alpha}}$ satisfies $\bv^\top\widehat{\bmu}_{\pmb{\alpha}}=0$ and thus it follows that $\bv^\top\widehat{\bmu}_{\pmb{\alpha}} = \sum_{ i=1}^C \alpha_i V_i$ where $V_i =  \sum_{j\in\mathcal{G}_i} v_j$. Hence for every $\pmb{\alpha}=(\alpha_1,\ldots,\alpha_C)^{\top}$ satisfying $\sum_{ i=1}^C \alpha_i V_i=0$ we find that $\widehat{\bmu}_{\pmb{\alpha}}$ solves \eqref{Pf-1-1}. Thus the LSE is not unique if the graph $\mathcal{G}$ is not connected. 

Hence, from this point on we will assume that $\mathcal{G}$ is connected. The proof will be completed once we explicitly find $\widehat{\bmu}$ and show that it is unique. By Lemma 2 in Osei and Davidov (2022) the function $Q(\bmu)$ can be expressed as a quadratic form in $\bmu$, i.e., 
$$Q(\bmu)= (\bmu - \bN^+ \bS)^\top \bN (\bmu - \bN^+ \bS) + c$$
where $c$ is some function of the data but not of $\bmu$. The Lagrangian associated with (\ref{Pf-1-1}) is
\begin{equation*}
\mathcal{L}\left( \bmu,\lambda \right) =Q(\bmu) +\lambda (\sum_{k=1}^{K}v_{k}\mu _{k}).
\end{equation*}
It follows that 
\begin{align*}
			&~\nabla_{\bmu} \mathcal{L}(\bmu, \lambda)
			=
			2\bN ( \bmu - \bN^+ \bS) + \lambda \bv 
			=
			2(\bN \bmu -\bN\bN^+ \bS + \frac{\lambda}{2} \bv) 
			=
			2(\bN \bmu - \bS + \frac{\lambda}{2} \bv),
\end{align*}
where $\bN \bN^+\bS=\bS$ because $\bS\in \text{im}(\bN)$.
Therefore, setting $\nabla _{{\bmu ,\lambda }}\mathcal{L}( \bmu,\lambda ) =\pmb{0}$ we obtain the system
	\begin{eqnarray}
			\bN \bmu - \bS+\frac{\lambda}{2}\bv &=&%
			\pmb{0},  \label{Pf-Thm-Unique.Eq1} \\
			\bv^\top\bmu &=&0.  \label{Pf-Thm-Unique.Eq2}
		\end{eqnarray}%
It is easy to see that $\b1^\top\bS=0$ and $\b1^\top\bN=\bzero$; in addition $\b1^{\top}\bv \neq 0$ by assumption. Thus by premultiplying (\ref{Pf-Thm-Unique.Eq1}) by $\pmb{1}^\top$ we find that $\lambda =0$. Thus we get
\begin{equation}
\widehat{\bmu}=\bN^{+}\bS+(\pmb{I}-
\bN^{+}\bN)\pmb{\xi }  \label{Pf-1-2}
\end{equation}%
where $\bN^{+}$ is the Moore--Penrose inverse of $\bN$
and $\pmb{\xi }\in \mathbb{R}^{K}$ is arbitrary. The rows of $\bN^{+}$, like those of $\bN$, sum to $0$ (Bapat, 2010). Thus $\pmb{1}^\top\bN^{+}\bS=0$. The matrix $\pmb{I}-\bN^{+}\bN$ is a projection onto $\ker \left( \bN\right) .$
Since $\mathcal{G}$ is connected, $\mathrm{rank}(\bN)=K-1$ and $
\dim (\ker \left( \bN\right) )=1$, therefore $(\pmb{I}-
\bN^{+}\bN)\pmb{\xi }$ \ equals $c\pmb{1}$
for some $c\in 
\mathbb{R}$ where $c\neq 0$ if and only if $\pmb{\xi }\neq \pmb{0}.$
Thus if we premultiply (\ref{Pf-1-2}) by $\bv^\top$ we have
\begin{equation*}
0=\bv^\top\widehat{\bmu}=\bv^\top(
\bN^{+}\bS+(\pmb{I}-\bN^{+}
\bN)\pmb{\xi })=\bv^\top\bN^{+}
\bS+c\bv^\top\pmb{1}
\end{equation*}
so $c=-\bv^\top\bN^{+}\bS/\bv^\top
\pmb{1}$ which we substitute back in (\ref{Pf-1-2}) to obtain
\begin{equation*}
\widehat{\bmu}=\bN^{+}\bS-\frac{
\bv^\top\bN^{+}\bS}{\bv^\top
\b1}\b1.
\end{equation*}
The uniqueness of $\widehat{\bmu}$ follows from the uniqueness
of $\bN^{+}.$ This concludes the proof.
\end{proof}

\begin{remark} \label{rem.on.v}
The condition that $\b1^{\top}\bv \neq 0$ is necessary. For if $\b1^{\top}\bv=0$ then $\bv \in \mathrm{im}(\bN)$ and consequently the system \eqref{Pf-Thm-Unique.Eq1}--\eqref{Pf-Thm-Unique.Eq2} does not have a solution.      
\end{remark}

{
\subsection*{Proof of Proposition \ref{Prop.LSE.uv}:}
\begin{proof}
Observe that
\begin{align*}
\widehat{\bmu}(\pmb{u})= \bN^{+}\bS-\frac{\pmb{u}^\top\bN^{+}\bS}{\pmb{u}^\top\b1}\b1 =~ \bN^{+}\bS-\frac{\b1\pmb{u}^{\top}}{\pmb{u}^{\top}\b1} \bN^{+}\bS= (\pmb{I}-\frac{\b1\pmb{u}^{\top}}{\pmb{u}^{\top}\b1}) \bN^{+}\bS=~ \bC_{\pmb{u}} \bN^{+}\bS=\bC_{\pmb{u}} \widehat{\bmu}(\b1),
\end{align*}
where $\widehat{\bmu}(\b1)=\bN^{+}\bS$. If $\pmb{C}_{\bv}^{+}$ is a matrix such that $\pmb{C}_{\bu}\pmb{C}_{\bv}^{+}\pmb{C}_{\bv}= \pmb{C}_{\bu}$, then
\begin{align*}
\widehat{\bmu}(\pmb{u}) = \pmb{C}_{\bu} \widehat{\bmu}(\b1)=\pmb{C}_{\bu}\pmb{C}_{\bv}^{+}\pmb{C}_{\bv}\widehat{\bmu}(\b1),
\end{align*}
and the proof is complete. We now verify $\pmb{C}_{\bu}\pmb{C}_{\bv}^{+}\pmb{C}_{\bv}= \pmb{C}_{\bu}$ where $\pmb{C}_{\bu}$ and $\pmb{C}_{\bv}^{+}$ are defined in the statement of the proposition. First note that 
\begin{align*}
\pmb{C}_{\bv}^{+}\pmb{C}_{\bv}= ( \pmb{I} - \frac{\bv\bv^{\top}}{\bv^{\top}\bv} )( \pmb{I} - \frac{\b1\bv^{\top}}{\bv^{\top}\b1} )=
\pmb{I} - \frac{\b1\bv^{\top}}{\bv^{\top}\b1} - \frac{\bv\bv^{\top}}{\bv^{\top}\bv} + \frac{(\bv\bv^{\top})(\b1\bv^{\top})}{(\bv^{\top}\bv)(\bv^{\top}\b1)}.
\end{align*}
However, $(\bv\bv^{\top})(\b1\bv^{\top}) = \bv(\bv^{\top}\b1\bv^{\top}) = \bv(\bv^{\top}\b1)\bv^{\top} = (\bv^{\top}\b1)\bv\bv^{\top}$ and therefore $\pmb{C}_{\bv}^{+}\pmb{C}_{\bv}= \pmb{C}_{\bv}$. Next, we establish that $\pmb{C}_{\bu}\pmb{C}_{\bv}= \pmb{C}_{\bu}$. Observe that
$$\pmb{C}_{\bu}\pmb{C}_{\bv} = \pmb{I} - \frac{\b1\bv^{\top}}{\bv^{\top}\b1} - \frac{\b1\bu^{\top}}{\bu^{\top}\b1} + \frac{(\b1\bu^{\top})(\b1\bv^{\top})}{(\bu^{\top}\b1)(\bv^{\top}\b1)},$$
and $(\b1\bu^{\top})(\b1\bv^{\top}) = \b1(\bu^{\top}\b1\bv^{\top}) = \b1(\bu^{\top}\b1)\bv^{\top} = (\bu^{\top}\b1)\b1\bv^{\top}$. So, the second and fourth terms in the display above are equal, and hence $\pmb{C}_{\bu}\pmb{C}_{\bv}=\pmb{C}_{\bu}$. Therefore, 
$$\pmb{C}_{\bu}\pmb{C}_{\bv}^{+}\pmb{C}_{\bv}=\pmb{C}_{\bu}.$$
Finally, it is straightforward to verify that $\pmb{C}_{\bv}^{2}= \pmb{C}_{\bv}$ and $(\pmb{C}_{\bv}^{+})^2=\pmb{C}_{\bv}^{+}$ so both $\pmb{C}_{\bv}$ and $\pmb{C}_{\bv}^{+}$ are idempotent and that $\pmb{C}_{\bv}^{+}$ is the Moore-Penrose inverse of $\pmb{C}_{\bv}$.
\end{proof}}

%------------------
\subsection*{Proof of Corollary \ref{Cor-EstDiff}:}
%----------------

\begin{proof}
Note that
\begin{eqnarray*}
\widehat{\mu }_{i}(\bv)-\widehat{\mu }_{j}(\bv) &=&
\pmb{e}_{i}^\top\widehat{\bmu}(\bv)-
\pmb{e}_{j}^\top\widehat{\bmu}(\bv)=(
\pmb{e}_{i}-\pmb{e}_{j})^\top\widehat{\bmu}(
\bv)=(\pmb{e}_{i}-\pmb{e}_{j})^\top\widehat{
\bmu}(\bv) \\
&=&(\pmb{e}_{i}-\pmb{e}_{j})^\top(\bN^{+}\pmb{S}-\frac{\bv^\top\bN^{+}\bS}{\bv^\top\b1}\pmb{1)} \\
&=&(\pmb{e}_{i}-\pmb{e}_{j})^\top\bN^{+}\pmb{S}-\frac{\bv^{\top}\bN^{+}\bS}{\bv^\top\b1}(\pmb{e}_{i}-\pmb{e}_{j})^\top\b1
\\
&=&(\pmb{e}_{i}-\pmb{e}_{j})^\top\bN^{+}\pmb{S}
\end{eqnarray*}
is independent of $\bv$. Thus $\widehat{\mu }_{i}(\bv)-
\widehat{\mu }_{j}(\bv)=\widehat{\mu }_{i}(\pmb{u})-
\widehat{\mu }_{j}(\pmb{u})$ as required.
\end{proof}

%------------------
\subsection*{Proof of Proposition \ref{prop1}:}
%----------------
\begin{proof}
By assumption $(\mathcal{G}_1,\mathcal{Y}_1) \supseteq (\mathcal{G}_2,\mathcal{Y}_2)$. Let $\mathcal{G}_3=(\mathcal{V},\mathcal{E}_3)$ where $\mathcal{E}_3=\mathcal{E}_1\backslash\mathcal{E}_2$ and let $\mathcal{Y}_{3} =\mathcal{Y}_{1} \backslash \mathcal{Y}_2$. Clearly $(\mathcal{G}_3,\mathcal{Y}_3)$ is a comparison graph. Let $\bN_1,\bN_2,$ and $\bN_3$ be the Laplacians of $\mathcal{G}_1,\mathcal{G}_2$, and $\mathcal{G}_3$ respectively. Clearly $\bN_3 =\bN_1 -\bN_2$ is a nonnegative definite matrix so $\bN_2 \preceq\bN_1$. Furthermore the uniqueness of the LSEs implies that  $\mathcal{G}_1$ and $\mathcal{G}_2$ are connected graphs and therefore the matrices $\bN_1$ and $\bN_2$ have rank $K-1$. Therefore using  Theorem 3.1 in Milliken and Akdeniz (1977) we get $\bN_1^{+} \preceq\bN_2^{+}$. This proves the first claim. The second claim follows immediately by using the first order approximation for the variance of $\pmb{\Phi}(\widehat{\bmu})$ (see for example, page 122, in Serfling 2009).
\end{proof}
%------
\subsection*{Proof of Theorem \ref{Thm-graph.WLLN}:}
%--------

\begin{proof}
{ 
Suppose that Condition \ref{Con(AlgConnec)} holds. If so, $\mathcal{G}$ is connected and the LSE is given by $\widehat{\bmu}_{n}=\bN^{+}\bS$ and its variance is $\mathbb{V}ar(\widehat{\bmu}_n)=\, \sigma^2\bN^{+}$. The Laplacian is symmetric and non--negative definite; its spectral decomposition is $\bN=\pmb{U\Lambda U}^\top$ where $\pmb{U}$ is an orthonormal matrix and $\pmb{\Lambda }=\mathrm{diag}(\lambda _{K},\ldots ,\lambda _{1})$ with $\lambda _{K}\geq \lambda_{K-1}\geq \cdots \geq \lambda _{2}\geq \lambda _{1}$ where $\lambda _{j}=\lambda _{j}\left( \bN\right)$ are the eigenvalues of $\bN$. Since $\mathcal{G}$ is connected $\lambda _{2}>0$ and $\lambda _{1}=0.$ Furthermore $\pmb{N}^{+}=\pmb{U}\pmb{\Lambda }^{+}\pmb{U}^\top$ where $\pmb{\Lambda }^{+}=\mathrm{diag}(1/\lambda _{K},\ldots ,1/\lambda_{2},0)$. 

Let $\mathcal{T}$ be a tree satisfying Condition \ref{Con(AlgConnec)}. Let $\mathcal{G}_{1}$ be the subgraph of $\mathcal{G}$ for which $n_{ij}=t$ for all edges in $\mathcal{T}$ where $t=\min\{n_{ij}:(i,j) \in \mathcal{T}\}$. Further define $\mathcal{G}_{2}=\mathcal{G}\backslash \mathcal{G}_{1}$ so by construction $\mathcal{G}=\mathcal{G}_{1}\cup \mathcal{G}_{2}$. Let $\bN_{1}$ and $\bN_{2}$ be the Laplacians of $\mathcal{G}_{1}$ and $\mathcal{G}_{2}$ respectively. It is easy to see that
\begin{equation*}
\bN=\bN_{1}+\bN_{2}.
\end{equation*}
By Weyl's inequality (Theorem 4.3.1, Horn and Johnson, 2007) we have
\begin{equation*}
\lambda _{2}( \bN) \geq \lambda _{2}( \bN_{1}) +\lambda _{1}( \bN_{2}) =t\lambda _{2}( \bN_{1}/t)
\end{equation*}
where $\bN_{1}/t$ is the Laplacian of a tree with muliplicity one on the vertices of $\mathcal{G}_{1}$ and $\lambda_{1}( \bN_{2}) =0$. Moreover, a spanning tree is a connected graph so $\lambda _{2}( \bN_{1}/t) >0$. It follows that $\lambda _{2}( \bN) \rightarrow \infty $ as $n\rightarrow \infty$. 
Furthermore $\lambda _{j}( \bN) \geq \lambda
_{2}\left( \bN\right) $ and consequently $1/\lambda_j(\bN) \to 0$ as $n\to \infty$ for all $j\geq 2$. Now
\begin{equation}\label{eq:N+:decomposition}
{\bN}^{+}=\pmb{U}\pmb{\Lambda }^{+}\pmb{U}^\top= \sum_{j=2}^{K}\frac{1}{\lambda_j(\bN)}\, \pmb u_j\pmb u_j^{\top}
\end{equation}
where $\pmb u_1,\ldots,\pmb u_K$ are the orthonormal columns of $\pmb U$. The orthonormality implies that the absolute value of all elements of $\pmb U$ must be bounded by $1$. Consequently the display above shows that $\bN^{+}\rightarrow \pmb{O}$, the zero matrix. Therefore $\mathbb{V}ar(\widehat{\bmu}_{n})\rightarrow \pmb{O}$. We have already shown that $\mathbb{E}(\widehat{\bmu}_{n})=\bmu$ and it
now follows that $\widehat{\bmu}_{n}$ is consistent, i.e., $\widehat{\bmu}_{n}\overset{p}{\rightarrow }\bmu$
as required.

Next, suppose that Condition \ref{Con(AlgConnec)} does not hold, if so, $m$ as defined in \eqref{m:maxmin}, satisfies $m<\infty$. Let $\mathcal{G}_{1}$ be the subgraph of $\mathcal{G}$ for which $n_{ij}=m$ for all edges in $\mathcal{T}_m$. Define $\mathcal{G}_{2}$, $\bN_{1}$ and $\bN_{2}$ as before. Now $\mathcal{G}_{2}$ is either a connected or disconnected graph. If $\mathcal{G}_{2}$ is disconnected then by Weyl's inequality we have
\begin{equation*}
\lambda _{2}\left( \bN\right) \leq \lambda _{2}\left( \bN_{2}\right) +\lambda _{K}\left( \bN_{1}\right).
\end{equation*}
Since $\mathcal{G}_{2}$ is disconnected $\mathrm{rank}(\bN_{2})\leq K-2$ so $\lambda _{2}\left( \bN_{2}\right) =0$. Further note that $\lambda _{K}\left( \bN_{1}\right) \leq mC$ for some finite constant $C$ and thus $\lambda _{2}\left( \bN\right) $ is finite even if $n\rightarrow \infty$. This shows that $\bN^{+}$ does not converge to $\pmb{O}$ and consequently $\mathbb{V}ar(\widehat{\bmu}_{n})\not\to \pmb{O}$. If, however, $\mathcal{G}_{2}$ is connected, then using the same procedure we can extract a tree from $\mathcal{G}_{2}$ and denote the remainder by $\mathcal{G}_{3}$. This procedure can be repeated until the remainder graph $\mathcal{G}_{l}$ is disconnected. Since $K$ is finite, $l$ must be finite. Thus applying Weyl's inequality $l-1$ times find that $\lambda _{2}(\bN)$ is finite as $n\to\infty$. Consequently $\bN^{+}\not\to\pmb{O}$ as $n\to\infty$.  

Finally, it follows from \eqref{eq:N+:decomposition} that 
$$\mathbb{V}ar(\widehat{\bmu}_n) = O(\frac{1}{\lambda_2(\bN)}) $$
Therefore, in order to prove that $\mathbb{V}ar(\widehat{\mu}_{i,n}) =O({1}/{m})$ we need to show that $\lambda_2(\bN)=O(m)$. To do this, we bound $\lambda_2(\bN)$ by $O(m)$ from above and below. Let $\mathcal{T}_{m}$ denote the tree on which the maximum is attained and let $\bN_{\mathcal{T}_{m}}$ denote the corresponding Laplacian. It follows that
$$\lambda_{2}(\bN_{\mathcal{T}_{m}}) \ge m \{2(1-\cos(\frac{\pi}{K}))\}.$$
where the bound $2(1-\cos(\pi/K)$ is the algebraic connectivity of a path graph, (e.g., Cvetkovic et al. 2009). Consequently $\lambda_{2}(\bN_{\mathcal{T}_{m}}) \geq O(m)$. By Proposition \ref{prop1}, $\bN_{\mathcal{T}_m} \preceq\bN$ and therefore 
\begin{align} \label{lambda2 lower}
\lambda_2(\bN)\geq O(m).
\end{align}
Next, let $\mathcal{G}_{1}$ be a subgraph of $\mathcal{G}$ for which $n_{ij}=m$ for all edges in $\mathcal{T}_m$. Define $\mathcal{G}_{2}$, $\bN_{1}$ and $\bN_{2}$ as before. Now $\mathcal{G}_{2}$ is either a connected or disconnected graph. If $\mathcal{G}_{2}$ is disconnected then by Weyl's inequality we have
\begin{equation*}
\lambda _{2}\left( \bN\right) \leq \lambda _{2}\left( \bN_{2}\right) +\lambda _{K}\left( \bN_{1}\right).
\end{equation*}
Since $\mathcal{G}_{2}$ is disconnected $\mathrm{rank}(\bN_{2})\leq K-2$ so $\lambda _{2}\left( \bN_{2}\right) =0$. Further note that 
$\lambda _{K}\left( \bN_{1}\right) \leq mC$ for some finite constant $C$ and thus $\lambda _{2}\left( \bN\right)\leq O(m) $. If, however, $\mathcal{G}_{2}$ is connected, then using the same procedure we can extract a tree from $\mathcal{G}_{2}$ and denote the remainder by $\mathcal{G}_{3}$. This procedure can be repeated until the remainder graph $\mathcal{G}_{l}$ is disconnected. Since $K$ is finite, $l$ must be finite. Thus applying Weyl's inequality $l-1$ times we find that 
\begin{align} \label{lambda2 upper}
\lambda_2(\bN)\leq O(m).
\end{align}
Equations \eqref{lambda2 lower} and \eqref{lambda2 upper} imply that $\lambda_2(\bN)= O(m)$. Therefore \eqref{var.mu.i.hat} holds, completing the proof. }
\end{proof}

%-----------
\subsection*{Proof of Theorem \ref{Thm-graph.SLLN}:}
%------------

The proof of Theorem \ref{Thm-graph.SLLN} requires additional notation and preliminary lemma's. In what follows denote by $\bN_{i}^{j}$ and $\bN_{is}^{jt}$ the submatrices arising from $\bN$ by deleting the $i^{th}$ row and $j^{th}$ column and the $i^{th}$ and $s^{th}$ rows and $j^{th}$ and $t^{th}$ columns, respectively. The $(i,j)$ entry of $\bN^{+}$ is denoted by $n_{ij}^{+}$. The determinant of matrix $\pmb{A}$ is denoted by $|\pmb{A}|$. The diagonal matrix with main diagonal $(n_{12},\ldots ,n_{1K},n_{23},\ldots ,n_{2K},\ldots ,n_{K-1,K})$ is denoted by $\pmb{D}$. It will be convenient to label the rows and columns of this and every other matrix having $K(K-1)/2$ rows and/or columns using double labels which are ordered lexicographically as are the $n_{ij}$'s in $\pmb{D}$. Let $\overline{\bS}$ be the vector of sample means $n_{ij}^{-1}\sum_{k =1}^{n_{ij}}Y_{ijk }$, where $1\leq i<j\leq K.$ If $n_{ij}=0$ for some pair $(i,j)$ then set $\overline{S}_{ij}=0$. Let $\pmb{C}$ be the $(K(K-1)/2)\times K$ contrast matrix each row of which contains a unity, a negative unity and zeros, such that 
\begin{equation*}
\pmb{\delta }=\pmb{C}\bmu=(\mu _{1}-\mu
_{2},\ldots ,\mu _{1}-\mu _{K},\mu _{2}-\mu _{3},\ldots ,\mu _{2}-\mu
_{K},\ldots ,\mu _{K-1}-\mu _{K})^\top.
\end{equation*}%
Then we have the following.

\begin{lemma}
\label{C properties}(a) $\bS=\pmb{C}^\top\pmb{D}\overline{
\bS}$; (b) $\bN=\pmb{C}^\top\pmb{D}
\pmb{C}$; (c) $\pmb{C}^\top\pmb{C}\pmb{\delta }=K
\bmu$; (d) $\pmb{C}\bN^{+}\pmb{C}^\top
\pmb{D}\pmb{\delta }=\pmb{\delta }$.
\end{lemma}

\begin{proof}
(a), (b) and (c) are easy to see. For (d) use (b) and the fact that $%
\bN^+\bN\bmu=\bmu$ to get $
\pmb{C}\bN^+\pmb{C}^T\pmb{D}\pmb{
\delta }= \pmb{C}\bN^+\pmb{C}^T\pmb{D}
\pmb{C}\bmu= \pmb{C}\bN^+\bN
\bmu= \pmb{C}\bmu= \pmb{\delta}$.
\end{proof}

\begin{lemma}
\label{N+}If the graph is connected we have 
\begin{equation}
n_{ij}^{+}=K^{-2}|\bN_{1}^{1}|^{-1}(-1)^{i+j}\sum_{\substack{ s=1 
\\ s\neq i}}^{K}\sum_{\substack{ t=1  \\ t\neq j}}
^{K}(-1)^{s+I(s>i)+t+I(t>j)}|\bN_{is}^{jt}|.  \label{ij of N+}
\end{equation}
\end{lemma}

\begin{proof}
Chebotarev and Shamis (1998) showed that if the graph is connected then $
\bN^{+}=(\bN+K^{-1}\pmb{J})^{-1}-K^{-1}
\pmb{J}$ where $\pmb{J}$ is the $K\times K$ matrix of ones. By
doing the required algebraic manipulations it can be verified that the $i,j$
entry of this matrix is as in \eqref{ij of N+}.
\end{proof}

%\begin{lemma} Let $1\leqslant\beta_1<\beta_2\leqslant K$. If $1\leqslant \alpha_1<\alpha_2<\gamma\leqslant K$ we have
%\begin{equation}\label{gamma above}
%|\bs N_{\alpha_1,\alpha_2}^{\beta_1,\beta_2}| = (-1)^{\gamma-\alpha_1-1}|\bs N_{\alpha_2,\gamma}^{\beta_1,\beta_2}|+(-1)^{\gamma-\alpha_2}|\bs N_{\alpha_1,\gamma}^{\beta_1,\beta_2}|
%\end{equation}
%while if $1\leqslant \gamma<\alpha_1<\alpha_2\leqslant K$ we have
%\begin{equation}\label{gamma below}
%|\bs N_{\alpha_1,\alpha_2}^{\beta_1,\beta_2}| = (-1)^{\alpha_1-\gamma}|\bs N_{\alpha_2,\gamma}^{\beta_1,\beta_2}|+(-1)^{\alpha_2-\gamma-1}|\bs N_{\alpha_1,\gamma}^{\beta_1,\beta_2}|.
%\end{equation}
%\end{lemma}

\begin{lemma}
\label{Nklij}Let $1\leqslant \alpha _{1}<\alpha _{2}\leqslant K$, $%
1\leqslant \beta _{1}<\beta _{2}\leqslant K$. Then, for any $\gamma \in
\{1,\ldots ,K\}\backslash \{\alpha _{1},\alpha _{2}\}$ we have 
\begin{equation}
|\bN_{\alpha _{1},\alpha _{2}}^{\beta _{1},\beta
_{2}}|=(-1)^{\gamma +1+\alpha _{1}+I(\alpha _{2}>\gamma )}|\bN
_{\alpha _{2},\gamma }^{\beta _{1},\beta _{2}}|+(-1)^{\gamma +1+\alpha
_{2}+I(\gamma >\alpha _{1})}|\bN_{\alpha _{1},\gamma }^{\beta
_{1},\beta _{2}}|.  \label{Bapat2013}
\end{equation}
\end{lemma}

\begin{proof}
When $\alpha_1<\alpha_2<\gamma$, \eqref{Bapat2013} coincides to the equation
given in Bapat (2013, Lemma 2). The proof for the
cases $\gamma<\alpha_1<\alpha_2$ and $\alpha_1<\gamma<\alpha_2$ follows
among the same lines.
\end{proof}

\begin{lemma}
\label{CNCD} Let $1\leqslant k<\ell\leqslant K$, $1\leqslant i<j\leqslant K$
. The $(k,\ell;i,j)$ entry of matrix $\pmb{C}\bN^+
\pmb{C}^T$ is $(-1)^{k+\ell+i+j}|\bN_{k\ell}^{ij}|/|
\bN_1^1|$ while that of $\pmb{C}\bN^+
\pmb{C}^T\pmb{D}$ is $(-1)^{k+\ell+i+j}n_{ij}|\bN
_{k\ell}^{ij}|/|\bN_1^1|$.
\end{lemma}

\begin{proof}
It is easy to verify that the $(k,\ell ;i,j)$ entry of $\pmb{C}
\bN^{+}\pmb{C}^\top$ equals $n_{ki}^{+}-n_{\ell
i}^{+}-n_{kj}^{+}+n_{\ell j}^{+}$. Note that 
\begin{align*}
K^{2}|\bN_{1}^{1}|(-1)^{k+i}n_{ki}^{+}=& \sum_{s\neq
k}\sum_{t\neq i}(-1)^{s+I(s>k)+t+I(t>i)}|\bN_{ks}^{it}| \\
=& \sum_{s<k}\sum_{t\neq i}(-1)^{s+t+I(t>i)}|\bN
_{ks}^{it}|+\sum_{k<s<\ell }\sum_{t\neq i}(-1)^{s+1+t+I(t>i)}|\bN%
_{ks}^{it}| \\
& +(-1)^{\ell +1}\sum_{t\neq i}(-1)^{t+I(t>i)}|\bN_{k\ell
}^{it}|+\sum_{s>\ell }\sum_{t\neq i}(-1)^{s+1+t+I(t>i)}|\bN
_{ks}^{it}|
\end{align*}
which by using \eqref{Bapat2013} with $\gamma =\ell $ becomes 
\begin{align*}
& \sum_{s<k}\sum_{t\neq i}(-1)^{s+t+I(t>i)}\{(-1)^{\ell +1+s}|\bN
_{k\ell }^{it}|+(-1)^{\ell +1+k+1}|\bN_{\ell s}^{it}|\} \\
& +\sum_{k<s<\ell }\sum_{t\neq i}(-1)^{s+1+t+I(t>i)}\{(-1)^{\ell +1+k}|
\bN_{\ell s}^{it}|+(-1)^{\ell +1+s+1}|\bN_{k\ell
}^{it}|\} \\
+& (-1)^{\ell +1}\sum_{t\neq i}(-1)^{t+I(t>i)}|\bN_{k\ell }^{it}|
\\
+& \sum_{s>\ell }\sum_{t\neq i}(-1)^{s+1+t+I(t>i)}\{(-1)^{\ell +1+k+1}|
\bN_{\ell s}^{it}|+(-1)^{\ell +1+s+1}|\bN_{k\ell
}^{it}|\} \\
=& \{(k-1)+(\ell -k+1)+1+(K-\ell )\}(-1)^{\ell +1}\sum_{t\neq
i}(-1)^{t+I(t>i)}|\bN_{k\ell }^{it}| \\
& +(-1)^{k+\ell }\{\sum_{s<k}\sum_{t\neq i}(-1)^{s+t+I(t>i)}|\bN
_{\ell s}^{it}|+\sum_{k<s<\ell }\sum_{t\neq i}(-1)^{s+t+I(t>i)}|\pmb{N
}_{\ell s}^{it}| \\
&+ {(-1)^{k+\ell}\bigg\{}\sum_{s>\ell }\sum_{t\neq
i}(-1)^{s+1+t+I(t>i)}|\bN_{\ell s}^{it}|\} \\
=& (K-1)(-1)^{\ell +1}\sum_{t\neq i}(-1)^{t+I(t>i)}|\bN_{k\ell
}^{it}|+(-1)^{k+\ell }\sum_{s\neq k,\ell }\sum_{t\neq i}(-1)^{s+I(s>\ell
)+t+I(t>i)}|\bN_{\ell s}^{it}|.
\end{align*}
On the other hand, 
\begin{align*}
K^{2}|\bN_{1}^{1}|(-1)^{\ell +i}n_{\ell i}^{+}=& \sum_{s\neq \ell
}\sum_{t\neq i}(-1)^{s+I(s>\ell )+t+I(t>i)}|\bN_{\ell s}^{it}| \\
=& (-1)^{k}\sum_{t\neq i}(-1)^{t+I(t>i)}|\bN_{k\ell
}^{it}|+\sum_{s\neq k,\ell }\sum_{t\neq i}(-1)^{s+I(s>\ell )+t+I(t>i)}|
\bN_{\ell s}^{it}|.
\end{align*}
Hence, 
\begin{align*}
K^{2}|\bN_{1}^{1}|& (n_{ki}^{+}-n_{\ell i}^{+}) \\
=& (-1)^{k+i}\bigg\{(K-1)(-1)^{\ell +1}\sum_{t\neq i}(-1)^{t+I(t>i)}|
\bN_{k\ell }^{it}|+ \\
& \phantom{(-1)^{k+i}\bigg\{}(-1)^{k+\ell }\sum_{s\neq k,\ell }\sum_{t\neq
i}(-1)^{s+I(s>\ell )+t+I(t>i)}|\bN_{\ell s}^{it}|\bigg\}- \\
& (-1)^{\ell +i}\bigg\{(-1)^{k}\sum_{t\neq i}(-1)^{t+I(t>i)}|\bN%
_{k\ell }^{it}|+\sum_{s\neq k,\ell }\sum_{t\neq i}(-1)^{s+I(s>\ell
)+t+I(t>i)}|\bN_{\ell s}^{it}|\bigg\} \\
=& K(-1)^{k+\ell +i+1}\sum_{t\neq i}(-1)^{t+I(t>i)}|\bN_{k\ell
}^{it}|.
\end{align*}
Similarly, we get $K^{2}|\bN_{1}^{1}|(n_{kj}^{+}-n_{\ell
j}^{+})=K(-1)^{k+\ell +j+1}\sum_{t\neq j}(-1)^{t+I(t>j)}|\bN%
_{k\ell }^{jt}|$. Next, write 
\begin{align*}
\sum_{t\neq i}(-1)^{t+I(t>i)}|\bN_{k\ell }^{it}|=&
\sum_{t<i}(-1)^\top|\bN_{k\ell }^{it}|+\sum_{i<t<j}(-1)^{t+1}|%
\bN_{k\ell }^{it}|+ \\
& (-1)^{j+1}|\bN_{k\ell }^{ij}|+\sum_{t>j}(-1)^{t+1}|\pmb{N%
}_{k\ell }^{it}|
\end{align*}
and apply once more \eqref{Bapat2013} with $\gamma =j$ to re-express it as 
\begin{align*}
& \sum_{t<i}(-1)^t\{(-1)^{j+1+t}|\bN_{k\ell
}^{ij}|+(-1)^{j+1+i+1}|\bN_{k\ell }^{jt}|\}+ \\
& \sum_{i<t<j}(-1)^{t+1}\{(-1)^{j+1+i}|\bN_{k\ell
}^{jt}|+(-1)^{j+1+t+1}|\bN_{k\ell }^{ij}|\}+(-1)^{j+1}|
\bN_{k\ell }^{ij}|+ \\
& \sum_{t>j}(-1)^{t+1}\{(-1)^{j+1+i+1}|\bN_{k\ell
}^{jt}|+(-1)^{j+1+t+1}|\bN_{k\ell }^{ij}|\} \\
=& \{(i-1)+(j-i+1)+1+(K-j)\}(-1)^{j+1}|\bN_{k\ell }^{ij}|+ \\
& (-1)^{i+j}\sum_{t<i}(-1)^\top|\bN_{k\ell
}^{jt}|+(-1)^{i+j}\sum_{i<t<j}(-1)^\top|\bN_{k\ell
}^{jt}|+(-1)^{i+j}\sum_{t>j}(-1)^{t+1}|\bN_{k\ell }^{jt}| \\
=& (K-1)(-1)^{j+1}|\bN_{k\ell }^{ij}|+(-1)^{i+j}\sum_{t\neq
i,j}(-1)^{t+I(t>j)}|\bN_{k\ell }^{jt}|.
\end{align*}%
Since 
\begin{equation*}
\sum_{t\neq j}(-1)^{t+I(t>j)}|\bN_{k\ell }^{jt}|=(-1)^{i}|
\bN_{k\ell }^{ij}|+\sum_{t\neq i,j}(-1)^{t+I(t>j)}|\bN
_{k\ell }^{jt}|,
\end{equation*}
we get that 
\begin{align*}
K^{2}|\bN_{1}^{1}|& (n_{ki}^{+}-n_{\ell i}^{+}-n_{kj}^{+}+n_{\ell
j}^{+}) \\
=& K(-1)^{k+\ell +i+1}\bigg\{(K-1)(-1)^{j+1}|\bN_{k\ell
}^{ij}|+(-1)^{i+j}\sum_{t\neq i,j}(-1)^{t+I(t>j)}|\bN_{k\ell
}^{jt}|\bigg\}- \\
& K(-1)^{k+\ell +j+1}\bigg\{(-1)^{i}|\bN_{k\ell
}^{ij}|+\sum_{t\neq i,j}(-1)^{t+I(t>j)}|\bN_{k\ell }^{jt}|\bigg\}
\\
=& K^{2}(-1)^{k+\ell +i+j}|\bN_{k\ell }^{ij}|
\end{align*}
which is equivalent to the stated result. The result for $\pmb{C}
\bN^{+}\pmb{C}^\top\pmb{D}$ is now straihtforward.
\end{proof}

%\begin{remark}
%The equality Ranjan et al.~(2014)\footnote{Ranjan G, Zhang Z-L, Boley D (2014). Incremental computation of pseudo-inverse of Laplacian. In: {\it Combinatorial Optimization and Applications}, Lecture Notes in Computer Science, 8881, Springer; 729--749.}
%\end{remark}

Let $\textfrak{T}$ be the set of all spanning trees and $\textfrak{T}_{ij}$
the set of spanning trees containing edge $(i,j)$. For any $\mathcal{T}\in\textfrak{T}$ let $w(\mathcal{T})=\prod_{(i,j)\in\mathcal{T}}n_{ij}$ be its weight. Then we have the
following:

\begin{lemma}
\label{CNCD bounded} All entries of $\pmb{C}\bN^+%
\pmb{C}^T\pmb{D}$ have absolute value bounded by one.
\end{lemma}

\begin{proof}
By the all minors matrix tree theorem (Chen, 1976, Chaiken, 1982) we get  
that $|\bN_{1}^{1}|=\sum_{\mathcal{T} \in \textfrak{T}}w(\mathcal{T} )$ and $n_{ij}|\bN_{k\ell }^{ij}|=\sum_{\mathcal{T} \in C_{ij}}s_{\mathcal{T},ij}w(\mathcal{T} )$, where $\textfrak{C}_{ij}\subseteq \textfrak{T}_{ij}$ is described in (Chaiken  1982) and $s_{\mathcal{T} ,ij}$
is some sign. Thus, 
$$\big|(-1)^{k+\ell +i+j}n_{ij}|\bN_{k\ell
}^{ij}|/|\bN_{1}^{1}|\big|\leqslant \sum_{\mathcal{T} \in C_{ij}}w(\mathcal{T}
)/\sum_{\mathcal{T} \in \textfrak{T}}w(\mathcal{T} )\leqslant 1.$$
\end{proof}

\medskip

Assume now that $n=\sum_{i,j}n_{i,j}\to\infty$ and Condition \ref{Con(AlgConnec)} holds. Note that there may exist $n_{ij}$'s which don't go to infinity and even if they
do they may go to infinity at different rates. For a spanning tree $\mathcal{T}$
let $r_\mathcal{T}(n)$ be the rate at which $w(\mathcal{T})$ goes to infinity with $
r_\mathcal{T}(n)=1$ if $w(\mathcal{T})\nrightarrow\infty$. Let $\tilde{r}(n)$ be the
maximum rate over $\mathcal{T}\in\textfrak{T}$. Let $\tilde{\textfrak{T}}=\{\mathcal{T}\in
\textfrak{T}: w(\mathcal{T})=O(\tilde{r}(n))\}$ and $\tilde{\mathcal{E}}$ be the set
of edges $(i,j)$ which belong to some $\mathcal{T}\in\tilde{\textfrak{T}}$. Let $
\tilde{\pmb{D}}=\mathrm{diag}(\tilde{n}_{12},\ldots,\tilde{n}_{1K},
\tilde{n}_{23},\ldots,\tilde{n}_{2K},\ldots,\tilde{n}_{K-1,K})$, where $
\tilde{n}_{ij}=n_{ij}$ if $(i,j)\in\tilde{\mathcal{E}}$ and zero otherwise.
Set finally $\tilde{\bN}=\pmb{C}^T\tilde{\pmb{D}}
\pmb{C}$ and let $\tilde{\bN}^+$ be its Moore-Penrose inverse. Note that $\tilde{\bN}$ is also the Laplacian of a multigraph. Then, we have the following.

\begin{lemma}\label{same limit} 
$\lim_{n\uparrow\infty}(\pmb{C}\bN^+
\pmb{C}^T\pmb{D}-\pmb{C}\tilde{\bN}^+
\pmb{C}^T\tilde{\pmb{D}})=\pmb{O}$.
\end{lemma}

\begin{proof}
By Lemma \ref{CNCD}, the $(k,\ell ;i,j)$ entries of $\pmb{C}
\bN^{+}\pmb{C}^\top\pmb{D}$ and $\pmb{C}
\tilde{\bN}^{+}\pmb{C}^\top\tilde{\pmb{D}}$ are $
(-1)^{k+\ell +i+j}n_{ij}$ $|\bN_{k\ell }^{ij}|/|\bN
_{1}^{1}|$ and $(-1)^{k+\ell +i+j}\tilde{n}_{ij}|\tilde{\bN}
_{k\ell }^{ij}|/|\tilde{\bN}_{1}^{1}|$, respectively. By applying
once more the all minors matrix tree theorem we get $|\tilde{\bN}
_{1}^{1}|=\sum_{\mathcal{T} \in \tilde{\textfrak{T}}}w(\mathcal{T} )$ and $\tilde{n}_{ij}|
\tilde{\bN}_{k\ell }^{ij}|=\sum_{\mathcal{T} \in C_{ij}\cap \tilde{
\textfrak{T}}}s_{\mathcal{T} ,ij}w(\mathcal{T} )$ by construction. Since $r_{\mathcal{T} }(n)=o(
\tilde{r}(n))$ for $\mathcal{T} \notin \tilde{\textfrak{T}}$ we may write $|
\bN_{1}^{1}|=|\tilde{\bN}_{1}^{1}|+o(\tilde{r}(n))$
and $n_{ij}|\bN_{k\ell }^{ij}|=\tilde{n}_{ij}|\tilde{\pmb{N%
}}_{k\ell }^{ij}|+o(\tilde{r}(n))$. Thus, 
\begin{equation*}
\frac{n_{ij}|\bN_{k\ell }^{ij}|}{|\bN_{1}^{1}|}-\frac{
\tilde{n}_{ij}|\tilde{\bN}_{k\ell }^{ij}|}{|\tilde{\bN}
_{1}^{1}|}=\frac{\tilde{n}_{ij}|\tilde{\bN}_{k\ell }^{ij}|+o(
\tilde{r}(n))}{|\tilde{\bN}_{1}^{1}|+o(\tilde{r}(n))}-\frac{
\tilde{n}_{ij}|\tilde{\bN}_{k\ell }^{ij}|}{|\tilde{\bN}%
_{1}^{1}|}=\frac{|\tilde{\bN}_{1}^{1}|o(\tilde{r}(n))-\tilde{n}
_{ij}|\tilde{\bN}_{k\ell }^{ij}|o(\tilde{r}(n))}{|\tilde{
\bN}_{1}^{1}|\{|\tilde{\bN}_{1}^{1}|+o(\tilde{r}(n))\}}
\end{equation*}
which converges to zero since the numerator is $o(\tilde{r}(n)^{2})$ and the
denominator $O(\tilde{r}(n)^{2})$.
\end{proof}

%\smallskip

\subsubsection*{Proof of Theorem \ref{Thm-graph.SLLN}:}

%\smallskip

\begin{proof}
Let $\widehat{\pmb{\delta }}=\pmb{C}\widehat{\bmu}$. By
applying Lemma \ref{C properties}(c) with $\widehat{\bmu}$ and $%
\widehat{\pmb{\delta }}$ in the place of $\bmu$ and $%
\pmb{\delta }$ we conclude that $\widehat{\bmu}$ is
consistent for $\bmu$ if and only if $\widehat{\pmb{\delta }}
$ is consistent for $\pmb{\delta }$.

By the strong law of large numbers, if $n_{ij}\rightarrow \infty $ (at
whatever rate) then $\bar{Y}_{ij}\rightarrow \mu _{i}-\mu _{j}$ almost
surely. Therefore, we may write $\bar{\bS}\overset{a.s.}{
\longrightarrow }\tilde{\pmb{\delta }}=(\tilde{\delta}_{11},\ldots ,
\tilde{\delta}_{1K},\tilde{\delta}_{23},\ldots ,\tilde{\delta}_{2K},\ldots ,
\tilde{\delta}_{K-1,K})$ where $\tilde{\delta}_{ij}=\mu _{i}-\mu _{j}$ or $
\bar{Y}_{ij}$ according to if $n_{ij}\rightarrow \infty $ or not.

Suppose first that Condition \ref{Con(AlgConnec)} is satisfied. By applying Lemma \ref{C
properties} with $\tilde{\bN}$ instead of $\bN$ we get 
$\pmb{C}\tilde{\bN}^{+}\pmb{C}^\top\tilde{
\pmb{D}}\tilde{\pmb{\delta }}=\tilde{\pmb{\delta }}$. However, if $
(i,j)\notin \tilde{\mathcal{E}}$ then the $(i,j)$th column of $\pmb{C}
\tilde{\bN}^{+}\pmb{C}^\top\tilde{\pmb{D}}$ contains
only zeros. Since $\tilde{\pmb{\delta }}_{ij}=\mu _{i}-\mu _{j}$ for
all $(i,j)\in \tilde{\mathcal{E}}$ we conclude that $\pmb{C}\tilde{%
\bN}^{+}\pmb{C}^\top\tilde{\pmb{D}}\tilde{
\pmb{\delta }}=\pmb{C}\tilde{\bN}^{+}\pmb{C}%
^\top\tilde{\pmb{D}}\pmb{\delta }=\pmb{\delta }$. Thus, 
\begin{align*}
\widehat{\pmb{\delta }}=& \pmb{C}\bN^{+}\pmb{C}
^\top\pmb{D}\bar{\bS} \\
=& (\pmb{C}\bN^{+}\pmb{C}^\top\pmb{D}-
\pmb{C}\tilde{\bN}^{+}\pmb{C}^\top\tilde{\pmb{
D}})\bar{\bS}+\pmb{C}\tilde{\bN}^{+}\pmb{%
C}^\top\tilde{\pmb{D}}(\bar{\bS}-\tilde{\pmb{\delta }
})+\pmb{C}\tilde{\bN}^{+}\pmb{C}^\top\tilde{%
\pmb{D}}\tilde{\pmb{\delta }} \\
=& (\pmb{C}\bN^{+}\pmb{C}^\top\pmb{D}-
\pmb{C}\tilde{\bN}^{+}\pmb{C}^\top\tilde{\pmb{
D}})\bar{\bS}+\pmb{C}\tilde{\bN}^{+}\pmb{
C}^\top\tilde{\pmb{D}}(\bar{\bS}-\tilde{\pmb{\delta }
})+\pmb{\delta }\overset{a.s.}{\longrightarrow }\pmb{\delta },
\end{align*}%
since the first term and second terms go almost surely to $\pmb{0}$:
the first by Lemma \ref{same limit} and the almost sure boundedness of $\bar{
\bS}$ while the second by the strong law and the fact that $
\pmb{C}\tilde{\bN}^{+}\pmb{C}^\top\tilde{\pmb{
D}}$ is bounded as shown in Lemma \ref{CNCD bounded}.

Suppose finally that Condition \ref{Con(AlgConnec)} does not hold. In this case the set $
\{1,\ldots ,K\}$ can be partitioned into two subsets $A_{1}$, $A_{2}$ so
that the limiting vector $\tilde{\pmb{\delta }}$ contains the 
differences $\mu _{i}-\mu _{j}$ only if both $i,j$ belong to $A_{1}$ or to 
$A_{2}$. This means that no difference $\mu _{i}-\mu _{j}$ with $i\in A_{1}$, 
$j\in A_{2}$ can be consistently estimated and thus $\widehat{\pmb{\delta 
}}$ can not be consistent.
\end{proof}

%-----------
\subsection*{Proof of Lemma \ref{cov:insure}:}
%-----------
\begin{proof}
Let $(\mathcal{G}_1,\mathcal{Y}_1)$ be the comparison graph obtained from $(\mathcal{G},\mathcal{Y})$ by retaining the edges $\mathcal{E}_1 \subset \mathcal{E}$ only if $n_{ij}=O(n)$. Let $\bN_1$ and $\bN$ be their corresponding Laplacians. Condition \ref{Con(Rate+CLT)} implies that there is a spanning tree $\mathcal{T}$ in $\mathcal{G}_1$ whose edges represent $O(n)$ paired comparisons and consequently  $\mathcal{G}_1$ is a connected graph. Let $\pmb{\Theta}=\lim_{n\to\infty}\bN/n$ and $\pmb{\Theta}_{1}=\lim_{m_1\to\infty}\bN_1/m_1$ where $m_1=\sum_{(i,j)\in\mathcal{E}_1}n_{ij}$. By construction the relations $n-m_1=o(n)$ and $\pmb{\Theta}=\pmb{\Theta}_1$ hold as $n\to\infty$.

Let $|\mathcal{E}_{1}|$ be the cardinality of $\mathcal{E}_{1}$ and let $\pmb{Q}=(q_{ij})$ be the $K \times |\mathcal{E}_{1}|$ matrix whose rows correspond to the items and column correspond to the edges in $\mathcal{E}_{1}$ arranged lexicographically. If the $l^{th}$ edge in $\mathcal{E}_{1}$ joins the vertices $i$ and $j$ then
\begin{align*}
	q_{il} = 
	\begin{cases}
		\sqrt{|\theta_{ij}|} & \text{ if } i < j\\
		-\sqrt{|\theta_{ij}|} &  \text{ if } i > j\\
		0 & \text{ otherwise},
	\end{cases}
\end{align*}
where $\theta_{ij}$ is the $(i,j)^{th}$ element of $\pmb{\Theta}$. 
Following the arguments in the proof of Lemma 2.2 in Bapat (2010, p. 14), it is easy to see that $\mathrm{rank}(\pmb{Q})=K-1$. Further $\pmb{\Theta}=\pmb{Q}\pmb{Q}^\top$ and consequently $\mathrm{rank}(\pmb{\Theta})=K-1$ as claimed. Now, the relation $n\bN^{+}\to\pmb{\Theta}^{+}$ follows immediately by applying Theorem 4.2 in Rako\^cevi\'c (1997). 
\end{proof}

\begin{remark}
The matrix $\pmb{Q}$ appearing in the proof of Lemma \ref{cov:insure} can be obtained from the incidence matrix (Bapat 2010) of the multigraph $\mathcal{G}_1$ by dividing by $\sqrt{n}$ and taking the limit. 
\end{remark}

%-----------
\subsection*{Proof of Theorem \ref{Thm-LST}:}
%------------
\begin{proof}
Under Condition \ref{Con(Rate+CLT)}, using Lemma \ref{cov:insure} we have $(\lim (\bN/n))^{+}=\lim ((\bN/n)^{+})$ which implies that $n\bN^{+}\rightarrow \Theta ^{+}$ or equivalently that $\mathbb{V}ar(\sqrt{n}\widehat{\bmu}_{n}) \to\sigma ^{2}\pmb{\Theta}^{+}.$ Let $\bV$ be a random vector of dimension $K(K-1)/2$ whose components $S_{ij},$ $1\leq i<j\leq K$, are arranged by their lexicographical order.  Clearly if $n\rightarrow \infty $ and $n_{ij}=O(n)$ then by the one dimensional CLT
\begin{equation}
\label{clt:simple}
\frac{1}{\sqrt{n}}\left( S_{ij}-n_{ij}(\mu _{i}-\mu _{j})\right) \Rightarrow 
\mathcal{N}\left( 0,\sigma ^{2}\theta _{ij}^*\right) .
\end{equation}
where $\theta _{ij}^*=-\theta _{ij}=\lim_{n\to\infty} n_{ij}/n .$ If $n_{ij}=o(n)$ then the LHS of converges to $0$ in probability, in which case $\theta_{ij}=0$ and \eqref{clt:simple} remains valid. Since the components of $\bV$ are independent it follows that as $n\rightarrow \infty $
\begin{equation*}
\frac{1}{\sqrt{n}}\left( \bV-\mathbb{E}(\bV)\right)
\Rightarrow \mathcal{N}_{K(K-1)/2}( \pmb{0},\sigma ^{2}\pmb{D}_{\pmb{\theta }} )
\end{equation*}
where $\pmb{D}_{\pmb{\theta}}$ is a diagonal matrix with elements $\theta_{ij}^*$ arranged by their lexicographical order. Since $S_{i}=\sum_{j\neq i}S_{ij}$ and $S_{ij}=-S_{ji}$ there exist a $K\times K(K-1)/2$ matrix $\pmb{L}=(l_{ij})$ where $l_{ij}\in \{-1,0,1\}$ such that $\bS=\pmb{LV}$. It follows that 
\begin{equation*}
\frac{1}{\sqrt{n}}\left( \pmb{LV}-\pmb{L}\mathbb{E}(
\bV)\right) \Rightarrow \mathcal{N}_{K}( \pmb{0}
,\sigma ^{2}\pmb{L}\pmb{D}_{\pmb{\theta }} \pmb{L}^\top)
\end{equation*}
which, by matching moments, can be rewritten as 
\begin{equation*}
\frac{1}{\sqrt{n}}(\bS-\pmb{N\mu })\Rightarrow \mathcal{N}
_{K}\left( \pmb{0},\sigma ^{2}\pmb{\Theta }\right) .
\end{equation*}
Premultiplying the above by $n\bN^{+}$ and using Lemma \ref{cov:insure} we find that
\begin{equation*}
\sqrt{n}\left( \widehat{\bmu}_{n}-\bmu\right)
\Rightarrow \mathcal{N}_{K}\left( \pmb{0},\sigma ^{2}\pmb{
\Theta }^{+}\right).
\end{equation*}
\end{proof}

{ 
\subsection*{Proof of Theorem \ref{theorem clt with 3.1}}
\begin{proof}
First note that for each $(i,j)\in\mathcal{E}$ we have
\begin{equation} \label{Eq.S_ij}
S_{ij} = n_{ij}(\mu_i-\mu_j)+\sigma\sqrt{n_{ij}} Z_{ij}+o_p(\sqrt{n_{ij}})
\end{equation}
where $\{Z_{ij}:(i,j)\in\mathcal{E}\}$ is a collection of independent $\mathcal{N}(0,1)$ RVs. Next recall that $\bS =(\sum_{j\neq 1}S_{1j}, \ldots, \sum_{j\neq K}S_{Kj})^{\top}$ so by \eqref{Eq.S_ij}
\begin{equation*}
\bS = \begin{pmatrix}
    \sum_{j\neq 1}n_{1j}(\mu_1-\mu_j)\\ \vdots\\ \sum_{j\neq K}n_{Kj}(\mu_K-\mu_j)
\end{pmatrix}  +   \sigma \begin{pmatrix}
    \sum_{j\neq 1}\sqrt{n_{1j}} Z_{1j}\\ \vdots\\ \sum_{j\neq K}\sqrt{n_{Kj}} Z_{Kj}
\end{pmatrix} + \begin{pmatrix}
    \sum_{j\neq 1}o_p(\sqrt{n_{1j}})\\ \vdots\\ \sum_{j\neq K}o_p(\sqrt{n_{Kj}})
\end{pmatrix}
\end{equation*}
which we rewrite as
\begin{equation} \label{Eq.S}
\bN\bmu + \sigma \bB\bZ + \pmb{\xi}
\end{equation}
where $\bB$ is a $K\times \binom{K}{2}$ with entries
$$
B_{ij} = 
\begin{cases} 
\sqrt{n_{ij}}, & \text{ if } i<j, \\
-\sqrt{n_{ij}}, & \text{ if } i>j,
\end{cases}
$$
$\bZ$ is the $K(K-1)/2$ vector whose components $Z_{ij}$ are arranged by their lexicographical order and $\pmb{\xi}= o_p(\sqrt{M}){\b1} $ where $M=\max\{n_{ij}:(i,j)\in\mathcal{E}\}$. Using the relations $\bN\bN^{+}\bmu=\bmu$, and $\bN=\bB\bB^{\top}$ we can reexpress the LSE as
\begin{align*}
\widehat{\bmu}_{n}=\bmu+(\bB\bB^{\top})^{+} \bB  \sigma\bZ + \bN^{+}\pmb{\xi}.
\end{align*}
It is easy to verify that $(\bB\bB^{\top})^{+} \bB$ is Moore--Penrose inverse of the matrix $\bB^{\top}$. It now follows that 
\begin{align*}
\sqrt{m}(\widehat{\bmu}_n-\bmu) = \sqrt{m} (\bB^{\top})^{+} \sigma\bZ + \sqrt{m}\bN^{+}\pmb{\xi}    
\end{align*}
By Theorem 3.1 $\bN^{+}=O(1/m)$; a bit of algebra now shows that 
\begin{equation} \label{Eq.N+xsi}
\sqrt{m}\bN^{+}\pmb{\xi} =o_p(1)\b1.    
\end{equation}
Next, arguments similar to those in the proof of Theorem \ref{Thm-graph.WLLN} can be used to show that the entries of $(\bB^{\top})^{+}$ are $O(m)$. Therefore, $\sqrt{m}(\widehat{\bmu}_n-\bmu)$ has a asymptotically normal distribution with mean $\bzero$. Furthermore,
\begin{align*}
\lim_m\mathbb{V}ar(\sqrt{m}(\widehat{\bmu}_n-\bmu))= \lim_m\mathbb{V}ar( \sqrt{m} (\bB^{\top})^{+} \sigma\bZ)= \lim_m m \bN^+=\pmb{\Psi}
\end{align*}
completes the proof.
\end{proof}
}

%-----------
\subsection*{Proof of Theorem \ref{Thm-AS}:}
%-----------
\begin{proof}
Without loss of generality assume that $\mu _{i}>\mu _{j}$ whenever $i<j$. Note that $\pmb{r}\left( \widehat{\bmu}_{n}\right) \neq 
\pmb{r}\left( \bmu\right) $ if and only if for some $i$
the event $E_{i}=\{\widehat{\mu }_{i}<\widehat{\mu }_{i+1}\}$ occurs. It follows that
\begin{equation*}
\mathbb{P}\left( \pmb{r}\left( \widehat{\bmu}_{n}\right)
\neq \pmb{r}\left( \bmu\right) \right) =\mathbb{P}\left(
\cup _{i=1}^{K-1}E_{i}\right) \leq \sum_{i=1}^{K-1}\mathbb{P}(\widehat{\mu }_{i}<\widehat{\mu }_{i+1})\leq (K-1)\max_{1\leq i<(K-1)}\mathbb{P}(\widehat{\mu }_{i}<\widehat{\mu }_{i+1}).
\end{equation*}
Under the stated conditions $\widehat{\mu }_{i}\overset{p}{\rightarrow }\mu_{i} $ for each $i$ and therefore $\widehat{\mu }_{i}-\widehat{\mu }_{i+1}
\overset{p}{\rightarrow }\mu _{i}-\mu _{i+1}=\delta _{i,i+1}$ where $\delta_{i,i+1}>0$ by assumption. Consequently, 
\begin{align*}
\mathbb{P}(\widehat{\mu }_{i}<\widehat{\mu }_{i+1})=\mathbb{P}(\widehat{\mu }_{i}-\widehat{\mu }_{i+1}<0)
=\mathbb{P}(\widehat{\mu }_{i}-\widehat{\mu}_{i+1}-\delta_{i,i+1} <-\delta_{i,i+1})
\leq \mathbb{P}(|\widehat{\mu }_{i}-\widehat{\mu}_{i+1}-\delta_{i,i+1}|
>\delta_{i,i+1}) \rightarrow 0
\end{align*}
as $n\rightarrow \infty$. Since $K$ is finite we have $\lim_{n}\mathbb{P}( \pmb{r}( \widehat{\bmu}_{n}) \neq \pmb{r}( \bmu) ) =0$ establishing the first claim.

Next observe that
\begin{equation*}
\mathbb{P(}\widehat{\mu }_{i}<\widehat{\mu }_{i+1})=\mathbb{P}(\bN_{i}^{+}\bS<\bN_{i+1}^{+}\bS)=\mathbb{P}((\bN_{i}^{+}-\bN_{i+1}^{+})\bS<0)
\end{equation*}
where $\bN_{i}^{+}$ is the $i^{th}$ row of the matrix $\bN^{+}$. A bit of algebra shows that
\begin{equation*}
(\bN_{i}^{+}-\bN_{i+1}^{+})\bS
=\sum_{k<l}^{K}[(n_{i,k}^{+}-n_{i+1,k}^{+})-(n_{i,l}^{+}-n_{i+1,l}^{+})]S_{kl}
\end{equation*}
where $S_{kl}=\sum_{t=1}^{n_{kl}}Y_{klt},$ the sum $\sum_{k<l}^{K}$ is taken over pairs $k<l$ for which $n_{kl}>0$. Observe that
\begin{eqnarray*}
\mathbb{P}((\bN_{i}^{+}-\bN_{i+1}^{+})\bS<0)&=&\mathbb{P}((\bN_{i}^{+}-\bN_{i+1}^{+})
\bS-\delta _{i,i+1}<-\delta _{i,i+1}) \\
&=&\mathbb{P}
(\sum_{k<l}^{K}[(n_{i,k}^{+}-n_{i+1,k}^{+})-(n_{i,l}^{+}-n_{i+1,l}^{+})]S_{kl}-\delta _{i,i+1}<-\delta _{i,i+1})
\\
&=&\mathbb{P}(\sum_{k<l}^{K}\xi _{i,kl}^{n}\frac{S_{kl}}{n_{kl}}-\delta
_{i,i+1}<-\delta _{i,i+1}) \\
&= &\mathbb{P}(\sum_{k<l}^{K}\xi _{i,kl}^{n}(\overline{S}_{kl}-\mu_{kl})<-\delta _{i,i+1})
\end{eqnarray*}
where $\overline{S}_{kl} ={S_{kl}}/n_{kl}$ and $\xi _{i,kl}^{n} =({n_{kl}}/n
)n[(n_{i,k}^{+}-n_{i+1,k}^{+})-(n_{i,l}^{+}-n_{i+1,l}^{+})]$.
By construction $\delta _{i,i+1}=$ $\sum_{k<l}^{K}\xi _{i,kl}^{n}\mu _{kl}$ since the expectation of $\overline{S}_{kl}$ is $\mu _{kl}$. Note that the collection of RVs $\{S_{kl}\}$ where $k<l$ are independent. Using the fact that weighted sums of independent subgaussian RVs are also subgaussian and the form of the moment generating functions of subgaussian RVs we find that
\begin{align*}
\mathbb{P}(\sum_{k<l}^{K}\xi _{i,kl}^{n}(\overline{S}_{kl}-\mu _{kl})
<-\delta _{i,i+1})\leq&~ 
\exp(-\delta_{i,i+1}^2/\sum_{k<l}((\xi _{i,kl}^{n})^2\sigma^2/n_{kl}).
\end{align*}
Note that under Condition \ref{Con(Rate+CLT)}, using Lemma \ref{cov:insure} for all pairs $k<l$ we have that $n\, n_{i,j}^{+}=\theta _{i,j}^{+}+o\left( 1\right) $ and 
\begin{equation*}
\xi _{i,kl}=\xi _{i,kl}^{n}+o\left( 1\right) =|\theta _{kl}|[(\theta
_{i,k}^{+}-\theta _{i+1,k}^{+})-(\theta _{i,l}^{+}-\theta
_{i+1,l}^{+})](1+o\left( 1\right) )
\end{equation*}%
where $\theta _{i,j}^{+}$ is the $\left( i,j\right) $ element of the matrix $\pmb{\Theta }^{+}$. If $\xi _{i}=\max_{k<l}\xi _{i,kl}^{2}$ then
\begin{align*}
\mathbb{P}(\sum_{k<l}^{K}\xi _{i,kl}^{n}(\overline{S}_{kl}-\mu _{kl})
<-\delta _{i,i+1})\leq&~ 
\exp(-\delta_{i,i+1}^2/(\xi_i^2\sigma^2 \sum_{k<l}1/n_{kl}))\\
=&~ \exp(-n\delta_{i,i+1}^2/(\xi_i^2\sigma^2 \sum_{k<l}1/|\theta_{kl}|) + o(1)).
\end{align*}
This can be rewritten as 
\begin{equation*}
\mathbb{P(}\widehat{\mu }_{i}<\widehat{\mu }_{i+1})\leq C_{1,i}\exp
(-C_{2,i}n)
\end{equation*}
for some constants $C_{1,i}$ and $C_{2,i}$. The result follows by taking $C_{1}=K\max C_{1,i}$ and $C_{2}=\min C_{2,i}$.
\end{proof}

%-----------
\subsection*{Proof of Proposition \ref{span:M:condition}}
%---------------
\begin{proof}
If $x_{ijk} = z_i -z_j$ for $1\leq i<j\leq K$ then it is easy to see that 
$$\bx = \bM (z_1,\ldots,z_K)^{\top}= z_1 \bM_1 +\ldots+ z_K \bM_K.$$
Thus $\bx\in \mathrm{im}(\bM)$. Next let $\bx\in \mathrm{im}(\bM)$ then there exists $\alpha_1,\ldots,\alpha_K$ such that
\begin{align*}
\bx = \alpha_1 \bM_1 +\ldots+ \alpha_K \bM_K
    = \bM \pmb{\alpha},
\end{align*}
where $\bM_i$ is the $i^{th}$ column of $\bM$. Therefore $x_{ijk} = \alpha_i -\alpha_j$ for $1\leq i<j\leq K$.
\end{proof}

%--------------
\subsection*{Proof of Theorem \ref{Thm-Est.Covariates}:}
%--------------

\begin{proof} 
The Lagrangian associated with \eqref{def.mu.hat.with.cov} is
		\begin{align*}
			\mathcal{L}(\bmu,\pmb{\beta},\lambda)=
			(\bY-\bM\bmu -\bX \pmb{\beta})^{\top}
			(\bY-\bM\bmu -\bX \pmb{\beta})+ 
			\lambda \bv^{\top}\bmu.
		\end{align*} 
Its derivatives with respect to $(\bmu,\pmb{\beta})$ and $\lambda$ equated to zero yield
		\begin{align}
			&~\frac{\partial}{\partial (\bmu,\pmb{\beta})} \mathcal{L}(\bmu,\pmb{\beta},\lambda)= 
			2\begin{pmatrix}
				\bM^{\top} \\ \bX^{\top}
			\end{pmatrix} (\bY-\bM\bmu -\bX \pmb{\beta}) 
			+ {\lambda}
			\begin{pmatrix}
				\bv \\ \pmb{0}
			\end{pmatrix} = 0,\label{cov:thm1:eq1}\\
			&~\frac{\partial}{\partial \lambda} \mathcal{L}(\bmu,\pmb{\beta},\lambda)= 
			\bv^{\top}\bmu=0. \label{cov:thm1:eq2}
		\end{align}
Premultiplying \eqref{cov:thm1:eq1} by $(\b1_{k\times1}^{\top}, \pmb{0}_{p\times1}^{\top})$, and recalling that $\bv^{\top}\b1\neq0$ we find that $\lambda =0$. Therefore a solution to equation \eqref{cov:thm1:eq1} is given by
	\begin{align}
         \label{cov:thm1:eq3}
			\begin{pmatrix}
				\widehat{\bmu}\\ \widehat{\pmb{\beta}}
			\end{pmatrix}
			= \begin{pmatrix}
				\bN & \bM^{\top} \bX\\
				\bX^{\top} \bM & \bX^{\top}\bX
			\end{pmatrix}^{+}
			\begin{pmatrix}
				\bM^{\top} \\ \bX^{\top}
			\end{pmatrix} \bY+ (\pmb{I}-\begin{pmatrix}
				\bN & \bM^{\top} \bX\\
				\bX^{\top} \bM & \bX^{\top}\bX
			\end{pmatrix}^{+}\begin{pmatrix}
				\bN & \bM^{\top} \bX\\
				\bX^{\top} \bM & \bX^{\top}\bX
			\end{pmatrix})\pmb{\xi},
	\end{align} 
where $\pmb{\xi}\in\mathbb{R}^{K+p}$ is arbitrary. If the graph $\mathcal{G}$ is connected and Condition \ref{condition:covariate:existence} holds, then $\mathrm{rank}(\pmb{H})=K+p-1$. Consequently $\mathrm{rank}(\pmb{H}^{\top} \pmb{H})=K+p-1$ and $(\b1^{\top},\pmb{0}^{\top})^{\top}$ spans $\mathrm{ker}(\pmb{H}^{\top} \pmb{H})$. Thus $(\pmb{I}-(\pmb{H}^{\top} \pmb{H})^{+}\pmb{H}^{\top} \pmb{H}) \pmb{\xi}=c(\b1^{\top}, \pmb{0}^{\top})^{\top}$ for some $c\in\mathbb{R}$ where $c\neq 0$ if and only if $\pmb{\xi}\neq \pmb{0}$. Premultiplying \eqref{cov:thm1:eq3} by $(\bv^{\top}, \pmb{0}^{\top})$ and using \eqref{cov:thm1:eq2} we obtain 
\begin{align*}
\bzero = (\bv^{\top}, \pmb{0}^{\top}) (\pmb{H}^{\top} \pmb{H})^{+} \pmb{H}^{\top}\bY + c (\bv^{\top}, \pmb{0}^{\top}) (\b1^{\top},\pmb{0}^{\top})^{\top}.
\end{align*}
Therefore $c=-(\bv^{\top}, \pmb{0}^{\top}) (\pmb{H}^{\top} \pmb{H})^{+} \pmb{H}^{\top}\bY/ (\bv^{\top}, \pmb{0}^{\top}) (\b1^{\top},\pmb{0}^{\top})^{\top}$. Thus the unique solution to \eqref{cov:thm1:eq1} and \eqref{cov:thm1:eq2} is given by
	\begin{align}
    \label{cov:thm1:eq4}
		\begin{pmatrix}
			\widehat{\bmu}\\ \widehat{\pmb{\beta}}
		\end{pmatrix}
		= (\pmb{H}^{\top} \pmb{H})^{+} \pmb{H}^{\top}\bY - \frac{(\bv^{\top}, \pmb{0}^{\top}) (\pmb{H}^{\top} \pmb{H})^{+} \pmb{H}^{\top}\bY}{\bv^{\top} \b1}(\b1^{\top},\pmb{0}^{\top})^{\top}.
	\end{align}
The uniqueness of the LSE follows from the uniqueness of the Moore Penrose inverse. 

Conversely, let $\mathrm{rank}(\pmb{H})\neq K+p-1$. Then, we have $\mathrm{rank}(\bM)\leq K-1$ and consequently $\mathrm{rank}(\pmb{H})<K+p-1$. Suppose     $\mathrm{rank}(\pmb{H})=K+p-l$ where $2\leq l\leq K+p-1$. Then $(\pmb{I}-(\pmb{H}^{\top} \pmb{H})^{+}\pmb{H}^{\top} \pmb{H}) \pmb{\xi}$ can be expressed as a linear combination of $l$ orthogonal vectors in $\mathrm{ker}(\pmb{H}^{\top} \pmb{H})$. Therefore premultiplying \eqref{cov:thm1:eq3} by $(\bv^{\top}, \pmb{0}^{\top})$ and using \eqref{cov:thm1:eq2}, we get a linear equation in $l$ variables, so the associated solution is not unique. Consequently the estimator given by \eqref{cov:thm1:eq3} is not unique. For example if $\mathrm{rank}(\pmb{H})=K+p-2$, then $(\pmb{I}-(\pmb{H}^{\top} \pmb{H})^{+}\pmb{H}^{\top} \pmb{H}) \pmb{\xi}=c_1(\b1^{\top}, \pmb{0}^{\top})^{\top}+c_2\pmb\chi$, where $\pmb\chi\in \mathrm{ker}(\pmb{H}^{\top} \pmb{H})$ and $\pmb\chi \perp (\b1^{\top}, \pmb{0}^{\top})^{\top}$. Therefore premultiplying \eqref{cov:thm1:eq3} by $(\bv^{\top}, \pmb{0}^{\top})$ and using \eqref{cov:thm1:eq2}, we get 
\begin{align*}
	\bzero = (\bv^{\top}, \pmb{0}^{\top}) (\pmb{H}^{\top} \pmb{H})^{+} \pmb{H}^{\top}\bY + c_1 (\bv^{\top}, \pmb{0}^{\top}) (\b1^{\top},\pmb{0}^{\top})^{\top} + c_2 (\bv^{\top}, \pmb{0}^{\top})\pmb{\chi}.
\end{align*}
This is a linear equation in two variables, so the associated solution is not unique. This proves the first part.
%%%%%%%%%%%%
        
Now, substituting $\bv = \b1$ into \eqref{cov:thm1:eq4} we find that $(\widehat{\bmu}^{\top},\widehat{\pmb{\beta}}^{\top})^{\top} = (\pmb{H}^{\top} \pmb{H})^{+} \pmb{H}^{\top}\bY$. An explicit solution is facilitated by noting that we can rewrite  
\begin{align*}
	\bY =~ \bM \bmu + \bX \pmb{\beta} + \pmb{\epsilon}
	=~ \bM \pmb{c}_1 + \widetilde{\bX} \pmb{c}_2 + \pmb{\epsilon},
\end{align*}
where $\widetilde{\bX}= (\pmb{I}- \bM \bN^{+} \bM^{\top})\bX$ and 
\begin{align*}
	\pmb{c}_2 &~= \pmb\beta,\\
	\pmb{c}_1 &~= \bmu + \bM^{+}\bM \bN^{+} \bM^{\top} \bX \pmb\beta.
\end{align*}
Note that $\widetilde{\bX}$ is the projection of $\bX$ onto the orthogonal complement of $\bM$, and consequently in terms of $(\pmb{c}_1,\pmb{c}_2)$ the normal equations are
\begin{align} \label{thm3.1:proof:eq5}
	\begin{pmatrix}
		\bN & \pmb{0}\\
		\pmb{0} & \widetilde{\bX}^{\top}\widetilde{\bX}
	\end{pmatrix}
	\begin{pmatrix}
		\pmb{c}_1 \\ \pmb{c}_2
	\end{pmatrix}
	= \begin{pmatrix}
		\bM^{\top}\bY\\ \widetilde{\bX}^{\top}\bY
	\end{pmatrix}.
\end{align}
Here ${\bX}$ is a full column rank matrix and $(\pmb{I}-\bM\bN^{+}\bM^{\top})$ is an $n\times n$ idempotent matrix. Therefore by Proposition 5.4 in Puntanen et al. (2007, p. 132) $\widetilde{\bX}^{\top}\widetilde{\bX}$ is a full rank matrix. It now follows from \eqref{thm3.1:proof:eq5} that 
\begin{align*}
    \widehat{\pmb{\beta}} = (\widetilde{\bX}^{\top}\widetilde{\bX})^{-1} \widetilde{\bX}^{\top}\bY= 
    (\bX^\top (\pmb{I}-\bM\bN^{+}\bM^{\top}) \bX)^{-1}\bX^\top(\pmb{I}-\bM\bN^{+}\bM^{\top}))\bY.
\end{align*}
Substituting the value of $\widehat{\pmb{\beta}}$ in \eqref{thm3.1:proof:eq5} and solving for $\bmu$, we obtain 
\begin{align*}
\widehat{\bmu}&=\bN^{+}(\bS-\bM^{\top}\bX\widehat{\pmb{\beta}}).
\end{align*}
\end{proof}

%-------------------
\subsection*{Proof of Lemma \ref{lemma:cod:with:cov:wlln}:}
%--------------------

\begin{proof}
Let $\widetilde{\pmb{H}}=(\bM, \widetilde{\bX})$ and note that $\mathrm{ker}(\pmb{H}^{\top}\pmb{H}) = \mathrm{ker}(\widetilde{\pmb{H}}^{\top}\widetilde{\pmb{H}}) = \mathrm{span} (\{(\b1_{k\times1}^{\top}, \pmb{0}_{p\times1}^{\top})^{\top}\})$. Consequently for every $\bv_1\in \mathbb{R}^{K+p}$ there exists a $\bv_2\in \mathbb{R}^{K+p}$ such that $\bv_1^{\top}\pmb{H}^{\top}\pmb{H} \bv_1 = \bv_2^{\top}\widetilde{\pmb{H}}^{\top}\widetilde{\pmb{H}} \bv_2$. Therefore, the second smallest eigenvalue of $\pmb{H}^{\top}\pmb{H}$ diverges if and only if the second smallest eigenvalue of $\widetilde{\pmb{H}}^{\top}\widetilde{\pmb{H}}$ diverges. Next, the second smallest eigenvalue of $\widetilde{\pmb{H}}^{\top}\widetilde{\pmb{H}}$ diverges if the second smallest eigenvalue of $\bM^{\top}\bM$ diverges and the smallest eigenvalue of $ \widetilde{\bX}^{\top}\widetilde{\bX}$ diverges. Condition \ref{Con(AlgConnec)} guarantees that the second smallest eigenvalue of $\bM^{\top}\bM$ diverges, cf., Theorem \ref{Thm-graph.WLLN}.
	
Let $\pmb{l}\in \mathbb{R}^p\backslash \{\pmb{0}\}$ and $\bx=\bX\pmb{l}$. Then, $\bx\in \mathrm{span}(\bX)$. If $\theta_n= \measuredangle(\bx, \bM)$, then $\theta_n\geq \phi_n$ and 
\begin{align*}
	\cos \theta_n = \frac{\|\bM \bN^{+} \bM^{\top}\bx\|}{\|\bx\|}.
\end{align*}
Using Pythagoras theorem, $ \|\bM \bN^{+} \bM^{\top}\bx\|^2 + \|(\pmb{I}- \bM \bN^{+} \bM^{\top})\bx\|^2 = \|\bx\|^2$. Next, $\theta_n>\phi$  for all $n>n_0$ amounts to saying that 
\begin{align} \label{lemma:cod:with:cov:wlln:proof:eq1}
	\|(\pmb{I}- \bM \bN^{+} \bM^{\top})\bx\|^2 = (1-\cos^2 \theta_n) \|\bx\|^2 > (1-\cos^2 \phi) \|\bx\|^2.
\end{align}
Thus, $\|(\pmb{I}- \bM \bN^{+} \bM^{\top})\bx\|^2$ diverges if $\|\bx\|^2$ diverges. Note that $\|\bx\|^2$ diverges for any $\pmb{l}$ if and only if the smallest eigenvalue of $\bX^{\top}\bX$ diverges. Similarly, the minimum eigenvalue of $ \widetilde{\bX}^{\top}\widetilde{\bX}$ diverging is equivalent to $\|(\pmb{I}- \bM \bN^{+} \bM^{\top})\bx\|^2$ diverging for any $\pmb{l}$. 
\end{proof}

%-----------------------
\subsection*{Proof of Lemma \ref{proposition:random:cov:properties}}
%----------------------

\begin{proof}
In the light of Theorem \ref{graph-GLM:wlln+clt}, it is sufficient to show $\phi_n\to\phi >0$ in probability as $n\to\infty$. Equivalently the angle between any column of $\bM$ and any column of $\bX$ converges to a strictly positive number. Let $\bM_i$ and $\bX_i$ denote the $i^{th}$ columns of $\bM$ and $\bX$ respectively. The angle between $\bM_i$ and $\bX_i$ is given by
\begin{align*}
\phi_{n,i} = \cos^{-1} (\frac{|\bM_i^{\top}\bX_i|}{\|\bM_i\| \|\bX_i\|}) = \cos^{-1} (|(\frac{\bM_i}{\|\bM_i\|})^{\top} \frac{\bX_i}{\|\bX_i\|}|).
\end{align*}
Note that ${\bM_i}/{\|\bM_i\|}$ and ${\bX_i}/{\|\bX_i\|}$ are normalized vectors. Let $\bM_i = (M_{i,1},\ldots,M_{i,n})^{\top}$ and $\bX_i = (X_{i,1},\ldots,X_{i,n})^{\top}$ denote the $i^{th}$ columns of $\bM$ and $\bX$ respectively. Then $M_{ij}\in\{-1,0,1\}$ and
\begin{align*}
({\bM_i}/{\|\bM_i\|})^{\top}({\bX_i}/{\|\bX_i\|})=
\frac{\sum_{j=1}^{n}M_{ij}X_{ij}}{\|\bM_i\| \|\bX_i\|}
\end{align*}

So using Lemma 3.2.4 in Vershynin (2020) we have $({\bM_i}/{\|\bM_i\|})^{\top}({\bX_i}/{\|\bX_i\|}) \to 0$ with probability one as $n\to\infty$. Therefore $\phi_{n,i}\to \pi/2$ with probability one as $n\to\infty$.
\end{proof}

%-------------------
\subsection*{Proof of Lemma \ref{lemma:cod:with:cov:clt}:}
%--------------------

\begin{proof}
Assume the same notations as in the proof of Lemma \ref{lemma:cod:with:cov:wlln}. Once again, relation \eqref{lemma:cod:with:cov:wlln:proof:eq1} holds and  $ \|(\pmb{I}- \bM \bN^{+} \bM^{\top})\bx\|^2/n > (1-\cos^2 \phi) \|\bx\|^2/n>0$ for all large $n$. Therefore,  $\mathrm{rank}(\lim_n \widetilde{\bX}^{\top}\widetilde{\bX}/n) = p$ and consequently
\begin{align} \label{lemma:cod:with:cov:clt:rank:condition}
    \mathrm{rank}(\lim_n \begin{pmatrix}
    \pmb{0} & (\pmb{I}-(\bN/n)(\bN/n)^{+})\pmb{W}\\
    \pmb{W}^{\top}(\pmb{I}-(\bN/n)^{+}(\bN/n)) & \widetilde{\bX}^{\top}\widetilde{\bX}/n 
    \end{pmatrix})\geq p,
\end{align}
where $\pmb{W}= (\bM^{\top}\bX)/n$. 

From Lemma \ref{cov:insure}, Condition \ref{Con(Rate+CLT)} gives $\mathrm{rank}(\lim_n\bN/n)=K-1$. So using Masarglia and Styan (1974, Theorem 19) and \eqref{lemma:cod:with:cov:clt:rank:condition}, we get  $\mathrm{rank}(\pmb{\Sigma}) \geq K+p-1$. Next, note that  $\mathrm{rank}(\bM,\bX)\leq K+p-1$ for all $n$, and hence $\mathrm{rank}(\pmb{\Sigma}) \leq K+p-1$. Combining previous two statements, we get $\mathrm{rank}(\pmb{\Sigma}) = K+p-1$.
\end{proof}

%------------
\subsection*{Proof of Theorem \ref{graph-GLM:wlln+clt}:}
%---------------

 \begin{proof}
 \begin{enumerate}
     \item Under the stated assumptions the expectation of the LSE is
		\begin{align*}
			\mathbb{E}\begin{pmatrix}
				\widehat{\bmu}\\ \widehat{\pmb{\beta}}
			\end{pmatrix} = (\pmb{H}^{\top} \pmb{H})^{+} \pmb{H}^{\top} \pmb{H} \begin{pmatrix}
			{\bmu}\\ {\pmb{\beta}}
		\end{pmatrix}.
		\end{align*}
		We have $(\b1^{\top},\pmb{0}^{\top})^{\top}\in \mathrm{ker}(\pmb{H}^{\top} \pmb{H})$ and $(\b1^{\top},\pmb{0}^{\top}) ({\bmu}^{\top}, {\pmb{\beta}}^{\top})^{\top}=0$, so $ ({\bmu}^{\top}, {\pmb{\beta}}^{\top})^{\top} \in \mathrm{im} (\pmb{H}^{\top} \pmb{H})$. Consequently $(\pmb{H}^{\top} \pmb{H})^{+} \pmb{H}^{\top} \pmb{H} ({\bmu}^{\top}, {\pmb{\beta}}^{\top})^{\top} = ({\bmu}^{\top}, {\pmb{\beta}}^{\top})^{\top}$.
		%%%%%%%%%%
    \item Observe that 
     \begin{align*}
          \mathbb{V}ar\begin{pmatrix}
             \widehat{\bmu}\\ \widehat{\pmb{\beta}}
         \end{pmatrix}= 
         (\pmb{H}^{\top} \pmb{H})^{+}\sigma^2 \b1.
     \end{align*}
     Under the assumption that the second smallest eigenvalue of $\pmb{H}^{\top} \pmb{H}$ diverges, using arguments similar to those in the proof of Theorem \ref{Thm-graph.WLLN}, we obtain $(\pmb{H}^{\top} \pmb{H})^{+}\to\pmb{O}$ as $n\to\infty$. Consequently $\mathbb{V}ar((\widehat{\bmu}^{\top},\widehat{\pmb{\beta}}^{\top})^{\top}\to\pmb{O}$.  We already have $\mathbb{E} (\widehat{\bmu}, \widehat{\pmb{\beta}})= ({\bmu}, {\pmb{\beta}})$. Therefore, we conclude that $\widehat{\pmb{\beta}}\to\pmb{\beta}$ and $\widehat{\bmu}\to\bmu$ in probability as $n\to\infty$.
     %%%%%%%%%%%%
     \item The proof is obvious using arguments similar to those in the proof of Theorem \ref{Thm-AS}.
     %%%%%%%%%%%%%%%%%
     \item Under the stated assumptions we have 
 \begin{align*}
 	\begin{pmatrix}
 		\widehat{\bmu}\\ \widehat{\pmb{\beta}}
 	\end{pmatrix}
 	= \begin{pmatrix}
 		\bN & \bM^{\top} \bX\\
 		\bX^{\top} \bM & \bX^{\top}\bX
 	\end{pmatrix}^{+}
 	\begin{pmatrix}
 		\bM^{\top} \\ \bX^{\top}
 	\end{pmatrix} \bY = \begin{pmatrix}
 	{\bmu}\\ {\pmb{\beta}}
 \end{pmatrix} + \begin{pmatrix}
 	\bN & \bM^{\top} \bX\\
 	\bX^{\top} \bM & \bX^{\top}\bX
 \end{pmatrix}^{+}
\begin{pmatrix}
\bM^{\top} \\ \bX^{\top}
\end{pmatrix} \pmb{\epsilon}.
 \end{align*}
Rearranging the terms we obtain
\begin{align*}
	\sqrt{n}(\begin{pmatrix}
		\widehat{\bmu}\\ \widehat{\pmb{\beta}}
	\end{pmatrix} - \begin{pmatrix}
		{\bmu}\\ {\pmb{\beta}}
	\end{pmatrix}) 
	= (\frac{1}{n}\pmb{H}^{\top} \pmb{H})^{+} \frac{1}{\sqrt{n}} \pmb{H}^{\top}\pmb{\epsilon}.
\end{align*}
Under our assumptions $ (\pmb{H}^{\top} \pmb{H}/n)^{+} \to \pmb{\Sigma}^{+} $ as $n\to\infty$. Observe that for $\pmb{a}\not\in \mathrm{ker}(\pmb{H})$ we have
$$\pmb{a}^{\top}\pmb{H}^{\top} \pmb{\epsilon} /\sqrt{n}= \sum_{i=1}^{n}d_i\epsilon_i,$$ 
where $d_i$ and $\epsilon_i$ are the $i^{th}$ components of $\pmb{a}^{\top}\pmb{H}^{\top} /\sqrt{n}$ and $\pmb{\epsilon}$ respectively. Further note that $d_i=\pmb{a}^{\top} \pmb{h}_i$ where $\pmb{h}_i$ is the $i^{th}$ row of  $\pmb{H}$. 
Using the Rayleigh's theorem (e.g., Horn and Johnson 2007, Theorem 4.2.2, p. 234), we have
\begin{align*}
\sum_{i=1}^{n} d_i^{2} =~
 \pmb{a}^{\top} (\pmb{H}^{\top} \pmb{H}) \pmb{a} 
\geq~ \pmb{a}^{\top}\pmb{a}\, \lambda_2(\pmb{H}^{\top} \pmb{H}).
\end{align*}
Now by the Cauchy--Schwartz inequality
\begin{align*}
	d_i^{2} = (\pmb{a}^{\top} \pmb{h}_i)^{2} \leq \pmb{a}^{\top}\pmb{a}\, \max_{i} \pmb{h}_i \pmb{h}_i^{\top}.
\end{align*}
Consequently using Condition \ref{cond:with:cov:clt+hajek+sidak} we have 
\begin{align*}
	\max_{1\leq i\leq n}\frac{d_i^{2}}{\sum_{i=1}^{n} d_i^{2}} \leq 
	\frac{ \max_{i} \pmb{h}_i \pmb{h}_i^{\top} }{\lambda_2(\pmb{H}^{\top} \pmb{H})} \to 0 \text{ as } n\to\infty.
\end{align*}
Therefore, applying the Hajek--Sidak CLT (e.g., Sen and Singer 1994, Theorem 3.3.6, p. 119) and Cramer-Wold device we establish the asymptotic normality of $\pmb{H}^{\top}\pmb{\epsilon}/{\sqrt{n}}$. Next, observe that $\mathbb{V}ar (\pmb{H}^{\top}\pmb{\epsilon}/{\sqrt{n}}) =\sigma^2 (\pmb{H}^{\top} \pmb{H}/n)$ and hence 
\begin{align*}
	\sqrt{n}(\widehat{\pmb{\theta}} - {\pmb{\theta}}) \Rightarrow \mathcal{N}_{K+p} (\pmb{0}, \sigma^2\, \pmb{\Sigma}^{+}).
\end{align*}
 \end{enumerate}
 \end{proof}

%------------------------------
\subsection{Proof of Theorem \ref{misspecified:model:normal:errors:thm}:}
%-----------------------------
\begin{proof}

First observe that given normal IID errors the density under the true model, $f_t$, is the density of a $\mathcal{N} (\bM\bmu_{t}+ \bX\pmb{\beta}, \sigma^2 \pmb{I}_n)$ RV and the density of the misspecified model, $f_m$, is that of $\mathcal{N} (\bM\bmu_{m}, \sigma^2 \pmb{I}_n)$ RV. It follows that the AKL is: 
\begin{small}
\begin{align}
\mathrm{AKL}(t,m) 
=&~\frac{1}{2n\sigma^2}\mathbb{E}_{f_t} (-(\bY-\bM\bmu_t-\bX\pmb{\beta})^{\top}(\bY-\bM\bmu_t-\bX\pmb{\beta}) + (\bY-\bM\bmu_m)^{\top}(\bY-\bM\bmu_m))\nonumber\\
=&~\frac{1}{2n\sigma^2}(\bM(\bmu_t-\bmu_m)+\bX\pmb{\beta})^{\top}(\bM(\bmu_t-\bmu_m)+\bX\pmb{\beta})\nonumber\\
=&~ \frac{1}{2n\sigma^2}((\bmu_t-\bmu_m)^{\top}\bN(\bmu_t-\bmu_m)+ 2\pmb{\beta}^{\top}\bX^{\top} \bM(\bmu_t-\bmu_m)+\pmb{\beta}^{\top}\bX^{\top}\bX\pmb{\beta}).\label{AKL:normal:errors}
\end{align}
\end{small}
It is now immediate that if $\bX^{\top}\bM/n =o(1)$ then the second term in \eqref{AKL:normal:errors} is negligible as $n\to\infty$ and therefore for a fixed value of $\bmu_t$ the value of $\mathrm{AKL}(t,m)$ is minimized, as a function of $\bmu_m$, when $\bmu_m=\bmu_t+o(1)$. This proves Part 1.  

\medskip

When the conditions of Theorem \ref{graph-GLM:wlln+clt} hold, so do those of Theorem \ref{Thm-LST}. Adapting the calculation in the proof of Theorem \ref{Thm-LST} we find that
\begin{align*}
\frac{1}{\sqrt{n}} (S_{ij}-n_{ij}(\mu_{t,i}-\mu_{t,j})- \sum_{k=1}^{n_{ij}} 
\bx_{ijk}^{\top}\pmb{\beta})\Rightarrow
 \mathcal{N}\left( 0,\sigma ^{2}\theta _{ij}^*\right)
\end{align*}
and consequently 
\begin{equation*}
\frac{1}{\sqrt{n}}( \bS-\bN\bmu_t- \bM^{\top}\bX\pmb{\beta})
\Rightarrow \mathcal{N}_{K}( \pmb{0},\sigma ^{2}\pmb{\Theta}) )
\end{equation*}
Premultiplying by $n\bN^{+}$ and using Lemma \ref{cov:insure} and the definition of $\pmb{b}$ we find that
\begin{equation*}
\sqrt{n}( \widehat{\bmu}_{m}-\bmu -\pmb{b})
\Rightarrow \mathcal{N}_{K}( \pmb{0},\sigma ^{2}\pmb{\Theta}^{+}).
\end{equation*}
Now if $\bM^{\top}\bX/\sqrt{n}=o(1)$ then $\sqrt{n}\,\pmb{b}= (n\bN^{+}) (\bM^{\top} \bX/\sqrt{n})\pmb{\beta}=O(1)\times o(1)=o(1)$ and consequently $\sqrt{n}(\widehat{\bmu}_m-{\bmu}_t)\Rightarrow \mathcal{N}_{K}( \pmb{0},\sigma ^{2}\pmb{\Theta}^{+})$. This completes the proof of Part 2. 

\medskip

The normal equations, cf. \eqref{eq:18}, can be rewritten as 
\begin{equation} \label{eq:18:modified}
\begin{pmatrix}
    \bN/n & \bM^{\top}\bX/n\\
    \bX^{\top} \bM/n & \bX^{\top}\bX/n
\end{pmatrix}
\pmb{\theta} =
\begin{pmatrix}
    \pmb{S}/n \\ \bX^{\top}\bY/n
\end{pmatrix}.
\end{equation}
Let $\bM_i = (M_{i,1},\ldots,M_{i,n})^{\top}$ and $\bX_i = (X_{i,1},\ldots,X_{i,n})^{\top}$ denote the $i^{th}$ columns of $\bM$ and $\bX$ respectively. Note that $n_i$ of elements of the vector $\bM_i$ are equal to $1$ or $-1$ and rest are zeros. So $\bM_i^{\top}\bX_j/n={n_i}/{n} ({1}/{n_i}\sum_{k=1}^{n_i}M_{i,k} X_{j,k})$. Since the components in $\bX_j$ are IID RVs with mean zero we have $({1}/{n_i})\sum_{k=1}^{n_i} M_{i,k} X_{j,k}\to 0$ in probability as $n\to\infty$.  Consequently $\bM^{\top}\bX/n= o_p(1)$ which plugged into \eqref{eq:18:modified} results in the solutions 
$$\widehat{\bmu}_t = \bN^{+}\bS + ~\pmb{\xi}_{1n} \text{ and }\widehat{\pmb\beta}~= (\bX^{\top}\bX)^{-1}\bX^{\top}\bY + ~ \pmb{\xi}_{2n}$$ 
for some random sequences $\{\pmb{\xi}_{1n}\}$ and $\{\pmb{\xi}_{2n}\}$ converging to $\bzero$ in probability. The fact that $\sqrt{n}\,\pmb{\xi}_{1n}=O_p(1)$ and $\sqrt{n}\,\pmb{\xi}_{2n}=O_p(1)$ follows from Part 2 above. Finally using $\bM^{\top}\bX/n= o_p(1)$ and $\lim_{n\to\infty}\bX^{\top}\bX/n=\pmb{\Sigma}_{\bX}$ in probability gives $\pmb{\Sigma}= \mathrm{BlockDiag}(\pmb{\Theta},\,\pmb{\Sigma}_{\bX})$. This completes the proof of Part 3.
\end{proof}

%----------
\subsection*{Proof of Theorem \ref{Thm-LST.LargeK}:}
%---------

For $1\leq i<j\leq K$ we have $n_{ij}=1$ so we drop third subscript in the following proof.

\begin{proof}
By Theorem \ref{Thm-UniqueE} we have that $\widehat{\bmu}=\bN^{+}\bS$. Here $\bN$ is a $K\times K$ matrix with elements
\begin{equation*}
			     \bN=\left( n_{ij}\right) =\left\{ 
				\begin{array}{ccc}
					K-1 & \text{if} & i=j \\ 
					-1 & \text{if} & i\neq j%
				\end{array}
				\right.
\end{equation*}
and $\bS=(S_{1},\ldots ,S_{K})^\top$ where $S_{i}=\sum_{j\neq i}Y_{ij}$ for $i=1,\ldots ,K$. It is easy to verify that 
\begin{equation*}
	\bN^{+}=\frac{1}{K^{2}}\bN=\left( n_{ij}^{+}\right) = \left\{ 
 \begin{array}{ccc}
    \frac{K-1}{K^{2}} & \text{if} & i=j \\ 
    -\frac{1}{K^{2}} & \text{if} & i\neq j%
 \end{array}
    \right. .
\end{equation*}
Thus, 
\begin{equation*}
\widehat{\mu}_i=\frac{K-1}{K^2}S_i-\sum_{j\neq i}\frac{1}{K^2}S_j
 = \frac{K-1}{K^2} \sum_{j\neq i}Y_{ij} - \sum_{j\neq i}\frac{1}{K^2} \sum_{k\neq j}Y_{jk} = \sum_{j\neq i}Y_{ik} (\frac{K-1}{K^2} + \frac{1}{K^2}) = \frac{1}{K}S_i.
\end{equation*}
Further, $\mathbb{V}ar(Y_{ij})=\sigma ^{2}$ for all $i,j$ implies $\mathbb{V}ar(K^{-1}S_i)=(K-1)\sigma ^{2}/K^{2}$ and  $\mathbb{C}ov(K^{-1}S_i,K^{-1}S_j)=-\sigma ^{2}/K^{2}$. It follows that
\begin{equation*} 
\sqrt{K}\{(\widehat{\mu}_{i_{1}},\ldots ,\widehat{\mu}_{i_{k}})^\top-(\mu
_{i_{1}},\ldots ,\mu _{i_{k}})^\top\}\Rightarrow \mathcal{N}_{k}(\pmb{0%
},\sigma ^{2}\pmb{I}_{k}).
\end{equation*}
Thus we obtain \eqref{RR.CLT}. 

Next, note that $\widehat{\mu}_{i}-\mu_{i}= \sum_{j\neq i} \epsilon_{ij}/K$. Since $\epsilon_{ij}$'s are subgaussian random variables with mean zero and parameter $\tau^2$ (where $\tau^2\ge \sigma^2$) the quantity $\widehat{\mu}_{i}-\mu_{i}$ is subgaussian with mean zero and parameter $\tau^2/K$. It follows that for any $t\in\mathbb{R}$
\begin{align}\label{eq:subexp:mgf:bound}
\mathbb{E}(\exp{(t(\widehat{\mu}_{i}-\mu_{i}))}) \leq \exp(\frac{t^2\tau^2}{K}).
\end{align}
Now for any $t>0$ we have 
\begin{align*}
    &\exp{(t\mathbb{E}(\max_{1\leq i\leq K}|\widehat{\mu}_{i}-\mu_{i}|))}
    \leq \mathbb{E}(\exp{(t\max_{1\leq i\leq K}|\widehat{\mu}_{i}-\mu_{i}|)})
    \leq \sum_{i=1}^K \mathbb{E}(\exp{(t|\widehat{\mu}_{i}-\mu_{i}|)})\\
     &\leq \sum_{i=1}^K \mathbb{E}(\exp{(t(\widehat{\mu}_{i}-\mu_{i}))}) + \sum_{i=1}^K \mathbb{E}(\exp{(-t(\widehat{\mu}_{i}-\mu_{i}))})
    \leq 2K \exp(\frac{t^2\tau^2}{2K}) ,
\end{align*}
the first inequality is by Jensen, the second is a property of the maximum function and the last inequality uses equation \eqref{eq:subexp:mgf:bound}. 

Setting $t=\sqrt{2K\log(K)}$ and taking the logarithm in the display above we find that
\begin{align*}
\mathbb{E}(\max_{1\leq i\leq K}|\widehat{\mu}_{i}-\mu_{i}|) \leq 
\frac{\log(2K)+\tau^2 \log(K)}{\sqrt{2K\log(K)}}= O(\sqrt{\frac{{\log (K)}}{K}}).
\end{align*}
for large enough $K$.
\end{proof}

%---------------------
\subsection*{Proof of Theorem \ref{inf:lse:properties}:}
%-------------------
\begin{proof}
{ 
Condition \ref{cond:inf:s1} guarantees that the comparison graph is connected with probability one as $K \rightarrow \infty$ (Erdos and Reyni 1961) so the LSE is well defined. The Theorem is established by showing that: $(a)$ the limit \eqref{random:clt} holds for an estimator $\tilde{\bmu}^{\top}=(S_{1}/N_{11},\ldots, S_{1}/N_{KK})$; and $(b)$ $\tilde{\bmu}$ and $\widehat{\bmu}$ are asymptotically equivalent, i.e., $\sqrt{Kp_{K}}(\widehat{\mu}_i-\tilde{\mu}_i)=o_{p}(1)$ for any $i=1,\ldots,K$. 

Note that for $i\in\{1,2,\ldots,K\}$ the diagonal element $N_{ii}$ of $\bN$ is a Binomial RV with parameters $(K-1,p_K)$ so $N_{ii}\to\infty$ as $K\to\infty$. Furthernote that we can express $S_{i}$ as the random sum $\sum_{j\neq i}Y_{ij}B_{ij}$. Using Lemma 2.9 in Finkelstein et al. (1994), who studied random sums with non--random normings, we have 
\begin{align*}
\sqrt{N_{ii}}(\tilde{\mu}_i - \mu_i)=\sqrt{N_{ii}}(\frac{1}{N_{ii}}S_i - \mu_i) \Rightarrow \mathcal{N}(0,\sigma^2),
\end{align*}
for all $i\in\{1,2,\ldots,K\}$ whenever $K\to\infty$. Furthermore for $1 \le i\neq j \le K$ we have
\begin{align*}
\mathbb{C}ov(\sqrt{N_{ii}}\tilde{\mu}_i,\sqrt{N_{jj}}\tilde{\mu}_j)=&~\mathbb{C}ov(\sqrt{N_{ii}}(\frac{1}{N_{ii}}S_i - \mu_i), \sqrt{N_{jj}}(\frac{1}{N_{jj}}S_j - \mu_j))\\
=&~ \mathbb{E}(\frac{1}{\sqrt{N_{ii}}}(S_i - N_{ii}\mu_i) \frac{1}{\sqrt{N_{jj}}}(S_j - N_{jj}\mu_j))\\
=&~ \mathbb{E}(\frac{1}{\sqrt{N_{ii}}} \frac{1}{\sqrt{N_{jj}}}\mathbb{E}((S_i - N_{ii}\mu_i)(S_j - N_{jj} \mu_j)\,|\, N_{ii},N_{jj}))\\
=&~ \mathbb{E}(\frac{1}{\sqrt{N_{ii}}} \frac{1}{\sqrt{N_{jj}}}(-\sigma^2 \mathbb{I}(B_{ij}=1))) \to 0.
\end{align*}
Combining the previous two statements and using the fact that $N_{ii}/(Kp_K) \to 1$ as $K\to\infty$, we obtain 
\begin{equation} \label{random:clt:tilde:mu}
{\sqrt{Kp_K}}((\tilde{\mu}_{i_1},\ldots,\tilde{\mu}_{i_r})-(\mu_{i_{1}},\ldots ,\mu_{i_{r}}))^\top\Rightarrow \mathcal{N}_{r}\left( \pmb{0},\sigma^2 \pmb{I}\right),
\end{equation}
establishing step $(a)$ of the proof.

Next, we prove step $(b)$. Let $\bN_{\#}$ denote the Laplacian of the complete graph with $K$ items and $\|\cdot\|_s$ the spectral norm of matrix. Using Theorem 1.1 in Meng and Zheng (2010) we have
\begin{align} \label{pertub:bound}
\|K\bN_{\#}^{+}/p_K-K\bN^{+}\|_s \leq \sqrt{2}\, \|K\bN_{\#}^{+}\|_s\, \|K\bN^{+}/p_K\|_s\, \|\bN/K-p_K\bN_{\#}/K\|_s.
\end{align}
Note that the nonzero eigenvalues of $p_K\bN_{\#}/K$ are equal to $p_K$ and those of $\bN/K$ are equal to $p_K + O_p(\sqrt{\log(K)/K})$ (see Juhasz 1991; Jiang 2012). So using Weyl's theorem, 
$\lambda_{\max}(\bN/K-p_K\bN_{\#}/K)= O_p(\sqrt{\log(K)/K})$ and $\lambda_{\min}(\bN/K-p_K\bN_{\#}/K)= O_p(\sqrt{\log(K)/K})$ as $K\to\infty$, and hence $\|\bN/K-p_K\bN_{\#}/K\|= O_p(\sqrt{\log(K)/K})$ as $K\to\infty$. Consequently 
\begin{align} \label{pertub:convergence}
\|K\bN_{\#}^{+}/p_K-K\bN^{+}\|_s \leq O_p(\sqrt{\log(K)/K}),
\end{align}	
as $K\to\infty$. It now follows from \eqref{pertub:convergence} that
$$
K\|\bN_i^{+} - \bN_{\# i}^{+}/p_K\| \leq O_p(\sqrt{\log(K)/K}), \text{ as } K\to\infty,
$$
where $\bN_i^{+}$ and $\bN_{\# i}^{+}$ denote the $i^{th}$ row of  $\bN^{+}$ and $\bN_{\#}^{+}$ respectively. Now using the fact that $K\bN_{\#}^{+} = \bN_{\#}/K$ we find that
\begin{align}\label{mu:si}
\widehat{\mu}_i = \frac{1}{Kp_K}S_i + o_p(\frac{1}{K})= \tilde{\mu}_i+ o_p(\frac{1}{K}) \text{ as } K\to\infty.
\end{align}
Using \eqref{random:clt:tilde:mu} and \eqref{mu:si} we obtain \eqref{random:clt}. Further using arguments similar to those in the proof of Theorem \ref{Thm-LST.LargeK} we obtain \eqref{inf:ER:wlln}.
}
\end{proof}

\end{document}